\documentclass[pdflatex,sn-mathphys-num]{sn-jnl}

\usepackage{graphicx}
\usepackage{multirow}
\usepackage{amsmath,amssymb,amsfonts}
\usepackage{amsthm}
\usepackage{mathrsfs}
\usepackage[title]{appendix}
\usepackage{xcolor}
\usepackage{textcomp}
\usepackage{manyfoot}
\usepackage{booktabs}
\usepackage{algorithm}
\usepackage{algorithmicx}
\usepackage{algpseudocode}
\usepackage{listings}

\usepackage{amscd,latexsym,enumerate}
\usepackage[mathscr]{euscript}
\usepackage{ma thrsfs}
\usepackage{epsfig}
\usepackage{fancybox}
\usepackage{verbatim}
\usepackage{tikz}
\usepackage{tikz-cd}
\usepackage{todonotes}
\usepackage{multicol}
\usepackage{graphicx}

\usepackage{color}

\definecolor{MyBlue}{cmyk}{1,0.13,0,0.63}
\definecolor{MyGreen}{cmyk}{0.91,0,0.88,0.52}
\newcommand{\mylinkcolor}{MyBlue}
\newcommand{\mycitecolor}{MyGreen}
\newcommand{\myurlcolor}{black}

\usepackage{hyperref}
\hypersetup{%
  bookmarksnumbered=true,bookmarksopen=false,%
  plainpages=false,
  linktocpage=true,%
  colorlinks=true,breaklinks=true,%
  linkcolor=\mylinkcolor,citecolor=\mycitecolor,urlcolor=\myurlcolor,%
  pdfpagelayout=OneColumn,%
  pageanchor=true,%
}
\newtheorem{theorem}{Theorem}[section]
\newtheorem{proposition}[theorem]{Proposition}
\newtheorem{lemma}[theorem]{Lemma}
\newtheorem{corollary}[theorem]{Corollary}
\newtheorem{definition}[theorem]{Definition}
\newtheorem{remark}[theorem]{Remark}
\newtheorem{example}[theorem]{Example}
\numberwithin{equation}{section}

\newcommand{\CM}{{\mathbb C}}
\newcommand{\NM}{{\mathbb N}}
\newcommand{\RM}{{\mathbb R}}
\newcommand{\SM}{{\mathbb S}}

\newcommand{\ZM}{{\mathbb Z}}
\newcommand{\PM}{{\mathbb P}}

\newcommand{\KM}{{\mathbb K}}

\newcommand{\Aa}{{\mathcal A}}
\newcommand{\Ee}{{\mathcal E}}
\newcommand{\Pp}{{\mathcal P}}

\newcommand{\BB}{{\bf B}}

\newcommand{\Bb}{{\mathcal B}}
\newcommand{\Dd}{{\mathcal D}}

\newcommand{\Uu}{{\mathcal U}}

\newcommand{\Ss}{{\mathcal S}}

\newcommand{\Tt}{{\mathcal T}}
\newcommand{\Rr}{{\mathcal R}}
\newcommand{\Nn}{{\mathcal N}}
\newcommand{\Mm}{{\mathcal M}}
\newcommand{\Cc}{{\mathcal C}}

\newcommand{\Ii}{{\mathcal I}}
\newcommand{\Ll}{{\mathcal L}}

\newcommand{\Kk}{{\mathcal K}}
\newcommand{\Hh}{{\mathcal H}}

\newcommand{\edgealgebra}{{\Ee}}
\newcommand{\hsedgealgebra}{{\Ee}_+}
\newcommand{\pmedgealgebra}{{\Ee}_\pm}
\newcommand{\hsalgebra}{{\hat{\Aa}_+}}
\newcommand{\neghsalgebra}{{\hat{\Aa}_-}}
\newcommand{\pmhsalgebra}{{\hat{\Aa}_\pm}}
\newcommand{\interfacealgebra}{{\hat{\Aa}_i}}
\newcommand{\bulkalgebra}{{\Aa_b}}

\newcommand{\scrA}{{\mathcal A}}

\providecommand{\abs}[1]{\left \lvert#1 \right \rvert} 
\providecommand{\norm}[1]{\left \lVert#1 \right \rVert}
\newcommand{\one}{{\bf 1}}

\newcommand{\Tr}{\mbox{\rm Tr}}

\newcommand{\SF}{{\rm Sf}}

\newcommand{\Ch}{{\rm Ch}}

\newcommand{\Ker}{{\rm Ker}} 
\newcommand{\Ran}{{\rm Ran}} 
 
\newcommand{\sgn}{{\rm sgn}} 
 
\newcommand{\diag}{{\rm diag}}

\newcommand{\mult}{\mathcal M}

\theoremstyle{thmstyleone}%

\theoremstyle{thmstyletwo}%

\theoremstyle{thmstylethree}%

\begin{document}

\title[Self-adjoint extensions]{The role of self-adjoint extensions in the bulk-edge correspondence}

\date{\today}

\author*[1]{\fnm{Johannes} \sur{Kellendonk}}\email{kellendonk@math.univ-lyon1.fr}

\author*[2]{\fnm{Tom} \sur{Stoiber}}\email{tstoiber@uci.edu}
\affil*[1]{Univerisité de Lyon, Université Claude Bernard Lyon 1, Institute Camille Jordan, CNRS UMR 5208, 69622 Villeurbanne, France}
\affil*[2]{Department of Mathematics, University of California, Irvine, CA, 92717, USA}


\abstract{We investigate the role of self-adjoint extensions in the bulk-edge correspondence for topological insulators. While the correspondence is well understood in discrete models with spectral gaps, complications arise in the presence of unbounded Hamiltonians and varying boundary conditions, leading to anomalous behavior that has recently been dubbed violations of bulk-edge correspondence. In this work we use a K-theoretic framework to identify precise conditions needed for unbounded Hamiltonians to be affiliated to the respective observable algebras and define K-theory classes. In special cases we can then exclude anomalous behaviour and obtain the standard bulk-edge correspondence, or, under weaker conditions, obtain a relative bulk-edge correspondence theorem, which compares pairs of Hamiltonians. Applying that relative approach in the bulk we recover among other things the so-called bulk-difference-interface correspondence for Hamiltonians that fail to define a bulk K-theory class in the conventional way. The second main result is that one can define K-theory classes in terms of von Neumann unitaries, which under changes in boundary conditions directly contribute to the number of protected edge states. This approach clarifies apparent violations of the classical bulk-edge paradigm and provides a systematic account of boundary-induced topological corrections.}




\maketitle

\newpage

\tableofcontents
\newpage
\section{Introduction}
The bulk-edge correspondence is an important paradigm in the theory of topological insulators. In this context, infinitely extended boundary-less insulators (bulk systems) are described by Hamiltonians
which have a gap in their spectrum at the Fermi energy.  Topological invariants can be associated to these gaps. These so-called bulk invariants are stable under homotopy as long as the gap stays open. 
A system with boundary can be seen as a system with an interface with empty space and since empty space is a system with trivial topology the gap of a topologically non-trivial bulk system should ``close at the boundary'' coinciding with the formation of gap-filling edge modes.

Let us make those notions explicit in the best-understood setting of two-dimensional free Fermi models. There, a bulk topological insulator corresponds to a one-particle Hamiltonian $H$ with a gap $\Delta$ in its spectrum and such that the Fermi energy lies in that gap. 
Its topological class is encoded in its Fermi projection onto the spectral subspace below the gap $\Delta$. Adding a constant to the Hamiltonian if necessary, we may assume that the gap $\Delta$ is around zero energy. Then $H$ is boundedly invertible. From the Fermi projection $P_{\leq 0}(H)$ one computes a quantized bulk invariant, the integer quantum Hall conductance $\sigma_b\in \ZM$. Truncating $H$ to a space with boundary will result in a Hamiltonian $\hat{H}$ which may have spectrum in $\Delta$; to detect that one can compute the conductivity $\sigma_e$ of the edge states in $\Delta$. It turns out that the latter is also quantized $\sigma_e \in \ZM$ in proper units.

\begin{definition}
Let $H,\hat{H}$ be a pair of bulk and boundary Hamiltonians such that $\hat{H}$ coincides with $H$ far away from the boundary and $H$ has a spectral gap $\Delta$ around $0$. We say that they satisfy Hatsugai's relation if both sides of 
$$\sigma_b=\sigma_e$$
are well-defined and equal.
\end{definition}

The above relation has been shown by Hatsugai \cite{Hatsugai} using complex analytic methods in the context of periodic tight binding operators (thus with rational magnetic flux) and since then extended with other methods to various degrees of generality \cite{KRS,ElgartGrafSchenker,GrafPorta,BKR,PSbook}. For the continuous-space models considered in this work one can obtain the relation (at least for Quantum Hall systems) using K-theory \cite{KS04,KS04b,BR2018}, as we will describe below. The main alternative approaches are direct proofs of the Hatsugai relation (e.g. \cite{DGR11,CorneanEtAl24,CMS23}) or, in the related setup of smooth domain walls, pseudodifferential methods \cite{Drouot21, Bal19, Bal21}. The main motivation for the present article is that Hatsugai's relation can be violated for certain unbounded Hamiltonians, which means that those must be exceptions not covered in any of these approaches to bulk-edge correspondence.

In the K-theoretic formalism the quantization of $\sigma_b$ and $\sigma_e$ is understood by identifying them with elements of two discrete groups and Hatsugai's relation arises from a homomorphism between them, the boundary map of an exact sequence: The bulk Fermi projection defines a class $[H]_0$ in the $K$-group $K_0(\Aa_b)$ for a suitable $C^*$-algebra $\Aa_b$ of bulk observables and the edge states of $\hat H$ define a class $[\hat{H}]_1$ in the group $K_1(\Ee)$ for an algebra of edge observables. The numerical invariants $\sigma_b$, $\sigma_e$ are obtained as pairings of those classes with specific cyclic cohomology classes. The idea that $H$ and $\hat{H}$ agree away from the boundary can be phrased algebraically by means of an exact sequence 
\begin{equation}
\label{eq:bulk_edge_exact_sequence}
0 \to \Ee \to \hat{\Aa} \stackrel{q}{\to} \Aa_b \to 0
\end{equation}
which identifies $\Aa_b$ with the quotient $\hat{\Aa}/\Ee$. The algebras in this exact sequence will be adapted to the context in this paper, but for now the reader may think of the algebra $\hat{\Aa}$ as an algebra of observables on a half-space and of $\Ee$ as operators located near the boundary; a concrete model will be given below. 

For these constructions to make sense, the Hamiltonians $H$, $\hat{H}$ need to be affiliated to the respective algebras. A self-adjoint operator $T$ on a Hilbert space $\Hh$ is a called a multiplier of a $C^*$-algebra $\Aa\subset \Bb(\Hh)$ if its bounded transform $F(T)=T(1+T^2)^{-\frac12}$ is in the multiplier algebra $\mult(\Aa)$ of $\Aa$ and $(T+\imath)^{-1}\Aa$ is a norm-dense subset of $\Aa$. In the $C^*$-algebraic sense an affiliated operator is instead defined as a regular adjointable operator which maps a dense subset of a $C^*$-algebra $\Aa$ into $\Aa$. An unbounded multiplier $T$ defines an affiliated operator by the obvious map $(T+\imath)^{-1}\Aa\to \Aa$ which makes it then independent of the Hilbert space representation. As we discuss in Section~\ref{sec:aff_bbc} this is a one-to-one correspondence and as such the notions of affiliated operator or $\Aa$-multiplier may be used interchangeably. 

The epimorphim $q:\hat{\Aa}\to \Aa_b$ has a unique extension to unbounded affiliated operators and therefore any self-adjoint operator $\hat{H}$ affiliated to $\hat{\Aa}$ is mapped canonically to an operator $H=q(\hat{H})$ which is affiliated to $\Aa_b$. By definition, $H$ is the bulk operator corresponding to $\hat{H}$ far away from the boundary, which would have been difficult to make precise without $C^*$-algebraic notions. The exact sequence \eqref{eq:bulk_edge_exact_sequence} gives us a natural boundary map $\partial: K_0(\Aa_b)\to K_1(\Ee)$. 
\begin{definition}
\label{def:K_theoretic_bulk_edge_correspondence}
Let $\hat{H}$ be a self-adjoint operator affiliated to $\hat{\Aa}$. Assume that $H=q(\hat{H})$  has a spectral gap around $0$.
The pair $H,\hat{H}$ 
satisfies the K-theoretic bulk-edge correspondence if both sides of 
\begin{equation}
\label{eq:ktheoretic_bulk_edge}
[\hat{H}]_1=\partial([H]_0)
\end{equation}
are well-defined and equal. 
\end{definition}
When pairing the $K$-group elements with suitable Chern cocycles (\ref{eq:ktheoretic_bulk_edge}) gives rise to numerical relations like the Hatsugai relation. Both sides of (\ref{eq:ktheoretic_bulk_edge}) may be ill defined or, even if they are well defined, may disagree. These cases have not been properly addressed in \cite{KS04,KS04b,BR2018}; we attempt here to remedy that.

Let us explain the general principles for models defined by translation invariant differential operators. The involved algebras can be described explicitly as follows.
\begin{enumerate}
    \item The bulk algebra $\Aa_b$ is the $C^*$-algebra generated by integral operators on $L^2(\RM^d)$ with kernels $k\in C(\RM^d\times \RM^d)$ with rapid off-diagonal decay
    $$k(x,y)=O(\abs{x-y}^{-\infty})$$
    satisfying the covariance relation 
    $$k(x,y)= k(0,y-x)$$
   One can identify $\Aa_b\simeq C_0(\RM^d)$ since the Fourier (or Bloch-) transform turns the now translation-invariant integral operators into multiplication operators. 
    \item The edge algebra $\Ee$ (we later denote it also by $\Ee_+$) is the $C^*$-algebra generated by integral operators on $L^2(\RM^{d-1}\times \RM_+)$ with kernels $k\in C((\RM^{d-1}\times \RM_+)\times (\RM^{d-1}\times \RM_+))$ which decay in the direction orthogonal to the boundary and have rapid off-diagonal decay in the parallel directions
    $$k(x', x_d,y', y_d)=O(\abs{x_d}^{-\infty}+ \abs{y_d}^{-\infty} + \abs{x' - y'}^{-\infty})$$
    as well as the covariance
     $$k(x', x_d,y', y_d)= k(0, x_d,y'-x', y_d).$$
    The partial Fourier transform in the first $d-1$ directions gives the isomorphism $\Ee\simeq C_0(\RM^{d-1})\otimes \KM(L^2(\RM_+))$, i.e.\ one has fiberwise compact operators.
    \item The half-space algebra $\hat{\Aa}$ (later also denoted by $\hat{\Aa}_+$) is the $C^*$-algebra consisting of all operators on $L^2(\RM^{d-1}\times \RM_+)$ that can be written in the form 
    $$\hat{a}= P_+ a_b P_+ + e$$
    for some $a_b \in \Aa_b$, $e\in \Ee$ and the orthogonal projection $P_+:L^2(\RM^{d})\to L^2(\RM^{d-1}\times \RM_+)$. We see that, modulo $\Ee$, $\hat a$ is the same as $P_+ a_b P_+$ and the covariance relation implies that the integral kernel of $P_+ a_b P_+$ uniquely determines that of $a_b$ so that the quotient $\hat \Aa/\Ee$ is in fact isomorphic to $\Aa_b$. 
\end{enumerate}  
We will also look at matrix-valued operators, in which case the above algebras are to be enlarged by taking their tensor product with a matrix algebra.
In Section~\ref{ssec:algebras} we will give more general constructions based on twisted crossed product algebras which can also handle external magnetic fields and non-translation-invariant ergodic potentials and therefore disordered models.

Let us focus now on the two-dimensional translation-invariant case, where $\Aa_b\simeq C_0(\RM^2)$ and $\Ee=C_0(\RM)\otimes \KM(L^2(\RM_+))$.
A bulk Hamiltonian $H$ is a matrix-valued translation-invariant self-adjoint differential operator. It corresponds via Fourier transform to a polynomial function $H: \RM^2 \to M_N(\CM)$. The continuous bounded functions of $H$ are $N\times N$-matrices with values in the multiplier algebra $\mult(\Aa)=C_b(\RM^2)$. If $H$ has a spectral gap at energy $0$ 
then the spectral projection $P_{\leq 0}(H)$ of $H$ onto states below the gap  is a continuous function $P:\RM^2 \to M_N(\CM)$. Under certain circumstances (which we assume for now) the formula
\begin{equation}
\label{eq:chern}
\sigma_b :=\frac{1}{2\pi \imath } \int_{\RM^2} \Tr(P(k) [\partial_{k_1} P(k), \partial_{k_2} P(k)]) \mathrm{d}k
\end{equation}
is well-defined and integer-valued defining a topological invariant that is associated to the gap. This invariant, called the Chern number of the projection, is our bulk invariant; it can be realised as a pairing $\sigma_b=\langle\Ch_2,[H]_0\rangle$ between a cyclic 2-cocyle $\Ch_2$ and the $K_0$-class $[H]_0$ and implements an isomorphism $K_0(C_0(\RM^2))\simeq \ZM$. It has the physical interpretation as the Hall conductance if the Fermi energy belongs to the gap \cite{TKNN}.

Any self-adjoint operator $\hat{H}$ affiliated to $\hat{\Aa}$ decomposes as a direct integral of self-adjoint operators $\hat{H}=\int_{\RM}^\oplus \hat{H}_k\, dk$ on $L^2(\RM_+)$. If this is a family of Fredholm operators (the spectrum of $\hat{H}_k$ 
in some interval $\Delta$ around $0$ consists of at most $N<+\infty$ eigenvalues)
and $\hat{H}_k$ is invertible for all large enough $\abs{k}$ then the family is continuous in what is sometimes called the Wahl topology \cite{Wahl} and therefore has a well-defined spectral flow $\SF(k\in \RM\mapsto \hat{H}_k)$, essentially 
the number of eigenvalues passing through $0$, counted with a plus or minus sign if they pass from below or above, as one varies $k$ from $-\infty$ to $+\infty$. One can interpret this integer K-theoretically as the pairing of $[\hat{H}]_1\in K_1(\Ee)\simeq \ZM$ with the  winding number cocycle and verify in this way that it corresponds to the edge conductivity
$$\sigma_e=\langle \Ch_1,[\hat H]_1\rangle = \SF(k\in \RM\mapsto \hat{H}_k)$$

The boundary map $\partial:K_0(\Aa_b)\to K_1(\Ee)$ in this case is an isomorphism of $\ZM$ with $\ZM$ and we have:
\begin{proposition}
For $\Aa_b$ and $\Ee$ as above the K-theoretic bulk-edge correspondence implies Hatsugai's relation in dimension $d=2$.
\end{proposition}

As already mentioned, the K-theoretic bulk-edge correspondence need not hold. However,
there is a simple sufficient condition which resolves the bulk-edge correspondence for realistic quantum Hall systems, which are naturally bounded from below. This is Theorem~\ref{th:bbc_strongly_affiliated} combined with Lemma~\ref{lemma:strong_vs_resolvent} in the main text.
\begin{proposition}
Given an exact sequence \eqref{eq:bulk_edge_exact_sequence}. Let $\hat{H}$ be a self-adjoint operator affiliated to $M_N(\hat{\Aa})$. Assume that  $H=q(\hat{H})$ has a spectral gap around $0$.
If 
\begin{equation}
\label{eq:res_aff}
(\hat{H}+\imath)^{-1}\in M_N(\hat{\Aa})
\end{equation} and $\hat{H}$ is bounded from below then the pair $H,\hat H$ satisfies the K-theoretic bulk-edge correspondence.

More generally, if $\hat{H}$ is not necessarily bounded from below then the K-theoretic bulk-edge correspondence holds under the strong affiliation condition
\begin{equation}
\label{eq:strong_aff}
F(\hat{H})= \hat{H}(1+\hat{H}^2)^{-\frac{1}{2}} \in M_N(\hat{\Aa}^\sim)
\end{equation} 
where $\hat{\Aa}^\sim$ is the conditional unitization.
\end{proposition}
To understand the significance of this condition note that $\hat{\Aa}$ is generally non-unital, hence the usual notion of affiliation only requires that $(H+\imath)^{-1}\in M_N(\mult(\hat{\Aa}))$ is a matrix with entries in the multiplier algebra $\mult(\hat{\Aa})$ which is much larger than $\hat{\Aa}^\sim$. It turns out that, for example, the usual Laplacian with Dirichlet or Neumann boundary conditions is bounded from below and satisfies the resolvent-affiliation condition \eqref{eq:res_aff} as one can check directly by computing its integral kernels. Likewise, the strong affiliation condition \eqref{eq:strong_aff} also holds for certain unbounded model Hamiltonians such as the regularized Dirac operator, but we can only assert this for a very small number of boundary conditions. For more general differential operators it is useful to study in more detail the problem of self-adjoint extensions of affiliated operators. Indeed, the von Neumann theory of self-adjoint extensions generalizes to the $C^*$-algebraic setting without large modifications \cite{WK92}, as we will recall in Section~\ref{sec:sa_extensions}. Based on this we will then show in Section~\ref{sec:aff_interface} how questions of affiliation can be answered constructively for half-space and interface problems.

If $\hat{H}$ is not bounded from below or otherwise not sufficiently regular (especially if it is non-elliptic), then a bulk-edge correspondence in form of an equality between a bulk invariant and an edge invariant which can be expressed as a spectral flow of edge modes might not exist. This has been most prominently observed in Dirac-type models or shallow-water models  \cite{DelplaceEtAl,TDV20, GJT21, TDV19, TT21, GrafTarantola2025, JudTauber2025}. In essence, the spectral flow of the half-space model can be seen to depend on the choice of self-adjoint boundary condition for $\hat H$, even within the relatively nice class of local translation-invariant boundary conditions. Importantly, this does not depend on how the spectral flow of the edge modes is employed to define the edge invariant. We use a definition of edge invariant which is motivated both by physics (the Hatsugai relation) and K-theory, looking at spectral flow through the gap. Other versions of the edge invariant are defined by counting how many edge modes flow out or into the bulk spectrum. In that case there are models for which the difference between the edge invariant and the bulk invariant is given by the winding number of a certain scattering phase at infinite energy \cite{TDV20, GJT21, GrafTarantola2025, JudTauber2025}. While this is a topological identity, it is not fully satisfying, since it is not robust against disorder. Indeed, neither the edge invariant nor the scattering phase of \cite{GJT21} can even be defined for aperiodic models. Moreover, when the above two versions of edge invariant can both be defined, they do not always agree with each other and therefore the scattering approach also does not directly make statements about the validity of the conventional Hatsugai relation.
One objective of this work is to derive an analogous correction term for the edge conductivity which seamlessly fits into K-theory, therefore showing robustness under a large class of perturbations, as long as certain affiliation conditions hold.

Since the Hatsugai relation can be violated, the K-theoretic bulk-edge correspondence \eqref{eq:ktheoretic_bulk_edge} also cannot hold in general. To make matters worse, the K-theory classes on either or both sides of can be ill-defined. Therefore, checking validity of \eqref{eq:ktheoretic_bulk_edge} alone is of too narrow scope, instead we propose a more general formalism which covers the usual bulk-edge correspondence as a special case. We call it relative bulk-edge correspondence, since it is based on comparing pairs of Hamiltonians with well-defined relative K-theoretic invariants.

\begin{definition}\label{def:I-comparable}
Let $\Aa$ be a $C^*$-algebra with closed ideal $\Ii$ and let $H,H'$ be two self-adjoint operators affiliated to $\Aa$.
\begin{itemize}
\item $H$ and $H'$ are $\Ii$-comparable if, for all 
$f\in C([-\infty,+\infty])$, $f(H)-f(H')\in \Ii$.

\item $H$ and $H'$ are locally $\Ii$-comparable if there exists an open interval $\Delta$ containing $0$ such that, for all $f\in C_c(\Delta)$, 
$f(H)-f(H')\in \Ii$ . 
\end{itemize}
\end{definition}

As we will see in Section~\ref{ssec:rel_inv}, if $H,H'$ are $\Ii$-comparable and both have a spectral gap containing $0$ one can then define a relative $K$-theory class
$[H,H']_0 \in K_0(\Ii)$. Likewise, for locally $\Ii$-comparable $H,H'$ one can define a relative class
$[H,H']_1 \in K_1(\Ii)$.

The corner stones of the relative bulk-edge correspondence are then two results. The first is Theorem~\ref{th:rel_bbc} in the main text.

\begin{theorem} Consider the exact sequence \eqref{eq:bulk_edge_exact_sequence}.
Let $\hat{H},\hat{H}'$ be affiliated to $\hat{\Aa}$ and $\hat{\Aa}$-comparable. Then $H=q(\hat{H})$, $H'=q(\hat{H}')$ are $\Aa_b$-comparable. If $H$ and $H'$ have a spectral gap around $0$ then
$\hat{H},\hat{H}'$ are locally $\Ee$-comparable and 
$$[\hat{H},\hat{H}']_1 = \partial([H,H']_0)\;\; \in \;\;K_1(\Ee).$$
\end{theorem}
The comparability condition holds in particular if $\hat{H}$, $\hat{H}'$ are relatively $\hat{\Aa}$-compact perturbations of each other. By the Kato-Rellich theorem $\hat{H}$ and $\hat{H}'$ must have the same domain and the same boundary conditions, hence this theorem mostly allows us to compare models which have the same boundary conditions but are different in the bulk, e.g. differ by addition of a potential.

The complement of this theorem compares two operators which are equal in the bulk, but are supplied with different boundary conditions. This is phrased algebraically by writing them as self-adjoint extensions corresponding to different von Neumann unitaries. In this case the difference in edge invariants can be read off entirely from the von Neumann unitaries, which define K-theory classes in the edge algebra. This is Prop.~\eqref{prop:comparison_saext_ktheory} in the main text.
\begin{theorem}
Let $\mathring{H}$ be a symmetric operator affiliated to $\hat{\Aa}$, then its self-adjoint extensions $H_u$ which are affiliated to $\hat{\Aa}$ are parametrized precisely by those von Neumann unitaries which are in the multiplier algebra $\mult(\hat{\Aa})$. Given two von Neumann unitaries $u,v\in \mult(\hat{\Aa})$ the following conditions are equivalent 
\begin{enumerate}
		\item[(i)] $\Cc(\hat{H}_u)-\Cc(\hat{H}_v) \in \Ee$ with $\Cc$ the Cayley transform.
		\item[(ii)] $(\hat{H}_u+\imath)^{-1}-(\hat{H}_v+\imath)^{-1} \in \Ee$
		\item[(iii)] $u-v \in \Ee$.
	\end{enumerate}
	If either holds one has
	$$[\hat{H}_u, \hat{H}_v]_1 = [\Cc(\hat{H}_u)\Cc(\hat{H}_v)^*]_1  = [1+uv^*-e_+]_1 \in K_1(\Ee)$$
	with the projection $e_+ = uu^*=vv^*$ onto the deficiency subspace $\Ker(\mathring{H}^*+\imath)$.
\end{theorem}
For two-dimensional differential operators the topological contribution of this K-theory class corresponds to a relative winding number that can be used to compare different boundary conditions and which directly manifests itself in a correction to the number of edge modes.

The combination of these two theorems allows us to understand a large part of phenomena which have previously been dubbed somewhat misleadingly violations of bulk-edge correspondence. Indeed, taking into account the topological contribution of the von Neumann unitaries we obtain a satisfying theory of bulk-edge correspondence for local translation-invariant boundary conditions under a natural resolvent-affiliation condition, see Theorem~\ref{th:bulk_edge_pluscorrection}. 

Another powerful application is bulk-interface correspondence, which is algebraically set up very similarly by an exact sequence which relates two asymptotic bulk operators on either side of a boundary hypersurface to an interface Hamiltonian. In that case the K-theory class of the interface states is related to the relative bulk K-theory class of the bulk Hamiltonians. Our main result is Theorem~\ref{th:bulk_interface_pluscorrection} which is a K-theoretic version of Bal's \cite{Bal19} so-called bulk-difference interface correspondence and simultaneously generalizes it to sharp interfaces which have non-trivial matching conditions at a fiducial hypersurface. In that case, the von Neumann unitary describing the boundary condition may once again give a correction to the number of interface modes. This is interesting in particular in view of violations of bulk-interface correspondence observed in certain models \cite{BCSR24, BalYu}.

In Section~\ref{sec:examples} we will finally discuss several two-dimensional examples in detail. After showing how von Neumann unitaries can be computed from boundary triples, we study boundary conditions for the Laplacian for which it has a band of edge modes below the bulk band and discuss the relative bulk-edge and bulk-interface correspondence for (regularized) massive Dirac Hamiltonians. Finally, we briefly discuss the shallow-water model of \cite{DelplaceEtAl, TDV20} in our framework.

\medskip

\noindent
{\bf Acknowledgements:} The authors gratefully acknowledge the hospitality of ETH Zürich, where part of this work was carried out. Some of the material previously appeared in the second author’s doctoral thesis \cite{StThesis}, but has been substantially extended and revised. This work was partially supported by the grants NSF DMS-2052899, DMS-2155211, Simons-681675, as well as the German research foundation Project-ID 521291358.

\section{Affiliation and comparison of unbounded multipliers}

\subsection{Affiliation of unbounded operators}
\label{sec:aff_bbc}
Let $\Aa$ be a $C^*$-algebra. We denote by $\mult(\Aa)$ its multiplier algebra, which can be abstractly defined as the algebra of all bounded linear operators on $\Aa$ which have an adjoint when one considers $\Aa$ as a right Hilbert $\Aa$-module. Affiliated operators, also called unbounded multipliers, are densely defined adjointable operators on $\Aa$, which for technical reasons shall usually be regular operators. There are several distinct ways to characterize this combination of properties. 
\begin{definition}[\cite{W90}]
\label{def:affiliated}
A densely defined linear operator $T$ on $\Aa$ is called affiliated to $\Aa$ if there is $z\in \mult(A)$, $\|z\|\leq 1$ such that the graph of $T$ is given by
$$\{(1-z^*z)^{\frac12} a,za):a\in \Aa\}$$
\end{definition}
An affiliated operator $T$ has an adjoint $T^*$ and is regular, i.e. $1+T^*T$ has a bounded inverse in $\mult(\Aa)$, in fact $z$ is uniquely determined as the bounded transform $z=F(T):=T (1+T^*T)^{-\frac12}$. The domain of $T$ (as an adjointable operator on $\Aa$) is the norm-dense right $\Aa$-module
$$(1+T^*T)^{-\frac12} \Aa = (1-z^*z)^{\frac12} \Aa.$$
Morphisms extend to affiliated operators:
\begin{proposition}
\label{prop:morphism} If $\varphi:\Aa\to\Bb$ is a $*$-homomorphism of $C^*$-algebras then it extends uniquely to the multiplier algebra $\varphi: \mult(\Aa)\to \mult(\Bb)$. If $\varphi$ is surjective and $\Aa,\Bb$ are separable then the extension to multipliers is also surjective \cite{APT}.
	
One also has an extension to affiliated operators \cite[Theorem 1.2]{W90}: If $T$ is an $\Aa$-multiplier and $\varphi(\Aa)\Bb$ is norm-dense in $\Bb$ then there is a canonical $\Bb$-multiplier $\varphi(T)$ with core $\varphi(\mathrm{Dom}_{\Aa}(T))\Bb$.
\end{proposition}
A special morphism is the bounded continuous functional calculus, which therefore takes values in the multiplier algebra:

\begin{proposition}
If $T$ is a self-adjoint  $\Aa$-multiplier then it defines a homomorphism $C_b(\RM)\to \mult(\Aa)$ via the continuous functional calculus of the bounded self-adjoint element $F(T)\in \mult(\Aa)$.

This homomorphism is strictly continuous in the sense that $f_n \to f$ locally uniformly implies $f_n(T)a \to f(T)a, af_n(T)\to af(T)$ in operator norm for each $a\in \Aa$.
\end{proposition}
While the above definitions and results hold in the abstract setting of Hilbert modules we are interested in operators on Hilbert spaces. In our context the $C^*$-algebra $\Aa$ is faithfully and non-degenerately represented on a Hilbert space $\Hh$. We simply write this as $\Aa\subset\Bb(\Hh)$. One can identify $\Bb(\Hh)$ with the multiplier algebra of the compact operators $\mathbb{K}(\Hh)$ in such a way that Proposition~\ref{prop:morphism} applied to the inclusion map guarantees that the image of any affiliated operator corresponds canonically to an operator affiliated to $\mathbb{K}(\Hh)$. The latter are precisely the closed densely defined operators on $\Hh$. Under this correspondence the bounded affiliated operators, i.e. the elements of the multiplier algebra $\Mm(\Aa)$, are mapped bijectively to the set of all operators $T\in \Bb(\Hh)$ for which $Ta $ and $aT$ belong to $\Aa$ for all $a\in \Aa$. 

Characterizing the unbounded affiliated operators is more subtle. It is well-known that a multiplier $z\in \mult(\Aa)$ is the bounded transform $z=F(T)$ of a (unique) regular affiliated operator $T$ if and only $\norm{z}\leq 1$ and $(1-z^*z)^{\frac{1}{2}}\Aa$ is norm-dense \cite[Theorem 10.4]{Lance}. One can rewrite those conditions in terms of $T$ to give a criterion for affiliation: 
\begin{lemma}[\cite{W90}:Example~4] \label{lem-aff}
Let $\Aa\subset \Bb(\Hh)$. The operators affiliated to $\Aa$ are in one-to-one correspondence with the closed operators $T$ on $\Hh$, acting on $\Aa$ by left multiplication, for which $F(T)\in\mult(\Aa)$ and ${(1+T^*T)^{-\frac12} \Aa}$ is norm-dense in $\Aa$. 
\end{lemma}
A closed operator on $\Hh$ which satisfies this condition will be called an $\Aa$-multiplier. If $T$ is an $\Aa$-multiplier then its adjoint $T^*$ is also an $\Aa$-multiplier and $F(T)^*=F(T^*)$ is compatible with the usual Hilbert space adjoint.

Let $T$ be a closed operator on $\Hh$ and $\Aa\subset \Bb(\Hh)$. 
We may always define its maximal domain as an operator on $\Aa$, namely the set $D_\Aa(T)\subset \Aa$ consisting of all $a\in \Aa$ for which $\mathrm{Dom}(T)\subset \Ran(a)$ and $Ta \in \Aa$. 

If $T$ is an $\Aa$-multiplier then $\Dd_\Aa(T)=(1+T^*T)^{-\frac12}\Aa$ is exactly the domain of $T$ as an affiliated operator. 

We can now reformulate the density condition of Lemma~\ref{lem-aff}:
\begin{lemma}
\label{lemma:concrete_multiplier}
Let $T$ be a  closed densely defined operator on $\Hh$. Let $\Aa\subset\Hh$ and suppose that $F(T)\in\mult(\Aa)$. The following properties are equivalent:
\begin{enumerate}
\item[(i)] $(1+T^*T)^{-\frac12}\Aa$ is dense in $\Aa$.
\item[(ii)] $(1+T^*T)^{-1}\Aa$ is dense in $\Aa$.
\item[(iii)] $\Dd_\Aa(1+T^*T)$ is dense in $\Aa$.
\end{enumerate}
\end{lemma}
\begin{proof}
First we note that if $F(T)\in \mult(\Aa)$ then $(1+T^*T)^{-1}\Aa\subset \Aa$ as $(1+T^*T)^{-\frac12}=(1-F(T^*)F(T))^{\frac12} \in \mult(\Aa)$. Then $(1+T^*T)^{-\frac12}\Aa$ is dense in $\Aa$ if and only if $(1+T^*T)^{-1}\Aa$ is dense in $\Aa$. The norm-density of $(1+T^*T)^{-1}\Aa$ and $(1+T^*T)^{-\frac12}\Aa$ in $\Aa$ is equivalent, since on the one hand $(1+T^*T)^{-\frac12}\Aa\subset (1+T^*T)^{-1}\Aa$ and on the other $(1+T^*T)^{-1}\Aa= (1+T^*T)^{-\frac12}(1+T^*T)^{-\frac12}\Aa$ is the image of a dense subset under a bounded operator with dense range. This shows $(i)\iff (ii)$.
			
The implication $(i)\implies(iii)$ holds since  $F(T)\in \mult(\Aa)$ implies that $(1+T^*T)^{-\frac12}\Aa\subset \Dd_\Aa(1+T^*T)$.
 
For the implication $(iii)\implies(ii)$, since $1+T^*T$ is a self-adjoint operator on $\Hh$ with a bounded inverse, one can write
$$a = (1+T^*T)^{-1} (1+T^*T) a$$
for each $a\in \Dd_\Aa(1+T^*T)$, thus $(1+T^*T)^{-1}\Aa$ contains the norm-dense set $\Dd_\Aa(1+T^*T)$. 
\end{proof}

The reformulation (iii) of the density condition is convenient if one wants to show that a concrete operator is an $\Aa$-multiplier since one needs to consider neither inverses nor fractional powers. For example, if $T$ is a differential operator then $1+T^*T$ is again a differential operator which may map a class of smoothing elements of $\Aa$ to itself; in contrast the inverse $(1+T^*T)^{-1}$ can be rather inaccessible. 

\begin{lemma}
	\label{lemma:selfadjointmultiplier}
	Let $H$ be a closed densely defined self-adjoint operator on $\Hh$ and $\Aa\subset\Bb(\Hh)$. The following are equivalent:
	\begin{enumerate}
		\item [(i)] $H$ is affiliated to $\Aa$.
		\item [(ii)] There exist dense subsets $\Aa_\pm \subset \Aa$ such that $$(H\pm \imath)^{-1}\Aa_\pm \subset \Aa$$ is dense.
		\item[(iii)] $(H+\imath)^{-1}\in \mult(\Aa)$ and $\Dd_\Aa(H)$ is dense in $\Aa$.
	\end{enumerate}
\end{lemma}
\begin{proof} (i) $\implies$ (ii): Suppose $H$ (and hence also $H\pm\imath$) is affiliated to $\Aa$. Then $(1+H^2)^{-1}\Aa$ is dense in $\Aa$ and $(H\pm\imath)^{-1}\Aa$ is dense in $\Aa$. 
Set $\Aa_\pm = (H\pm\imath)(1+H^2)^{-1}\Aa = (H\mp\imath)^{-1}\Aa$. Then  $\Aa_\pm$ and $(H\pm\imath)^{-1}\Aa_\pm$ are dense in $\Aa$.  

For (ii) $\implies$ (i) we first note that we can replace $\Aa_\pm$ by $\Aa$ due to density, hence the condition implies $(H\pm\imath)^{-1}\in \mult(\Aa)$ and with that $f(H)\in \mult(\Aa)$ for any $f\in C_0(\RM)$. Let $f_n \in C_0(\RM)$ be a sequence of functions that converges on any compact set uniformly to $F$, then $$f_n(H)(H+\imath)^{-1}a \to F(H)(H+\imath)^{-1}a$$ converges in operator norm for each $a\in \Aa$. Hence $f_n(H) \tilde{a} \to F(H)\tilde{a} \in \Aa$ for each $\tilde{a}\in (H+\imath)^{-1}\Aa$ and density implies that this actually holds for all $\tilde{a}\in \Aa$. Taking adjoints a similar argument shows $a f_n(H)\to a F(H)\in \Aa$ for all $a\in \Aa$, hence $F(H)\in \mult(\Aa)$.

For $(ii)\implies (iii)$ it remains to show that $\Dd_\Aa(H)$ is dense in $\Aa$.
Let $a = (H+\imath)^{-1}c$ with $c\in \Aa$. Then $Ha = c+ia\in\Aa$. Hence $a\in  \Dd_\Aa(H)$. showing that $\Dd_\Aa(H)$ contains $\Aa_-$ which is dense. 
\end{proof}

As hinted in the introduction, we will for the purposes of $K$-theory want to single out multipliers whose functional calculus takes values in a smaller algebra: 

\begin{definition}
\label{def:B_affiliation}
Let $\Bb$ be a closed unital subalgebra of $\mult(\Aa)$ containing $\Aa$. We say that an operator $T$ which is affiliated to $\Aa$ is even $\Bb$-affiliated if its bounded transform $F(T)$ lies in $\Bb$ (as opposed to merely in $\mult(\Aa)$).
\end{definition}

Note that a self-adjoint multiplier $T$ is $\Bb$-affiliated if and only if  $f(T)\in \Bb$ for all $f\in C([-\infty,+\infty])$, and this is the case if $\Theta(T)\in \Bb$ for some strictly monotonous continuous function $\Theta:\RM\to\RM$ which has finite limits at $\pm\infty$.

\begin{definition}\label{def:B_resolvent_affiliated}
Let $\Bb$ be a closed subalgebra of $\mult(\Aa)$. We say that an operator $T$ which is affiliated to $\Aa$ is $\Bb$-resolvent-affiliated if
$$(1+T^*T)^{-1}, (1+T^*T)^{-1} \in \Bb.$$
In the self-adjoint case this is equivalent to
$$(T+\imath)^{-1} \in \Bb$$
or to $f(T)\in \Bb$ for all $f\in C_0(\RM)$. 
	
If $\Bb=\Aa$ we simply say that $T$ is resolvent-affiliated.
\end{definition}
In practice, $\Bb$-resolvent-affiliation is often vastly easier to verify than $\Bb$-affiliation, but it is strictly weaker.

\begin{lemma}
\label{lemma:strong_vs_resolvent}
Let $T$ be a self-adjoint $\Aa$-multiplier and $\Bb\subset \mult(\Aa)$ a unital algebra containing $\Aa$. Then 
\begin{enumerate}
	\item[(i)] If $T$ is $\Bb$-affiliated then it is $\Bb$-resolvent-affiliated.
	\item[(ii)] If $T$ is $\Bb$-resolvent-affiliated and bounded from below then $T$ is $\Bb$-affiliated.
\end{enumerate} 
\end{lemma}
\begin{proof}
This follows easily via the characterization in terms of functional calculus. If $T$ is bounded from below then $F(T) = 1 + f(T) $ for some  $f\in C_0(\RM)$. 
\end{proof}

We will also want to study stability of $\Bb$-affiliation under perturbations:

\begin{proposition}\label{prop-res-perturb}
	Let $H$ be a self-adjoint $\Aa$-multiplier, $\Aa\subset \Bb(\Hh)$, and $V$ a closed symmetric operator on $\Hh$ whose domain contains $\mathrm{Dom}(H)$.
\begin{enumerate}
	\item[(i)]	If $V$ satisfies the Kato-Rellich bound $\norm{V\psi}\leq \alpha\norm{H\psi}+\beta\norm{\psi}$ with some $\alpha<1$ for all $\psi\in \mathrm{Dom}(H)$ and 
	$$V (H+\imath)^{-1}\in \mult(\Aa)$$
	then $H+V$ is also a self-adjoint $\Aa$-multiplier on the same domain as $H$.
	\item[(ii)] Let $\Ii$ be a norm-closed ideal of $\Aa$. If
	$$V (H+\imath)^{-1}\in \Ii$$
	then (i) holds and further
	$$F(H+V)-F(H) \in \Ii.$$
\end{enumerate}
\end{proposition}
\noindent{\bf Proof.} (i) By the Kato-Rellich theorem (which is more generally also true in the Hilbert module context \cite{KL0}) the sum is a self-adjoint operator with the same domain as $H$ and for every $\lambda \neq 0$ one has
\begin{align*}
(H\pm  \lambda \imath)^{-1}-(H+V\pm \lambda \imath)^{-1}&=(H+V\pm \lambda\imath)^{-1} V (H\pm \lambda\imath)^{-1}\\
&= (H\pm  \lambda \imath)^{-1} V (H+V\pm  \lambda \imath)^{-1}.
\end{align*}
Iterating this resolvent identity for large enough $\lambda$ gives a norm-convergent series expansion from which one can conclude that the right-hand side is in $\mult(\Aa)$ and also $(H+V\pm \lambda \imath)^{-1}(\Aa) \subset (H\pm\lambda\imath)^{-1}(\Aa)$.  Lemma~\ref{lemma:selfadjointmultiplier}(ii) therefore shows that $H+V$ is an $\Aa$-multiplier if $H$ is an $\Aa$-multiplier.

(ii) For bounded $V\in \mult(\Aa)$ it is easy to show that the difference of bounded transforms lies in $\Ii$, e.g. by \cite[Lemma 2.7(i)]{CP98} which gives a norm-convergent integral formula for the difference of bounded transforms with an integrand that lies in $\Ii$ under the given condition.

For the case of relatively bounded $V$ note first that the Kato-Rellich bound holds with $\alpha < 1$ since 
$$\norm{V\psi} =\norm{V(H+\lambda\imath)^{-1}(H+\lambda\imath)\psi} \leq \norm{V(H+\lambda\imath)^{-1}}(\norm{H\psi}+\lambda \norm{\psi})$$
for any $\lambda\neq 0$ and one has operator-norm-convergence
$$V(H+\lambda \imath)^{-1}=V(H+ \imath)^{-1}(H+ \imath)(H+\lambda \imath)^{-1}\stackrel{\lambda\to \infty}{\to} 0$$
due to strict convergence of  the functional calculus $(H+ \imath)(H+\lambda \imath)^{-1}$ to $0$ in $\mult(\Aa)$ (see Proposition~\ref{prop:morphism}).

One can then reduce to the bounded case using continuity of the bounded transforms w.r.t. to graph norms \cite{Lesch05}. The idea is to approximate $V$ by the bounded operators $V_n=\phi_n(H)V\phi_n(H) \in \Ii$ for $\phi_n \in C_0(\RM)$ functions that converge locally uniformly to $1$. As in the proof of \cite[Proposition A.8]{VandenDungen2019} one can show that $\lim_{n\to\infty} F(H+V_n)=F(H+V)$ in operator norm, which concludes the proof.

\hfill $\Box$

For $\Ii=\Aa$ we will call perturbations $V$ like this relatively ($\Aa$-)compact, by virtue of $\Aa$ being the compact operators on $\Aa$ in the Hilbert module sense.

\subsection{Comparison of self-adjoint operators}
 In Section~\ref{sec:ktheory} we will study $K$-theory classes defined by pairs of projections or unitaries that come from comparable self-adjoint multipliers. 

\begin{definition}
A {\sl normalizing function} is a smooth non-decreasing function $\Theta:\RM\to \RM$ with $\Theta(0)=0$, $\Theta'(0)>0$ and $\Theta^2-1 \in C_0(\RM)$. If the support of the derivative $\Theta'$ is contained in the interior of an interval $\Delta$ then we say that $\Theta$ is a $\Delta$-normalizing function.
\end{definition}

The function $F(x) = x(1+x^2)^{-\frac12}$ is a normalizing function but it is not $\Delta$-normalizing for a proper sub-interval $\Delta$ of $\RM$.

Recall from Def.~\ref{def:I-comparable} that, given a closed ideal $\Ii$ of $\Aa$, two self-adjoint $\Aa$-multipliers $H,H'$ are $\Ii$-comparable if 
for all $f\in C([-\infty,+\infty])$ we have $f(H)-f(H')\in \Ii$. This condition is equivalent to $\Theta(H)-\Theta(H') \in \Ii$ for some normalizing function $\Theta$. Note that, by Proposition~\ref{prop-res-perturb}, if $H'=H+V$ is a relatively $\Aa$-compact perturbation then $H$ and $H'$ are $\Aa$-comparable.

Another notion that will be important for $K$-theory is being invertible modulo an ideal. As usual, an unbounded $\Aa$-multiplier $H$ is called invertible if it has a bounded inverse, which is then in $\mult(\Aa)$. Therefore, an invertible operator has a spectral gap around $0$.
\begin{definition}
\label{def:invertible_modulo}
Let $\Ii \subset \Aa$ be a closed ideal and $H$ a self-adjoint $\Aa$-multiplier. 
We say that $H$ is \emph{mod $\Ii$ invertible} 
if there is an open interval $\Delta\subset\RM$ containing $0$ such that 
$f(H)\in \Ii$ for each $f\in C_c(\Delta)$.
\end{definition}

A more general notion applying to pairs of operators is as follows:
\begin{definition}
	\label{def:locally_comparable}
	We say that two self-adjoint $\Aa$-multipliers $H$, $H'$ are locally $\Ii$-comparable (in $\Delta$) for $\Ii$ (a clossed ideal of $\Aa$) if there exists an open interval $\Delta\subset \RM$ around $0$ such that $f(H)-f(H')\in \Ii$ for each $f\in C_c(\Delta)$.
\end{definition}
The term ``locally" refers to being comparable only in a small spectral interval around $0$. 
\begin{lemma}
Let $\Ii \subset \Aa$ be a closed ideal and $H$ a self-adjoint $\Aa$-multiplier. Then the following are equivalent:
\begin{enumerate}
\item[(i)]
$H$ is invertible modulo $\Ii$ for the spectral interval $\Delta$.
\item[(ii)] One has $1-\Theta^2(H)\in \Ii$ for every $\Delta$-normalizing function.
\item[(iii)] There exists a normalizing function such that $1-\Theta^2(H)\in \Ii$.
\item[(iv)] There exists a normalizing function such that $\Theta(H) + \Ii$ is invertible in $\mult(\Aa)/\Ii$.
\item[(v)] There exists some positive self-adjoint $V\in \Ii$ such that
$$H^2 + V > c \one$$
is bounded from below by a positive constant $c>0$.
\end{enumerate}
\end{lemma} 
\begin{proof} 
$(i)\implies (ii)$ as $x\mapsto 1-\Theta^2(x) \in C_0(\Delta)$ and $C_c(\Delta)$ is dense in $C_0(\Delta)$.

$(ii)\implies(iii)$ evident.

$(iii)\implies(iv)$ since, under (iii),  $\Theta(H)+\Ii$ is its own inverse. 

$(iv)\implies(v)$:
Suppose that $\Theta(H)+I$ has a bounded inverse in the quotient $\mult(\Aa)/\Ii$, then there exists an open interval $\Delta'$ containing $0$ such that $f(\Theta(H)) \in \Ii$ for all $f\in C_c(\Delta')$. For any positive continuous function $f$ supported in $\Delta'$ with $f(0)>0$ the expression $H^2+f(\Theta(H))$ has a strictly positive lower bound due to the spectral mapping theorem.
 
 $(v)\implies (i)$ holds since it implies $g(H^2)=g(H^2)-g(H^2+V) \in \Ii $ for all $g\in C_c([-c,c])$ (c.f. \cite[Proposition 5]{SSt1}). Hence $f(H)\in\Ii$ where 
 $f = g \circ Q$ with $Q(x) = x^2$. This shows that $f(H)\in\Ii$ for all $f\in C_c([-\sqrt{c},\sqrt{c}])$.
\end{proof}
Invertibility modulo an ideal
is stable under small perturbations (the following result is standard but slightly difficult to find in the literature):
\begin{lemma}
\label{lemma:pert_stability}
Let $H$ be a self-adjoint $\Aa$-multiplier that is invertible modulo an ideal $\Ii$. Then there exists a constant $C>0$ depending on $H$ such that
$$H+ V$$
is still invertible modulo $\Ii$ for all $V=V^*\in \mult(\Aa)$ with $\norm{V}\leq C$.
\end{lemma}
\noindent{\bf Proof.}
Without loss of generality one can assume $H$ is bounded, otherwise one replaces it by its bounded transform, which replaces the perturbation $V$ by a perturbation $V' =F(H+V)-F(H) \in \mult(\Aa)$ which can be expressed by a norm-continuous convergent integral formula, i.e. there exists a universal constant such that $\norm{V'}\leq c \norm{V}$.

For bounded $H\in \mult(\Aa)$ invertibility modulo $\Ii$ is equivalent to the image $\pi(H)\in \mult(\Aa)/\Ii$ having a spectral gap $(-\epsilon,\epsilon)$ around $0$. Hence this property is stable under addition of perturbations with $\norm{V + \Ii} < 2\epsilon$.  
\hfill $\Box$

The local version is more difficult to characterize, except for the trivial observation:
\begin{proposition}
Let $H$, $H'$ be self-adjoint $\Aa$-multipliers. If $H$ is invertible module $\Ii$ then $H$ and $H'$ are locally $\Ii$-comparable if and only if $H'$ is also invertible modulo $\Ii$.
\end{proposition}

The most important sufficient condition for local $\Ii$-comparability is, however, in terms of resolvents:
\begin{definition}\label{def:resolvent-comparable}
We say that two self-adjoint $\Aa$-multipliers $H$,$H'$ are $\Ii$-resolvent-comparable if one of the two equivalent conditions holds:
\begin{enumerate}
	\item[(i)] $(H+\imath)^{-1} - (H'+\imath)^{-1} \in \Ii$.
	\item[(ii)] $f(H)-f(H')\in \Ii$ for all $f\in C_0(\RM)$.
\end{enumerate}
\end{definition}
The equivalence of $(i)$ and $(ii)$ holds since $\Ii$ is an ideal and the $*$-algebra generated by the resolvent is dense in $C_0(\RM)$. 
\begin{proposition}
\label{prop:normalizing_homotopy}
Let $H$,$H'$ be two self-adjoint $\Aa$-multipliers which are $\Ii$-resolvent-comparable. Let $\Theta$ be a normalizing function (we don't need that its derivative has support contained in a bounded interval $\Delta$ around $0$, neither that $\Theta(0)=0$). Then 
$$ e^{i\pi\Theta(H)} e^{-i\pi\Theta(H')} - 1 \in \Ii$$
and the homotopy class in $\Ii$ of this element does not depend on the choice of $\Theta$.
\end{proposition} 
\begin{proof}
The function $x \mapsto e^{\imath\pi\Theta(x)}-1$ belongs to $C_0(\RM)$ and hence 
$ e^{\imath\pi\Theta(H)} - e^{\imath\pi\Theta(H')}  \in \Ii$. Given two normalizing functions $\Theta_0,\Theta_1$ their convex combination $\Theta_t = t\Theta_1+(1-t)\Theta_0$ is also a normalizing function and thus $ e^{\imath\pi\Theta_t(H)} e^{-\imath\pi\Theta_t(H')}-1$ is a continuous path in $ \Ii$.
\end{proof}

\subsection{Self-adjoint extensions}
\label{sec:sa_extensions}
To implement and compare different boundary conditions of a differential operator on a manifold with a boundary one usually starts with a symmetric operator and considers its self-adjoint extensions. An important question in our algebraic formalism is to characterize those extensions of the symmetric operator which are affiliated to a relevant algebra. If the symmetric operator is already affiliated then this can be studied by means of von Neumann's theory of self-adjoint extensions adapted to the Hilbert module case (where the Hilbert module is just the $C^*$-algebra itself) \cite{WK92, Lance}. Recall that, if $\mathring{H}$ is a closed symmetric operator on $\Hh$ whose deficiency spaces $\ker(\mathring{H}^*-i)$ and $\ker(\mathring{H}^*+i)$ have equal dimension, then the self-adjoint extensions of $\mathring{H}$ are classified by the partial isometries $u:\Hh\to\Hh$ whose domain projections $u^*u$ coincide with the orthogonal projection $e_-$ onto $\ker(\mathring{H}^*-i)$ while their range projections $u u^*$ coincide with the orthogonal projection $e_+$ onto $\ker(\mathring{H}^*+i)$. 

The corresponding self-adjoint extension $H_u$ is uniquely defined through its Cayley transform which has the form
	$$\Cc(H_u) = (H_u-\imath)(H_u+\imath)^{-1} := \Cc(\mathring{H}) + u$$
	Here $\Cc(\mathring{H}) = (\mathring{H}-\imath)(\mathring{H}+\imath)^{-1}$ is a partial isometry from $\ker(\mathring{H}^*-\imath)^\perp$ to  $\ker(\mathring{H}^*+\imath)^\perp$. The correspondence between such partial isometries $u$ and extensions ${H}_u$ is one-to-one with the domain of the extensions
	$$\mathrm{Dom}({H}_u) = \{\psi + \psi_+ + u\psi_+: \, \psi \in \mathrm{Dom}(\mathring{H}), \psi_+\in \Ker(\mathring{H}^*-\imath)\}.$$
	It is customary to call the partial isometry $u$ associated to a self-adjoint extension the {\it von Neumann unitary}, since $u$ is unitary as an operator from $\ker(\mathring{H}^*-\imath)$ to $\ker(\mathring{H}^*+\imath)$. The difference of von Neumann unitaries coincides up to a factor to the difference of  the resolvents:
\begin{equation}
\label{eq:resolvent_cayley}
(H_u+\imath)^{-1}- (H_v+\imath)^{-1} = -2 \imath (u-v).
\end{equation}
If $\mathring{H}$ is an $\Aa$-multiplier for a $C^*$-algebra $\Aa\subset \Bb(\Hh)$ then one can characterize precisely the extensions which are again $\Aa$-multipliers:
\begin{theorem}[{\cite{WK92, Lance}}]
	\label{theorem:saextensions}
	Let $\Aa\subset \Bb(\Hh)$ be a $C^*$-algebra and $\mathring{H}$ be a symmetric $\Aa$-multiplier. Then the orthogonal projections $e_\pm$ onto the deficiency spaces 
	$\Ker(\mathring{H}^* \pm \imath)$ are in $\mult(\Aa)$. The self-adjoint extension $H_u$ determined by a von Neumann unitary, i.e.\ a partial isometry  $u\in \Bb(\Hh)$ with  $u^*u=e_-$ and $u u^* = e_+$, is an $\Aa$-multiplier if and only if $u\in \mult(\Aa)$.
\end{theorem}

We note the following observation which will be of use when comparing self-adjoint extensions using $K$-theory.
\begin{proposition}
\label{prop:von_neumann_ktheory}
Let $\Aa\subset \Bb(\Hh)$ be a $C^*$-algebra and $\mathring{H}$ a closed symmetric $\Aa$-multiplier on $\Hh$. Let $\Ii$ be an ideal of $\Aa$. Let $u,v$ be von Neumann unitaries for $ \mathring{H}$ such that $u-v\in \Ii$. 
Then \begin{enumerate}
    \item[(i)] $H_u$ and $H_v$ are $\Ii$-resolvent comparable.
    \item[(ii)] for any normalizing function $\Theta$ the unitaries $$ u v^*-e_++1 \sim_h e^{i\pi\Theta(H_u)} e^{-i\pi\Theta(H_v)} $$ are homotopic in $\Uu(\Ii^\sim)$.
\end{enumerate} 
\end{proposition}
\begin{proof}
That $H_u$ and $H_v$ are $\Ii$-resolvent comparable follows directly from (\ref{eq:resolvent_cayley}). Then
$f(H_u)-f(H_v')\in \Ii$ holds for all 
$f\in C(\RM)$ with $\lim_{x\to\pm\infty}f(x) = 1$. By Proposition~\ref{prop:normalizing_homotopy} the homotopy class of the unitary
$e^{i\pi\Theta(H_u)} e^{-i\pi\Theta(H_v)}$ inside $ \Ii^\sim$ does not depend on $\Theta$ for any $\Delta$-normalizing function. Note that $\Delta$ is a completely arbitrary compact interval in the resolvent-comparable case. By a density argument it is therefore easy to see that the unitaries $e^{i\pi\Theta(H_u)} e^{-i\pi\Theta(H_v)}$ are homotopic in $\Ii^\sim$ for any normalizing function $\Theta$ without support condition for $\Theta'$. For the choice $\Theta(x)=\frac{2}{\pi}\arctan(x)$ one has $e^{i\pi\Theta(x)} = -\Cc(x)$ and simple algebra gives $$ e^{i\pi\Theta(H_u)} e^{-i\pi\Theta(H_v)} =\Cc(H_u) \Cc(H_v)^* = u v^*-e_++1.$$
\end{proof}
The formalism above requires that the symmetric operator $\mathring H$ is affiliated to $\Aa$, which may be difficult to verify in practice. It is often easier to check if some specific self-adjoint extension $H_u$ is affiliated, notably if one can compute its resolvent. One can then work backwards:
\begin{proposition}
\label{proposition:symmetric_multipliers}
Let $\mathring{H}$ be a closed symmetric operator on a Hilbert space $\Hh$ with $\Aa\subset \Bb(\Hh)$. Assume that there is some self-adjoint extension $H_u$ of $\mathring{H}$ such that
\begin{enumerate}
	\item[(i)] $H_u$ is an $\Aa$-multiplier.
	\item[(ii)] the von Neumann unitary is a multiplier $u\in \mult(\Aa)$.
	\item[(iii)] $\Dd_\Aa(\mathring{H})$ is norm-dense in $\Aa$.
\end{enumerate} 
Then $\mathring{H}$ is a symmetric $\Aa$-multiplier and consequently any other self-adjoint extension $H_{v}$ of $\mathring{H}$ is a self-adjoint $\Aa$-multiplier if and only if $v\in \mult(\Aa)$.
\end{proposition}
\noindent{\bf Proof.}
The Cayley transform of $H$ is a partial isometry with $$\Cc(\mathring{H})\rvert_{(\mathring{H}+\imath)\Hh}=(\mathring{H}-\imath)(\mathring{H}+\imath)^{-1}$$ and extended with $0$ to the orthogonal complement of $(\mathring{H}+\imath)\Hh$. Note that is well-defined since $\mathring{H}+\imath$ is injective for symmetric $\mathring{H}$. By the von Neumann theory one has
$$\Cc(\mathring{H})=\Cc(\mathring{H}_u)-u$$
hence $\Cc(\mathring{H}) \in \mult(\Aa)$. The Cayley transform determines $\mathring{H}$ uniquely; similarly one has a reconstruction in the context of affiliated operators \cite[Proposition 5.1]{WK92}: If $c\in \mult(\Aa)$ is a partial isometry with the property that $(c-\one)c^*c\Aa\subset \Aa$ is norm-dense then it is the Cayley-transform of a unique symmetric operator $\hat{T}_c$ affiliated to $\Aa$ (this is a priori not an operator on $\Hh$ but defined on a dense submodule of the Hilbert module $\Aa$).

For $c=\Cc(\mathring{H})$ one finds that $c^*c$ is the projection to $\Ran(\mathring{H}+\imath)$, which is a closed subspace of $\Hh$ since $\mathring{H}$ is closed. On the range of this projection one has for every $\xi\in \mathrm{Dom}(\mathring{H})$
\begin{align*}
(c-\one)c^*c (\mathring{H}+\imath)\xi&= (c-\one) (\mathring{H}+\imath)\xi \\
&= (\mathring{H}-\imath)\xi - (\mathring{H}+\imath)\xi = -2\imath \xi.
\end{align*}
By definition, one has for every $a\in \Dd_\Aa(\mathring{H})$ that $\Ran(a)\subset \mathrm{Dom}(\mathring{H})$ and $(\mathring{H}+\imath)a\in \Aa$, thus we can write 
$$a = \frac{1}{-2\imath } (c-\one)c^*c (\mathring{H}+\imath)a$$
showing that $(c-\one)c^*c \Aa$ contains the norm-dense subset $\Dd_\Aa(\mathring{H})$. 

Therefore there is an affiliated operator $\hat{T}_c$ densely defined on $\Aa$ with Cayley transform $c$. That implies $F(\hat{T}_c)\in \mult(\Aa)$ using the functional calculus in the Hilbert module sense. Being an affiliated operator, $\hat{T}_c$ has a concrete realization as a densely defined closed operator $T_c$ on $\Hh$ which is an $\Aa$-multiplier and whose bounded transform in the Hilbert space sense satisfies $F(\hat{T}_c)=F(T_c)$. One can show \cite{WK92} that a closed operator $T$ is symmetric if and only if its bounded transform $z=F(T)$ satisfies
$$z^*\sqrt{1-z^*z}=\sqrt{1-z^*z}z$$
and can then write its Cayley transform purely algebraically as
$$\Cc(T)=(z-\imath \sqrt{\one - z^*z})(z^*-\imath \sqrt{\one - z^*z}).$$
Due to the representation property, $T_c$ is therefore symmetric with $\Cc(T_c)=c$. Since the Cayley transform determines an operator uniquely, $\mathring{H}$ must coincide with the Hilbert space realization $\mathring{H}=T_c$ of $\hat{T}_c$ which shows that it is an $\Aa$-multiplier. 

We can now see that any self-adjoint extension of $\mathring{H}_v$ which is an $\Aa$-multiplier and whose von Neumann unitary $v$ is in $\mult(\Aa)$ defines a self-adjoint extension of $\hat{T}_c$ in the Hilbert module sense. The converse is provided by Theorem~\ref{theorem:saextensions}, showing a one-to-one correspondence.
\hfill $\Box$

\section{K-theory}
\label{sec:ktheory}

In the most common formulation of operator $K$-theory, the group $K_0(\Aa)$ for a unital $C^*$-algebra $\Aa$ is defined as the Grothendieck group 
of the semigroup of stable homotopy classes of projections in $M_n(\Aa)$ with addition induced by the direct sum. Here $M_n(\Aa)$ is included into $M_{n+m}(\Aa)$ as the upper left block in a block diagonal matrix which has the zero matrix $0_m$ in the lower right block. Similarily $K_1(\Aa)$ is the group of homotopy classes of unitaries in $M_n(\Aa)$. Now 
$U_n(\Aa)$ is included into $U_{n+m}(\Aa)$ as an upper left block in a block diagonal matrix which has the identity matrix $\one_m$ in the lower right block.
Homotopy classes of unitaries form a group under the operation of direct sum and the direct sum of two unitaries is stably homotopic to their product. 
Any unital homomorphism $\phi:\Aa\to \Bb$ between two unital $C^*$-algebras induces a group homomorphism $\phi_*:K_i(\Aa)\to K_i(\Bb)$ by $\phi_*([x]_i) = [\phi(x)]_i$. 

If $\Aa$ is non-unital its $K$-groups are defined via the homomorphism $s_*:K_i(\Aa^\sim)\to K_i(\CM)$ induced by the scalar map 
$s:\Aa^\sim\to \CM$, $s(a,\lambda) = \lambda$, namely
$K_i(\Aa):=\ker s_*$. Every class in $K_0(\Aa)$ can be represented by a formal difference $[p]_0-[q]_0$ of a pair of projections $p,q\in M_N(\Aa^\sim)$ 
with $s(p) = s(q)$. Every class $[u]_1$ in $K_1(\Aa)$ by a unitary $u\in M_N(\Aa^\sim)$ with $s(u) = \one_N$. In that case we say that we have representatives in the standard picture. 

Below we will use another picture for the elements of the $K$-group of a non-unital algebra. Let $\Bb$ be a unital $C^*$-algebra which contains $\Aa$ as an ideal and consider the $C^*$-algebra of pairs 
$$\PM(\Bb,\Aa):= \{(b_0,b_1) \in \Bb \oplus \Bb: \, b_1-b_2 \in \Aa\}$$
with entry-wise addition and multiplication (which is well-defined as $\Aa$ is an ideal of $\Bb$). 
Then we have a split exact sequence of algebras
$$\Aa \stackrel{\imath}\hookrightarrow \PM(\Bb,\Aa) \stackrel{q}{\twoheadrightarrow} \Bb $$
where $\imath(a) = (a,0)$ and $q(b_1,b_2) = b_2 \in \Bb$ and the split homomorphism $t:\Bb\to \PM(\Bb,\Aa)$ is given by $t(b) = (b,b)$. 
This leads to a split exact sequence of the $K$-groups
\begin{equation}\label{eq:K-pair-splitting}K_i(\Aa) \stackrel{\imath_*}\hookrightarrow K_i(\PM(\Bb,\Aa)) \stackrel{q_*}{\twoheadrightarrow} K_i(\Bb)
\end{equation}
and hence 
\begin{equation}\label{eq:K-alt}
K_i(\Aa)\stackrel{\tilde\imath_*}\simeq K_i(\PM(\Bb,\Aa))/t_* K_i(\Bb)
\end{equation}
where $\tilde\imath_*$ is $\imath_*$ followed by the quotient map.
\begin{definition}
Let $\Aa$ be an ideal in unital $C^*$-algebra $\Bb$. We define for a pair of projections $p,q\in M_N(\Bb)$ and $p-q\in M_N(\Aa)$ the class
$$[p,q]_0 = \tilde{i}_*^{-1}([(p,q)]_0) \in K_0(\Aa)$$
and for a pair of unitaries $u,v\in M_N(\Bb)$ and $u-v\in M_N(\Aa)$ the class
$$[u,v]_0 = \tilde{i}_*^{-1}([(u,v)]_1) \in K_1(\Aa).$$
\end{definition}
As the classes of pairs of the form $(x,y)$ are divided out, we have $[x,y]_i = - [y,x]_i$. For $K_1(\Aa)$ it is easy to convert into standard picture representatives, since $[u,v]_1 = [uv^*,1]_1 = [uv^*]_1$ and $uv^*-1\in M_n(\Aa)$, while finding convenient representatives for $K_0(\Aa)$ could be challenging, hence we will instead try to compute boundary maps and topological invariants directly from pairs of projections.

\begin{remark}
If $K_i(\Bb) = 0$ then $K_i(\Aa)\stackrel{\imath_*}\simeq K_i(\PM(\Bb,\Aa))$, a situation which arises, for instance, if $\Aa$ is stable and we take for $\Bb$ the multiplier algebra $\mult(A)$ \cite[Proposition 12.2.1]{Bla}. If $\Aa$ is not stable, we can stabilise it with the algebra $\KM$ of compact operators on an infinite-dimensional separable Hilbert space and so obtain
$$K_i(\Aa)\simeq K_i(\Aa\otimes \KM)\stackrel{\imath_*}\simeq K_i(\PM(\mult(\Aa\otimes \KM),\Aa\otimes \KM))$$
The first isomorphism is induced by the inclusion $a\mapsto a\otimes e$ where $e$ is a rank one projection of $\KM$.
\end{remark}

\subsection{Boundary maps for strongly affiliated operators}

Recall Definition~\ref{def:B_affiliation} by which a multiplier $H$ is $\Bb$-affiliated if its bounded transform $F(H)$ belongs to the algebra $\Bb\subset \mult(\Aa)$. In the standard picture of $K$-theory classes are represented by elements of $M_N(\Aa^\sim)$, hence the theory automatically includes operators which are matrix-valued. 
\begin{definition}
\label{def:strong_affiliation}
A self-adjoint $M_N(\Aa)$-multiplier $H$ is called strongly affiliated if it is $M_N(\Aa^\sim)$-affiliated.
\end{definition}
The condition $F(H)\in M_N(\Aa^\sim)$ is substantially weaker than
$F(H)\in (M_N(\Aa))^\sim$ which is why we cannot simply replace $\Aa$ by $M_N(\Aa)$ to handle the matrix-valued case.
\begin{example}
If $\Aa=C_0(X)$ for some locally compact topological space  then $\Aa^\sim = C(X^+)$ with the one-point compactification $X^+$. In that case, an $M_N(\Aa)$-multiplier is a continuous $M_N(\CM)$-valued function $H$ with bounded transform $F(H)\in C_b(X,M_N(\CM))$. It is strongly affiliated if and only if $F(H)$ admits a unique limit at the infinite point.
\end{example}

\begin{lemma} Let $H$ be a self-adjoint strongly affiliated $M_N(\Aa)$-multiplier. Suppose that $H$ is invertible, i.e.\ it has a gap $\Delta$ at $0$. Then the spectral projection $P_{\leq 0}(H)$ belongs to $M_N(\Aa^\sim)$ and hence defines a class
$$[H]_0:=[P_{\leq 0}(H)]_0 - [s(P_{\leq }(H))]_0 \in K_0(\Aa).$$

\end{lemma}
\begin{proof}
	As $H$ has a gap around $0$ there is a normalizing function
	$\Theta$ with $\mathrm{supp}(\Theta')\subset \Delta$. Then $$\Theta(H)=1_N-2 P_{\leq 0}(H) \in M_N(\Aa^\sim)$$
	for some $N$. Therefore also $P_{\leq 0}(H) 
	\in M_N(\Aa^\sim)$. 
\end{proof}
Recall Definition~\ref{def:invertible_modulo}) which says that, given a closed ideal $\Ii$ of $\Aa$, an $M_N(\Aa)$-multiplier $H$ is invertible modulo $M_N(\Ii)$ if there is an open interval $\Delta$ of $0$ such that for all $f\in C_c(\Delta)$ we have $f(H)\in M_N(\Ii)$.
\begin{lemma} 	\label{def:gapless_invariant}
	If $H$ is a self-adjoint $M_N(\Aa)$-multiplier which is invertible modulo $M_N(\Ii)$ then, for any  normalizing function $\Theta$ with $\mathrm{supp}(\Theta')\subset \Delta$ we have $-e^{\imath\pi \Theta(H)}-\one \in M_N(\Ii)$ and hence we obtain a $K_1$-class 
	$$[H]_1:=[e^{\pi \imath \Theta(H)}]_1\in K_1(\Ii).$$
This class does not depend on the choice of such normalizing function and is trivial if $H$ has a gap at $E=0$.
\end{lemma}

\begin{proof}
By Proposition~\ref{prop:normalizing_homotopy} the class is well-defined, independent of the normalizing function and defines a $K_1$-class in the standard picture. The overall sign $-1$ does not matter as $-\one$ is homotopic to $\one$. 

If $H$ has a spectral gap around $E\in (-\delta,\delta)$ then one can choose a continuous normalizing function $\Theta$ such that $\Theta(H)=\sgn(H)$, hence $e^{\pi \imath \Theta(H)}=-\one$ showing that we get the $0$-element of $K_1(\Ii)$.
\end{proof}
If $[H]_1$ is non-trivial then it is a topological invariant which is an obstruction to gap-opening with respect to small perturbations.

The mechanism behind the $K$-theoretic bulk-edge correspondence is that, given an exact sequence of $C^*$-algebras 
\begin{equation}\label{eq:SES}
 \Ee \hookrightarrow \hat\Aa \stackrel{q}{\twoheadrightarrow} \Aa 
\end{equation}  
there is a homomorphism $\partial: K_0(\Aa)\to K_1(\Ee)$. It is defined as follows:
If $p$ is a projection in $M_N(\Aa^\sim)$ take a self-adjoint element 
$a\in M_N(\hat\Aa^\sim)$ such that $q(a) = p$ (this is called a lift of $p$). Then $\partial $ maps the $K_0$-class defined by $p$ to the $K_1$-class defined by the unitary $e^{2\pi i a}$. 

Applying this to a self-adjoint Hamiltonian $H$ affiliated to the algebra $\Aa$ we need a lift of the Fermi projection $P_{\leq 0}(H)$. To draw meaningful physical conclusions this lift should also be constructed from the functional calculus of a Hamiltonian $\hat{H}$ affiliated to $\hat{\Aa}$. That Hamiltonian should be a lift of $H$ under $q$, written as $q(\hat H)=H$  (note that this is meaningful since $q$ extends to unbounded multipliers by Proposition~\ref{prop:morphism}. Given any $\Delta$-normalizing function $\Theta$ one then gets $1-2 P_{\leq 0}(H) =\Theta(H)=\Theta(q(\hat{H}))=q(\Theta(\hat{H}))$ which obviously can be used to give a lift of $P_{\leq 0}(H)$. However, a priori $\Theta(\hat H)$ 
belongs only to the multiplier algebra of $M_N(\hat\Aa)$, not to $M_N(\hat\Aa^\sim)$, which means that it cannot generally be used to compute the boundary map in the standard picture. If one can assert that condition, however, then one has the following $K$-theoretic identity.

\begin{theorem}
	\label{th:bbc_strongly_affiliated}
	Let $\hat{H}$ be a self-adjoint strongly affiliated $M_N(\hat{\Aa})$-multiplier and such that $H=q(\hat{H})$ is an invertible $M_N(\Aa)$-multiplier.
	Then $\hat{H}$ is invertible modulo $M_N(\Ee)$, $H$ is a strongly affiliated and 
	$$[\hat{H}]_1 = \partial ([H]_0)$$
	where $\partial: K_0(\Aa) \to K_1(\Ee)$ is the connecting map for the exact sequence (\ref{eq:SES}).
\end{theorem}
\begin{proof}
Apart from the assumptions we needed for the last two lemmata the crucial additional assumption is that also $\hat H$ is strongly affiliated. Indeed, 
as shown above, 
$q(P_{\leq 0}(H))=\frac12(\Theta(\hat{H})+1)$.
As $\hat H$ is strongly affiliated $\frac12(\Theta(\hat{H})+1)$ actually lies in $M_{N}(\hat\Aa^\sim)$ and we can apply that last lemma to see that $\partial [H]_0=[\hat{H}]_1$.\end{proof}

The above theorem forms the core of the standard K-theoretic bulk-boundary correspondence as formulated in \cite{KS04}\footnote{In \cite{KS04} operators of the form $H = L + V$ were considered, where $L$ is the Landau Hamiltonian (magnetic Laplacian) and $V$ a covariant potential. In this case strong affiliation was proved by explicit calculation for Dirichlet boundary conditions.}: $\Aa$ can be taken as an algebra of bulk observables modeling elements of some observable algebra $\hat{\Aa}$ far away from a boundary, whereas the ideal $\Ee$ describes observables that are localized around the boundary. One starts with a bulk Hamiltonian $H$ which is a strongly affiliated $\Aa$-multiplier and can use the theorem to make spectral and K-theoretic statements about any Hamiltonian $\hat H$ that is a strongly affiliated $\hat \Aa$-multiplier such that 
$q(\hat H) = H$. 

\begin{remark}
If $\hat{\Aa}$ is unital then strong affiliation reduces to the condition that $\hat H$ is a bounded operator in $M_N(\hat \Aa)$. This is generally the case in the context of tight-binding models for topological insulators where the Hamiltonian is a bounded operator acting on $\ell^2(\Ll)$ for some Delone set $\Ll$.

In contrast, in continuous space the Hamiltonian is a differential operator and therefore $\hat{\Aa}$ and $\Aa$ must be non-unital (since unital $C^*$-algebras do not admit unbounded affiliated operators \cite[Proposition 1.3]{W90}). 
The conditions of the theorem still apply in some situations (see Section~\ref{sec:applications}), however, strong affiliation may only hold for a very small class of boundary conditions.
\end{remark}
Strong affiliation is perhaps the weakest condition which ensures that all occurring $K$-theory classes can be represented in the standard picture. Our goal is now to explore ways to weaken this condition. The starting point is a simple observation:
\begin{lemma}
\label{lemma:strongly_aff_vs_comparable}
	Let $H$ be a self-adjoint strongly affiliated $M_N(\Aa)$-multiplier. If $H$ is strongly affiliated then it is $M_N(\Aa)$-comparable to a self-adjoint element of $M_N(\CM)$.
\end{lemma}
\begin{proof}
	If $\Aa$ is unital so that $\Aa^\sim=\Aa$ then $H$ is $M_N(\Aa)$-comparable to the $0$ operator. So suppose that $\Aa^\sim=\Aa\oplus \CM\one$. 
	
	Choose a normalizing function $\Theta$ which is constant outside $(-1+\epsilon,1-\epsilon)$ for some $\epsilon>0$. If $H$ is strongly affiliated then $\Theta(H)$ has a scalar part $S := s(\Theta(H))$. It is a finite-dimensional self-adjoint scalar matrix. As $\Theta: [-1,1]\to [-1,1]$ surjective, there exists a preimage $T\in M_N(\CM)$ such that $\Theta(T) =S$. Hence $\Theta(H)-\Theta(T)\in M_N(\Aa)$.
\end{proof}
We will in the following derive more generally a bulk-edge correspondence principle for comparable pairs of self-adjoint operators of which Theorem~\ref{th:bbc_strongly_affiliated} is a special case.

\subsection{Relative invariants} 
\label{ssec:rel_inv}
Let $\Bb$ be a unital $C^*$-algebra containing $\Aa$ as an ideal and consider again an extension in the form \eqref{eq:SES}. Then all elements of $\Bb$ are multipliers of $\Aa$. Define
$$\hat \Bb := q^{-1}(\Bb)\subset \mult(\hat{\Aa}),\quad \hat\Bb_\Ee =  q^{-1}(0)\subset \mult(\hat{\Aa})$$
(we used the extension to multipliers).
We have an induced exact sequence
\begin{equation}
\label{eq:pair_exact_seq}
\PM(\hat \Bb_\Ee,\Ee) \hookrightarrow \PM(\hat \Bb,\hat\Aa)
\stackrel{q}\twoheadrightarrow \PM(\Bb,\Aa).
\end{equation}
\begin{proposition}
	\label{proposition:multiplierconnectingmaps}
Upon identifying $K_0(\Aa)\stackrel{i_*}\simeq K_0(\PM(\Bb,\Aa))/t_*K_0(\Bb)$ and
$K_1(\Ee)\stackrel{i_*}\simeq K_1(\PM(\hat\Bb_\Ee^\sim,\Ee))/t_*K_0(\hat\Bb_\Ee^\sim)$ as above the boundary map becomes
\begin{equation}
\label{eq:boundary_map_pairs}\partial[p_+,p_-]_0 = [e^{-\imath 2\pi \hat{p}_+},e^{-\imath 2\pi \hat{p}_-}]_1.
\end{equation}
\end{proposition}
\begin{proof} 
By naturality of the boundary map, we have
$$ \partial_{\PM} \circ \tilde\imath_* = \tilde\imath_* \circ \partial $$
where $ \partial_{\PM}$ is the boundary map of the exact sequence \eqref{eq:pair_exact_seq} and $i$ denotes the inclusions of $\Aa$ and $\Ee$, resp., in the pair algebras. Using the usual expression for the boundary map one has
$$\partial_{\PM}( [(p_+,p_-)]_0 + t_*(K_0(\Bb))) = ([e^{-\imath 2\pi \hat{p}_+},e^{-\imath 2\pi \hat{p}_-}]_1 + \partial_{\PM}(t_*(K_0(\Bb))).$$
Since $t_*(K_0(\Bb))$ is represented by diagonal elements and the boundary map preserves the diagonal form it is clear that $\partial_{\PM}(t_*(K_0(\Bb)))\subset t_*(K_1(\hat \Bb_\Ee))$, hence \eqref{eq:boundary_map_pairs} is compatible with the quotients.
\end{proof}

With this formulation of $K$-theory in mind we now define $K$-theory classes for pairs of Hamiltonians. We start with the invertible (gapped) the case:

\begin{definition}
	Let $\Aa$ be a (non-unital) $C^*$-algebra and $\Bb$ a unital sub-algebra of $\mult(\Aa)$ containing $\Aa$ as an ideal. Let $H,H'$ be self-adjoint invertible  
	and $M_N(\Aa)$-comparable $M_N(\Aa)$-multipliers whose bounded transforms lie in $M_N(\Bb)$. We define
	$$[H,H']_0 :=[P_{\leq 0}(H), \, P_{\leq 0}(H')]_0 \in K_0(\PM(\Bb,\Aa))/t_*K_0(\Bb).$$
\end{definition}
In this notation $\Aa$ and $\Bb$ are left implicit, but since we specify $H$ as an $M_N(\Aa)$-multiplier and the group on the right hand side does not depend on $\Bb$  up to a canonical isomorphism this will not pose a problem. In any case, if $\Bb$ is not specified then we take $\Bb = \mult(\Aa)$. Note that, if $H, H'$ both are strongly affiliated to $\Aa$ then we may take $\Bb=\Aa^\sim$ and get
$$[H,H']_0 = \tilde\imath_*([H]_0) - \tilde\imath_*([H']_0).$$

For models with boundary we are interested in $K$-theoretic obstructions to gap opening for pairs of Hamiltonians.
\begin{definition}
	\label{def:relative_gapless_invariants}
	Let $\Aa$ be a $C^*$-algebra and $\Ii$ an ideal in $\Aa$. 
	Let $H, H'$ be two self-adjoint $M_N(\Aa)$-multipliers which are locally $\Ii$-comparable and such that their bounded transform belongs to $M_N(\Bb)$. Then $H, H'$ define a class in $K_1(\Ii)$, namely
	$$[H,H']^\Ii_1 = [e^{\pi \imath \Theta(H)}, e^{\pi \imath \Theta(H')}]_1.$$	
Here $\Theta$ is a normalizing function $\Theta$ whose derivative $\Theta'$ is supported in an interval $\Delta$ as required by Definition~\ref{def:locally_comparable}.	
If the ideal $\Ii$ is clear from context we will more simply denote $[H,H']_1=[H,H']^\Ii_1$.
\end{definition}
We note that $[H,H']_1$ can also be written $[e^{\pi \imath \Theta(H)}e^{-\pi \imath \Theta(H')},1]_1$ and then directly viewed as an element of $K_1(\Ii)$ as we explained above. 
In nice cases we have a more concrete description of this relative invariant. Recall that $\Cc(H)=(H-\imath)(H+\imath)^{-1}$ (Cayley transform).
\begin{lemma}
\label{lemma:properteis_relative_gapless}
	Let $H,H'$ be as in Definition~\ref{def:relative_gapless_invariants}. We have
	\begin{enumerate}
		\item[(i)] The class $[H,H']_1$ is well-defined and independent of $\Theta$.
		\item[(ii)] 	If $H$,$H'$ are invertible modulo $\Ii$ then
		$$[H,H']_1 = [H]_1 - [H']_1 .$$
		\item[(iii)]
		If $H$,$H'$ are $\Ii$-resolvent comparable then
		$$[H,H']_1  = [\Cc(H),\Cc(H')]_1= [f(H),\overline{f}(H')]_1$$
		for any function $f\in C(\RM)$ of modulus $1$ which tends to $1$ at $\pm\infty$ and has winding number one. 	\end{enumerate}
\end{lemma}
\begin{proof}
(i) We have $x\mapsto e^{\imath \pi \Theta(x)}-e^{\imath \pi \Theta(x)}\in C_0(\Delta)$ as 
the derivative of $\Theta$ is supported in $\Delta$. Hence $e^{\imath \pi \Theta(H)}-e^{\imath \pi \Theta(H')}\in I$. Therefore $[e^{\imath \pi \Theta(H)},e^{\imath \pi \Theta(H')}]_1$ is well-defined and its independence of the choice of $\Theta$ follows by the same homotopy argument as in Proposition~\ref{prop:normalizing_homotopy}. 

(ii) If $H$ and $H'$ are invertible modulo $\Ii$ then $x\mapsto e^{\imath \pi \Theta(x)}-1$ and $x\mapsto e^{\imath \pi \Theta(x)}-1$ belong to  $C_0(\Delta)$ and the classes $[(e^{2\pi \imath \Theta(H)}, 1)]_1$ and $[(1,e^{2\pi \imath \Theta(H')})]_1$ are well defined. The result follows from 
$$[e^{2\pi \imath \Theta(H)}, e^{2\pi \imath \Theta(H')}]_1 = [e^{2\pi \imath \Theta(H)}, 1]_1 + [1,e^{2\pi \imath \Theta(H')}]_1$$
and anti symmetry $[1,e^{2\pi \imath \Theta(H')}]_1 - [e^{2\pi \imath \Theta(H')},1]_1$.

(iii) If $H$ and $H'$ are $\Ii$-resolvent comparable then
$f(H)-f(H')\in \Ii$ for all 
$f\in C(\RM)$ with $\lim_{x\to\pm\infty}f(x) = 1$. Therefore 
$[e^{\pi \imath \Theta(H)}, e^{\pi \imath \Theta(H')}]_1$ is well-defined even without the condition that the support of the derivative of $\Theta$ lies in $\Delta$. By the same homotopy argument as in (ii), it is independent of the choice of normalizing function. The function entering the Cayley transform is of that form, $\Cc(x) = \frac{x-i}{x+i}=-e^{\pi \imath \Theta(x)}$ for some normalizing function $\Theta$. Moreover, again the same type of homotopy argument shows that we can replace $\Theta$ by any continuous function tending to $\pm 1$ at $\pm \infty$ still obtaining the same homotopy class.  
\end{proof}

\subsection{Relative bulk-boundary correspondence}
\label{sec:rel_bbc}

In this section we formulate the abstract algebraic version of bulk-edge correspondence for comparable pairs of Hamiltonians. It is also based on the short exact sequence (\ref{eq:SES}) describing edge, half-space and bulk of a material. 
\begin{proposition} \label{prop:E}
Consider the short exact sequence (\ref{eq:SES}).
	Let $\hat{H}, \hat{H}'$ be two self-adjoint $\hat{\Aa}$-multipliers which are $\hat{\Aa}$-comparable and for which the self-adjoint $\Aa$-multipliers $q(\hat{H})$ and $q(\hat{H}')$ are invertible.
	
	Then $\hat{H}$, $\hat{H}'$ are locally $\Ee$-comparable and hence the class $[\hat{H},\hat{H}']_1 \in K_1(\Ee)$ is well-defined.
\end{proposition}
\begin{proof}
As $q(\hat{H})$ and $q(\hat{H}')$ are invertible their resolvent sets contain an open interval $\Delta$ around $0$. As $\hat{H}, \hat{H}'$ are $\hat{\Aa}$-comparable there is a normalizing function $\Theta$ such that 
 $\Theta(\hat{H})-\Theta(\hat{H}')\in \hat{\Aa}$. Hence also 
 $f(\hat{H})-f(\hat{H}')\in \hat{\Aa}$ for any function $f$ of the form $f= G\circ \Theta$ with $G\in C_0(\RM)$. Since $\Theta'(0)>0$ any function compactly supported in a small enough interval $\Delta'$ around $0$ can be written in that form. We may suppose $\Delta'\subset \Delta$. Then any $f\in C_c(\Delta')$ satisfies  $$q\left(f(\hat{H})-f(\hat{H}')\right) = q\left(f(\hat{H})\right)-q\left(f(\hat{H}')\right) =0$$ 
which implies that $f(\hat{H})-f(\hat{H}')$ lies in the ideal $\Ee$.
\end{proof}
\begin{theorem}[Relative bulk-edge correspondence]
	\label{th:rel_bbc} Consider the short exact sequence (\ref{eq:SES}).
		Let $\hat{H}, \hat{H}'$ be two self-adjoint $\hat{\Aa}$-multipliers which are $\hat{\Aa}$-comparable and for which the self-adjoint $\Aa$-multipliers $H:=q(\hat{H})$ and $H'=q(\hat{H}')$ are invertible.

	Then
		$$[\hat{H},\hat{H}']_1 = \partial [{H},{H}']_0 .$$
Here, $[{H},{H}']_0$ is an element of $K_0(\PM(\mult(\Aa),\Aa))/t_*(K_0(\mult(\Aa)))\stackrel{\tilde\imath_*^{-1}}\simeq K_0(\Aa))$ and $[\hat{H},\hat{H}']_1$ an element of $K_1(\Ee)$.
\end{theorem}
\begin{proof}
By Propositions~\ref{prop:E} $[\hat{H},\hat{H}']_1$ is a well-defined element of $K_1(\Ee)$. The statement follows therefore from Proposition~\ref{proposition:multiplierconnectingmaps} and the definition of the relative classes.
\end{proof}
The above immediately also applies to matrix valued operators, by just replacing $\Aa$ with $M_N(\Aa)$ etc. which has isomorphic K-groups. In contrast, for the notion of strong affiliation we needed to take the matrix size into account explicitly, as unitization does not commute with tensor products. If the bounded transforms of two matrix valued $M_N(\hat\Aa)$-comparable operators $H$ and $H'$ take values in $M_N(\Bb)$ 
for some smaller unital algebra $\Bb$ containing $\Aa$ 
then $[{H},{H}']_0$ can be represented by an element in $K_0(\PM(\Bb,\Aa))$.
In particular, for strong affiliation $\Bb=\Aa^\sim$ one always has $[H,H']_0=[H]_0-[H']_0$ as a difference of two classes in the standard picture of $K_0(\Aa)$. Note in particular that if we  choose $\hat{H}'$ to be a scalar matrix as in Lemma~\ref{lemma:strongly_aff_vs_comparable} then we recover Theorem~\ref{th:bbc_strongly_affiliated} as a special case of Theorem~\ref{th:rel_bbc}.

\begin{remark}
The condition of $\hat{\Aa}$-comparability is much more general than strong affiliation, but can still be rather subtle since functional calculus with functions from $C([-\infty,+\infty])$ can be difficult to control. The most important sufficient condition is relative compactness, namely the sufficient condition of Proposition~\ref{prop-res-perturb}, which for relatively bounded perturbations $\hat{V}:=\hat{H}-\hat{H}'$ reads
\begin{equation}
	\label{eq:rel_compactness}
	\hat{V}(\hat{H}+\imath)^{-1} \in \hat{\Aa}.
\end{equation} 
Note that this condition can only hold if $\hat{H}$ and $\hat{H}'$ have the same domain, whereas generally self-adjoint operators can be $\hat{\Aa}$-comparable without having equal domains. 
\end{remark}

\subsection{K-theoretic comparison of self-adjoint extensions} We now consider the relative invariants $[\hat{H},\hat{H}']_1$ defined by self-adjoint extensions $\hat{H}$, $\hat{H}'$ of a common symmetric operator. As will be seen in the examples, $\hat{H}$ and $\hat{H}'$ can have different edge invariants, i.e.\ a non-trivial relative class $[\hat{H},\hat{H}']_1$ while being asymptotically equal, i.e.\ $q(\hat{H})=q(\hat{H'})$. They must then fail the conditions of Theorem~\ref{th:rel_bbc}. It will usually be the $\hat{\Aa}$-comparability which fails when we compare different self-adjoint extensions.

\begin{remark}
In general $(\hat{H}+\imath)^{-1}-(\hat{H}'+\imath)^{-1} \in M_N(\Aa^\sim)$ does not imply $F(\hat{H})-F(\hat{H}') \in M_N(\Aa^\sim)$ since the difference of the bounded transforms $F(\hat{H})-F(\hat{H}')$ cannot generally be expanded as a norm-convergent integral formula in terms of resolvent differences (as can also be seen by the fact that the Riesz-topology is finer than the gap-topology for unbounded operators, see e.g. \cite{Lesch05}). Indeed, we will see many examples of $\hat{\Aa}$-resolvent comparable self-adjoint extensions which are not $\hat{\Aa}$-comparable.
\end{remark}

If the relative class $[\hat{H},\hat{H}']_1$ comes from different self-adjoint extensions of a single symmetric operator it can be connected more transparently to the data which define the self-adjoint extension:
\begin{proposition}
	\label{prop:comparison_saext_ktheory}
	Let $\mathring{H}$ be a symmetric $\Ee$-multiplier and let $\hat{H}_u$, $\hat{H}_v$ be self-adjoint extensions determined by partial isometries $u,v \in \mult(\Ee)$. Then the following conditions are equivalent
	\begin{enumerate}
		\item[(i)] $\Cc(\hat{H}_u)-\Cc(\hat{H}_v) \in \Ee$
		\item[(ii)] $(\hat{H}_u+\imath)^{-1}-(\hat{H}_v+\imath)^{-1} \in \Ee$
		\item[(iii)] $u-v \in \Ee$.
	\end{enumerate}
	If either holds one has
	$$[\hat{H}_u, \hat{H}_v]_1 = [1+uv^*-e_+]_1 \in K_1(\Ee)$$
	with the projection $e_+ = uu^*=vv^*$ onto the deficiency subspace $\Ker(\mathring{H}^*+\imath)$.
\end{proposition}
\noindent{\bf Proof.} All conditions are equivalent to $\hat{H}_u$ and $\hat{H}_v$ being $\Ee$-resolvent comparable and then the equivalence follows from  Lemma~\ref{lemma:properteis_relative_gapless} and Proposition~\ref{prop:von_neumann_ktheory}.
\hfill $\Box$

If one can compute deficiency subspaces and the corresponding von Neumann unitaries of self-adjoint extensions then this gives a rather explicit formula for the relative invariant.

In practice one often would like to simplify those computations by adding or dropping some symmetric relatively bounded terms $\mathring{H}' = \mathring{H}+V$, which are known not to affect the domains of their self-adjoint extensions by the Kato-Rellich theorem. 
\begin{proposition}
	\label{prop:extensions_perturbation_ktheory_rel_bounded}
	Let $\mathring{H}$ be a symmetric $\Ee$-multiplier with two self-adjoint extensions $\mathring{H}_{u},\mathring{H}_{v}$ corresponding to von Neumann unitaries $u,v\in \mult(\Ee)$ with $u-v\in \Ee$. If $V$ is a symmetric $\Ee$-multiplier on $\mathrm{Dom}(\mathring{H}_{u})\cup \mathrm{Dom}(\mathring{H}_{v})$ whose restrictions to $\mathrm{Dom}(\mathring{H}_{u})$ and $\mathrm{Dom}(\mathring{H}_{v})$ both satisfy the conditions of the Kato-Rellich theorem for $\mathring{H}_{u}$ and $\mathring{H}_{v}$ respectively, then the self-adjoint multipliers $\mathring{H}_{u}+V$ and $\mathring{H}_{v}+V$ are $\Ee$-resolvent comparable and
	$$[\Cc(\mathring{H}_{u})\Cc(\mathring{H}_{v})^*]_1 =[\Cc(\mathring{H}_{u}+V)\Cc(\mathring{H}_{v}+V)^*]_1 \quad \in K_1(\Ee).$$
\end{proposition}
\noindent{\bf Proof.}
The Kato-Rellich bound means that for $\mu\in \RM$ large enough one has
$$\norm{V (\mathring{H}_{u}+\imath \mu)^{-1}}< 1, \quad \norm{V (\mathring{H}_{v}+\imath \mu)^{-1}}< 1.$$

Thus one can expand into a norm-convergent sum
$$(\mathring{H}_u + tV+\imath \mu )^{-1} = (\mathring{H}_{u}+\imath \mu )^{-1} \sum_{n=0}^\infty  t^n (-V(\mathring{H}_{u}+\imath \mu )^{-1})^n$$
and similarly for $\mathring{H}_v+tV$. The term-wise differences
$$(\mathring{H}_{u}+\imath \mu )^{-1} (-V(\mathring{H}_{u}+\imath \mu )^{-1})^n-(\mathring{H}_{v}+\imath \mu )^{-1} (-V(\mathring{H}_{v}+\imath \mu )^{-1})^n$$
lie in $\Ee$ since $(\mathring{H}_{u}+\imath \mu )^{-1}-(\mathring{H}_{v}+\imath \mu )^{-1}\in \Ee$ by assumption and all other factors are bounded $\Ee$-multipliers. The usual resolvent identities then allow us to analytically continue to obtain
$$(\mathring{H}_u + tV+\imath)^{-1} -(\mathring{H}_v + tV+\imath)^{-1} \in \Ee$$
which implies the same for the difference of the Cayley transforms.

To prove that the $K$-theory classes coincide it is then sufficient to note that $t \in [0,1] \mapsto \mathring{H}_{u}+tV, \mathring{H}_{v}+tV$ leads to norm-resolvent continuous paths of $\Ee$-resolvent comparable operators, hence
$$[\hat{H}_u, \hat{H}_v]_1 = [\Cc(\hat{H}_u), \Cc(\hat{H}_v)]_1 = [\Cc(\hat{H}_{u}+V), \Cc(\hat{H}_{v}+V)]_1$$
by norm-continuous homotopy in $\PM(M(\Ee), \Ee)$.
\hfill $\Box$

In particular, for bounded $V=V^*\in M(\Ee)$ the extensions of $\mathring{H}'$ and $\mathring{H}$ are in one-to-one correspondence:
\begin{proposition}
	\label{prop:extensions_perturbation_ktheory}
	Let $\mathring{H}$ be a symmetric $\Ee$-multiplier and $\mathring{H}' = \mathring{H}+V$ with $V$ a bounded self-adjoint $\Ee$-multiplier. If $\mathring{H}_u$ is a self-adjoint extension of $\mathring{H}$ (as an $\Ee$-multiplier) then there exists a unique extension $(\mathring{H}')_{\tilde{u}}$ with the same domain. If two extensions of $\mathring{H}_{u},\mathring{H}_{v}$ satisfy $u-v\in \Ee$ then the corresponding extensions satisfy $\tilde{u}-\tilde{v}\in \Ee$ and
	\begin{equation}
	\label{eq:extensions_perturbation_ktheory}
	\begin{split}
[\Cc(\mathring{H}_{u})\Cc(\mathring{H}_{v})^*]_1 &= [1+uv^*-uu^*]_1 \\ &=[1+\tilde{u}\tilde{v}^*-\tilde{u}\tilde{u}^*]_1 = [\Cc(\mathring{H}'_{\tilde{u}})\Cc(\mathring{H}'_{\tilde{v}})^*]_1 \in K_1(\Ee).
	\end{split}
	\end{equation}
\end{proposition}
\noindent{\bf Proof.}
Since $\mathring{H}_{u}+V$ is a self-adjoint extension of $\mathring{H}'$ there must be a corresponding von Neumann unitary $\tilde{u}$ and by exchanging the roles of $\mathring{H}$ and $\mathring{H}'$ this is a one-to-one correspondence. By the previous Lemma, given two such pairs of extensions $(u,v)$ respectively $(\tilde{u},\tilde{v})$ one has
$$[\Cc(\mathring{H}_{u})\Cc(\mathring{H}_{v})^*]_1 = [\Cc(\mathring{H}_{u}+V)\Cc(\mathring{H}_{v}+V)^*]_1 = [\Cc(\mathring{H}'_{\tilde{u}})\Cc(\mathring{H}'_{\tilde{v}})^*]_1,
$$
which completes the proof upon rewriting the K-theory classes in terms of the von Neumann unitaries.
\hfill $\Box$

Note that the version for unbounded perturbations $V$ cannot hold in the same generality, since there can be extensions $\mathring{H}_u$ on whose domain $V$ is not even symmetric, hence one clearly cannot always define a self-adjoint extension of $\mathring{H}'$ via $\mathring{H}_u+V$.

\section{Applications to differential operators}\label{ssec:algebras}
In the introduction we gave a working model for the algebras and the exact sequence they form which works for translational invariant differential operators. In this simple situation it was possible to directly define the algebras as norm-closed subalgebras of the bounded operators on $L^2(\RM^d)$ or $L^2(\RM^{d-1}\times \RM_+)$. For aperiodic models which, 
by general principles, are studied through covariant families of operators, the situation is more complicated and the standard approach is to study them with the help of crossed product $C^*$-algebras.  

We recall the definition of a twisted crossed product of a twisted $C^*$-dynamical system $(\Cc,\RM^d,\alpha,\gamma)$.
Here $\Cc$ is a $C^*$-algebra which carries a strongly continuous $\RM^d$-action $\alpha$ and $\gamma: \RM^d \times \RM^d \to S^1$ a continuous $2$-cocycle. 

Let $\rho$ be a representation of  $\Cc$ on some Hilbert space $\Hh$. It gives rise to a covariant representation $(\pi_\rho,U)$ of the pair $(\Cc, \RM^d)$ 
on $L^2(\RM^d, \Hh)$:
	\begin{equation}
		\label{eq:regular_rep}
		\begin{split}
			(\pi_\rho(c)\psi)(x) &= \rho(\alpha_{-x}(c))\psi(x),\quad c\in \Cc\\
			(U(t)\psi)(x) &= \gamma(-t,x) \psi(x-t),\quad t\in\RM^d.
		\end{split}
	\end{equation}
If $\rho$ is faithful and non-degenerate then the twisted crossed product
$\Cc\rtimes_{\alpha,\gamma}\RM^d$ is isomorphic to the norm closed sub-algebra of $	\Bb(L^2(\RM^d, \Hh))$ generated be elements of the form
	\begin{equation}
		\label{eq:fourier_generators_crossedproduct}
	\tilde\pi_\rho(f) :=	\int_{\RM^d} \pi_\rho (f(t)) U(t) d t
	\end{equation}
where $f\in L^1(\RM^d,\Cc)$. Equivalently, $\tilde\pi_\rho$ can be understood as a faithful representation of the Banach algebra $L^1(\RM^d,\Cc)$ on $L^2(\RM^d,\Hh)$ equipped with the $\gamma$-twisted convolution product and the twisted crossed product $C^*$-algebra is its completion.

Identifying the twisted crossed product with such a faithful representation $\tilde \pi$ (we suppress the dependence on $\rho$) we have another description of its elements.
	Denote by $p_1,...,p_d$ the unbounded self-adjoint generators of the one-parameter groups of unitaries $\lambda \in \RM^d \mapsto U(\lambda e_i)$, $i=1,...,d$. Then the elements of the form
	\begin{equation}
		\label{eq:linear_span_crossedproduct}
		\sum_{j=1}^N \tilde\pi(c_j) f_{j,1}(p_1)\cdots f_{j,d}(p_d)
	\end{equation}
	with $c_j\in \Cc$ and $f_{j,i}\in C^\infty_c(\RM)$ yield a norm-dense subset of the crossed product algebra. This follows from the density of $\tilde{\pi}(L^1(\RM^d,\Cc))$ as the pre-image of operators of this form under \eqref{eq:fourier_generators_crossedproduct} contains a total subset of $L^1(\RM^d,\Cc)$. Note, however, that the form of \eqref{eq:linear_span_crossedproduct} is not preserved under products or adjoints, hence they form has a dense linear subspace but not a subalgebra.

\subsection{Bulk algebras for differential operators on \texorpdfstring{$\RM^d$}{R\textasciicircum d}}
\label{sec:crossedproducts}
In this section we investigate how covariant families of magnetic Schr\"odinger operators on $L^2(\RM^d)$ 
are affiliated to twisted crossed product algebras. The reader should have in mind operators of the form
$$H = D_A + V$$ where $D_A$ is a magnetic elliptic differential operator, i.e.\  a non-commutative polynomial in the magnetic derivatives $\nabla_A$, whose components we write as $\nabla_{i,A} := \partial_i + \imath A_i$, and $V$ is a covariant potential. Here 
$A$ is the vector potential (in the symmetric gauge) 
$$A_i(x) = -\frac{1}{2}\sum_{k=1}^d \BB_{ik} x_k$$
of the magnetic field which is defined by a constant antisymmetric matrix $\BB\in M_d(\RM)$.
In applications one is also interested in the case where $D$ and $V$ are matrix valued, but this is a simple modification, essentially obtained by tensoring an internal Hilbert space $\CM^N$ onto $L^2(\RM^d)$ and $M_N(\CM)$ onto the algebra.

The magnetic derivatives generate a unitary projective representation 
\begin{equation}
	\label{eq:magtrans}
	t\in \RM^d \mapsto U_A(t):=e^{-\nabla_A\cdot t}= e^{-\frac{\imath}{2}\langle t, \BB X\rangle} U_0(t)
\end{equation}
where $X$ is the vector-valued position operator on $L^2(\RM^d)$ and $t\in \RM^d\mapsto U_0(t)$ are the usual translations $\psi\mapsto \psi(\cdot- t)$. They satisfy the twisted multiplication rule
$$U_A(t+s) = e^{\frac{\imath}{2} \langle t, \BB s\rangle} U_A(s)U_A(t), \qquad \forall t,s\in \RM^d.$$
Algebraically this is encoded by the twisting $2$-cocycle
\begin{equation}
	\label{eq:cocycle}
	\gamma: \RM^d \times \RM^d \to S^1, \quad \gamma(t,s)=e^{\frac{\imath}{2} \langle t, \BB s\rangle}.
\end{equation}	
Other gauge choices for $A$ are also possible, e.g. the Landau gauge, and correspond to cohomologous cocycles, thus isomorphic crossed product algebras.

To describe covariant potentials one needs to fix a locally compact Hausdorff space $\Omega$, which physically encodes the space of configurations, and which carries a continuous action of $\RM^d$, denoted $\omega\mapsto t\triangleright \omega$, and an invariant locally finite Borel measure $\PM$ with full support. There is an induced strongly continuous action 
\begin{equation}
	\label{eq:action}\alpha: \RM^d\times C_0(\Omega) \to C_0(\Omega), \quad \alpha_t(f)(\omega) = f((-t)\triangleright \omega)
\end{equation}
which extends to a strictly continuous action on the multiplier algebra $C_b(\Omega)$ by the same formula. 
A family of potentials $(V_\omega)_{\omega\in \Omega} \in C_b(\RM^d)^{\times \Omega}$ is called a covariant family if
$$V_\omega(x) = f(x \triangleright \omega)$$
for some function $f\in C_b(\Omega)$. Then $H=(H_\omega)_{\omega\in \Omega}$, $H_\omega=D_A+V_\omega$ is an example of a covariant family of magnetic Schr\"odinger operators.

Consider the evaluation representation $\mathrm{ev}_\omega:C_0(\Omega)\to\CM$ at $\omega \in \Omega$, 
$\mathrm{ev}_\omega(f)=f(\omega)$. This representation induces the covariant representation on $L^2(\RM^d)$ 
given by
\begin{equation}\label{eq-cov-rep}
	\begin{split}
		(\pi_\omega(f)\psi)(x) &= f(x\triangleright\omega) \psi(x)\\
		(U_A(t)\psi)(x) &= \gamma(-t,x) \psi(x-t)
	\end{split}
\end{equation}
where we abbreviated $\pi_{\mathrm{ev}_\omega}$ as $\pi_\omega$. 

In general, the associated integrated representation $\tilde\pi_\omega$ (as in \ref{eq:fourier_generators_crossedproduct}) is not faithful and so we have to consider the direct sum representation 
$\rho=\int_\Omega^\oplus \mathrm{ev}_\omega \mathrm{d}\PM(\omega)$. 

The corresponding integral representation $\tilde \pi= \int_\Omega^\oplus \tilde\pi_\omega \mathrm{d}\PM(\omega)$ (we suppress the dependence on $\rho$) is faithful and so that we can identify $C_0(\Omega)\rtimes_{\alpha,\gamma}\RM^d$ with the closure of the algebra generated by $\tilde \pi (f)$ with $f\in L^1(\RM^d, C_0(\Omega))$.
In practice, its elements can often be represented by covariant integral kernels.  We call a measurable function $k: \Omega \times \RM^d \times \RM^d \to \CM$ covariant if there is a measurable function $f: \Omega \times \RM^d\to \CM$ such that almost surely (in all variables) 
\begin{equation}\label{eq:kernel_functions}
k(\omega, x,y) = k_f(\omega, x,y) := e^{-\imath \langle x,\BB y\rangle} f(x\triangleright\omega, x-y).
\end{equation}	
If $f\in L^1(\RM^d, C_0(\Omega))$ then, by construction, 
	the associated integral operator
	$$(\mathrm{op}_\omega(k_f) \psi)(y) = \int_{\RM^d} k_f(\omega, x,y) \psi(y) \mathrm{d}x$$
	$\PM$-almost surely satisfies $\mathrm{op}_\omega(k_f) = \tilde\pi_\omega(f)$.

In particular, integral kernels affiliated to the crossed product can be spotted by the covariance property
$$k(\omega, x+z,y+z)= k(z\triangleright\omega, x,y) e^{-\frac{\imath}{2} \langle z|\BB|x-y\rangle}$$
combined with fast enough off-diagonal decay.

We study under which condition a covariant family of magnetic Schr\"odinger operators is affiliated to the twisted crossed product $C_0(\Omega)\rtimes_{\alpha,\gamma}\RM^d$. 
First we look for affiliation and then we investigate which operators are even resolvent or strongly affiliated.
 
With the usual multi-index notation for differentiable operators, if $I=(i_1,\cdots,i_d)\in\NM^d$ we define
$$\nabla_A^I := {\nabla_{1,A}}^{i_1}\cdots {\nabla_{d,A}}^{i_d}$$
We denote by $C_c^\infty(\Omega)\subset C(\Omega)$ the functions which are compactly supported and norm-smooth w.r.t.\ the action $\alpha$ of $\RM^d$.
\begin{proposition} 
\label{prop:aff_diff_op}
Consider a family $H = (H_\omega)_{\omega\in\Omega}$ of essentially self-adjoint differential operators
	\begin{equation}
		\label{eq:diff_polynomial}
		H_\omega = \sum_{I\in \NM^d} \pi_\omega(c_I) \nabla_A^I
	\end{equation}
	with $c_I\in C^\infty(\Omega)$ 
on $D_0:=C_c^\infty(\RM^d, C_c^\infty(\Omega)) \subset L^2(\Omega, \PM) \otimes L^2(\RM^d) $ (i.e.\ $H$ a direct integral over $\Omega$ of essentially self-adjoint operators on $L^2(\RM^d)$).	
	The closure of $H$ is a self-adjoint $\Aa$-multiplier if and only if $(H+\imath)^{-1}\in \mult(\Aa)$.
\end{proposition}
\begin{proof}
By Lemma~\ref{lemma:selfadjointmultiplier}(iii) we need to show that $\Dd_\Aa(H)$ is dense in $\Aa$.
For $f\in D_0$ the "quantization map"
$$Q: D_0 \mapsto \Aa, \quad Q_f := \int_{\RM^d} e^{t\cdot \nabla_A}\pi(f(t_1,...,t_d)) \mathrm{d}t$$
is well-defined and has norm-dense range in $\Aa$.
We claim that $\Dd_\Aa(H)$ contains $Q(D_0)$ and hence is dense in $\Aa$.

Each $Q_f$ preserves $D_0\subset \Hh:=L^2(\Omega)\otimes L^2(\RM^d)$ and $HD_0 \subset D_0$. Using the commutation relations of the $e^{t\cdot \nabla_A}$ one computes that
$$H Q_f \psi = Q_{Hf}\psi$$
for all $\psi \in D_0$, i.e. left multiplication of $H$ acts on the symbol exactly as it does under the inclusion $D_0\subset \Hh$. Hence $Q(D_0)\subset \Dd_\Aa(H)$.
\end{proof}
The last result easily extends to matrix valued operators, showing that these are affiliated to $M_N(\Aa)$ under the analogous condition. 

The next question is for which differential operators the resolvent lies in $M_N(\Aa^\sim)$. In the translation-invariant case, that is when $\Omega=\{*\}$ is one point and the magnetic field $0$, this question is simple to study using the Fourier transform: Any differential operator transforms into a matrix-valued function $\RM^d\to M_N(\CM)$ and one can always compute the resolvent using linear algebra, though the resulting expressions may be difficult to analyze by hand already for moderately sized matrices. 

In the presence of a magnetic field and for non-trivial coefficients one has the following result which applies to the most common spaces $\Omega$:

\begin{theorem}
\label{th:resolvent_aff_elliptic}
Assume that $\Omega$ is compact and contains an element $\omega\in \Omega$ with dense orbit under the $\RM^d$-action $\alpha$. The associated representation $\pi_\omega$ is faithful and allows us to identify $C(\Omega)$ with a subalgebra $\Cc\subset C_{\mathrm{buc}}(\RM^d)$ of the bounded uniformly continuous functions by setting $\tilde{f}(x)=(\alpha_x f)(\omega)$ for $f\in C(\Omega)$. 

Let $H = \sum_{I\in \NM^d} \pi_\omega(c_I) \nabla_A^I$, $c_I\in C(\Omega)$, be a self-adjoint operator which is the magnetic Weyl quantization of a smooth elliptic polynomial symbol with coefficients in $\Cc$ in the sense of \cite{LeinEtAl10}, i.e. the magnetic quantization of a real function of the form $\Pp(x,\xi) = \sum_{n\in \NM^d} \tilde{f}_n(x) \xi_1^{n_1} .... \xi_d^{n_d}$ with smooth coefficients $\tilde{f}_n\in \Cc$. Then the resolvents of $H$ lie in $\pi_\omega(C(\Omega) \rtimes_{\alpha,\gamma} \RM^d)$, hence $H$ defines a resolvent-affiliated multiplier of $C(\Omega) \rtimes_{\alpha,\gamma} \RM^d$.
\end{theorem}
\begin{proof}
If $\omega\in \Omega$ has a dense orbit then $f\in C(\Omega)\mapsto (\phi_\omega f)(x)=(\alpha_x f)(\omega)$ is an $\RM^d$-equivariant isomorphism onto a closed translation-invariant subalgebra $\mathcal{C}\subset C_{\mathrm{buc}}(\RM^d)$. This is exactly the setting for the pseudodifferential calculus developed in \cite{LeinEtAl10}. The algebra generated by the quantizations of negative-order symbols with coefficients in $\Cc$ can be identified with the twisted crossed product algebra $\Cc\rtimes_{\alpha,\gamma} \RM^d$ and by \cite[Theorem 3.19]{LeinEtAl10} that algebra contains the resolvents of elliptic differential operators with coefficients in $\Cc$. We conclude $(H+\imath)^{-1}\in \Cc\rtimes_{\alpha,\gamma}\RM^d\simeq \pi_\omega(C(\Omega) \rtimes_{\alpha,\gamma} \RM^d)$ and since $\pi_\omega$ is faithful we can conclude resolvent-affiliation (noting that the argument of Proposition~\ref{prop:aff_diff_op} already shows dense norm-dense domain as an unbounded multiplier).  
\end{proof}

Different types of additive or multiplicative perturbations can also be more simply incorporated perturbatively. In particular, as already mentioned above, if $H_0$ is an elliptic differential operator and $V=V^*$ a matrix over $C_b(\Omega)$ then
$$H=H_0 + V$$
is again a multiplier by Proposition~\ref{prop-res-perturb}. If more strongly $(H_0+\imath)^{-1} \in M_N(\Aa)$ and $\Omega$ is compact then $(H+V+\imath)^{-1} \in M_N(\Aa)$. If only $(H+\imath)^{-1} \in M_N(\Aa^\sim)$ (which can happen if the symbol of the differential operator is not elliptic, see Example~\ref{sec-shallow})
then whether $(H+V+\imath)^{-1} \in M_N(\Aa^\sim)$ holds depends on the shape of the matrix $V$ and the scalar part of the resolvent.

\subsection{Algebras for halfspace and interface problems}
\label{sec:halfspaces_and_interfaces}

In this section we fix an algebra $\bulkalgebra:=  C(\Omega)\rtimes_{\alpha,\gamma} \RM^d$ based on a compact space $\Omega$ modeling bulk observables. 
A simple approach to model (families of) half-space Hamiltonians or interface Hamiltonians that asymptotically look like ones affiliated to the bulk algebra is to augment the space $\Omega$ by the two-point compactification of the real line $[-\infty, \infty]$ whose elements describe the position of the boundary or interface. 
We thus set $\Omega^* := \Omega\times [-\infty, \infty]$. Since $\Omega$ is compact, we have $C(\Omega^*)\simeq C(\Omega)\otimes C([-\infty,\infty])$.  The space $\Omega^*$ is a compactification of $\Omega^0:= \Omega \times \RM$ for which $C_0(\Omega^0)=C(\Omega)\otimes C_0(\RM)$. The augmented spaces $\Omega^*$, $\Omega^0$ carry an $\RM$ action
$$x\triangleright(r,\omega) = (x\triangleright\omega,r-x_d)$$
where $\pm\infty - x_d = \pm\infty$. It gives rise to a strongly continuous action on $C(\Omega^*)$ which we denote by $\alpha^*$. 

Since the evalutations $q_\pm: C(\Omega^*)\to C(\Omega)$ at $\pm \infty$ are $\RM^d$-equivariant they extend to the crossed product and one gets an exact sequence
\begin{equation}\label{eq-SEShalf}
C_0(\Omega^0)\rtimes_{\alpha^*,\gamma}\RM^d \hookrightarrow C(\Omega^*)\rtimes_{\alpha^*,\gamma}\RM^d \stackrel{q_-\oplus q_+}{\twoheadrightarrow} C(\Omega)\rtimes_{\alpha,\gamma}\RM^d\,\oplus \,C(\Omega)\rtimes_{\alpha,\gamma}\RM^d
\end{equation}
We denote the ideal by $\edgealgebra$, the extension by $\interfacealgebra$ and the quotient by $\bulkalgebra \oplus \bulkalgebra$.
Since $\Omega^*$ is compact Proposition~\ref{prop:aff_diff_op} and Theorem~\ref{th:resolvent_aff_elliptic} imply that $\interfacealgebra$ (or rather its representations on $L^2(\RM^d)$) contains the resolvents of magnetic differential operators which are perturbed by covariant potentials over $\Omega^*$. 
Now, covariance implies that these potentials depend on the distance to the hypersurface at $x_d=0$, which is supposed to mimic the interface, and to converge at $\pm \infty$ to potentials in the multiplier algebra of the bulk algebra. We therefore call  $\interfacealgebra$ the interface algebra. 
Note that the limiting potentials are covariant w.r.t.\ $\Omega$ only. The kernel $\edgealgebra$ of the evaluation map must finally be generated by operators that decay with the distance from the hypersurface.

By restricting the algebras to the kernel of either $q_-$ or $q_+$ one obtains the two exact sequences 
\begin{equation}\label{eq-half-WH}
C_0(\Omega^0)\rtimes_{\alpha^*,\gamma}\RM^d \hookrightarrow C_0(\Omega^\pm)\rtimes_{\alpha^*,\gamma}\RM^d \stackrel{q_\pm}{\twoheadrightarrow} C(\Omega)\rtimes_{\alpha,\gamma}\RM^d
\end{equation}
with $\Omega^+=\Omega\times (-\infty,\infty]$ and $\Omega^-= \Omega\times [-\infty,\infty)$ respectively. 
This approach of defining algebras associated to the half-spaces $\RM^{d-1}\times \RM^{\pm}$ is the one originally used in \cite{KS04, KS04b}, its main advantage is that the above exact sequences are Wiener-Hopf extensions. However, we will use below a compressed version as half-space algebras.

As for the bulk algebra, the interface  algebra admits a family of representations $(\tilde\pi_{\omega,r})_{ (\omega,r)\in \Omega^*}$ giving rise to a covariant family of operators on $L^2(\RM^d)$. 
These are obtained by extending the evaluation representations which induce the covariant representations (\ref{eq-cov-rep}) to the points of $\Omega^*$. Then $\tilde\pi_{\omega,r}$ is the associated integrated representation. There are two interesting features about these representations: The first is that the representation $\tilde\pi_{\omega,\pm\infty}$ for $\omega \in \Omega$ can be interpreted directly as the natural surjection onto $\tilde\pi_\omega(\bulkalgebra)$ which factors through the evaluation map $q_\pm$. The second is that the representations $\tilde\pi_{\omega,r}$ and $\tilde\pi_{re_d\triangleright\omega,0}$ are unitarily equivalent for any finite $r\in \RM$ through conjugation by a so-called dual magnetic translation, hence one can actually ignore the variable $r$ if it is finite and focus on the family of representations $(\tilde\pi_{\omega,0})_{\omega\in \Omega}$. In particular, if $\Omega=\{\omega\}$ is one point only, which is the case if the operators are translation invariant, then  the representation alone $\tilde\pi_{\omega,0}$ is faithful and we can describe the algebras as in the introduction. In the absence of an external magnetic field we then get $\Aa_b=\CM\rtimes \RM^d\cong C_0(\RM^d)$ the isomorphism being given by the Fourier transform on $\RM^d$, and $\hat \Aa_i = \CM\rtimes \RM^{d-1}\otimes C([-\infty,+\infty])\rtimes_{tr}\RM
\cong C_0(\RM^{d-1})\otimes C([-\infty,+\infty])\rtimes_{tr}\RM$. 
The latter isomorphism is given by a partial Fourier transform and $tr$ denotes the translation action. Finally,  
$\Ee\cong C_0(\RM^{d-1})\otimes \Kk(L^2(\RM))$. 

In \cite{KS04, KS04b} the extension $C_0(\Omega^+)\rtimes_{\alpha^*,\gamma}\RM^d$ of (\ref{eq-half-WH}) was used to describe Hamiltonians which live on the positive half-space with the argument that the integral kernels of their resolvents tend to $0$ when $x_d\to -\infty$. However, 
below we want to consider half-space operators which are obtained from restrictions to the half-space of differential operators on $\RM^d$. It is not clear that the integral kernels of their resolvents are given by elements of the above extension. The apparent difficulty is that the set of elements used to define $\interfacealgebra$, notably $\tilde\pi(f)$ with $f\in L^1(\RM^d,C(\Omega^*))$, is built from functions that depend continuously on the position of the boundary, so it is not clear how sharply truncating a differential operator can result in a multiplier. The fact of the matter is, the algebras are larger than one might think: 
\begin{lemma}
\label{lemma:nonsmoothelements}
	Consider the Hilbert-space representations $(\tilde\pi_{\omega,0})_{\omega\in \Omega}$ of $\edgealgebra = C_0(\Omega^0) \rtimes_{\alpha^*, \gamma} \RM^d$ on $L^2(\RM^d)$, which define a faithful non-degenerate representation $\tilde\pi = \int_\Omega^\oplus \tilde\pi_{\omega,0} \mathrm{d}\PM(\omega)$ on $L^2(\RM^d)\otimes L^2(\Omega)$. 
	
	\noindent Let $X$ be the unbounded multiplication operator $$(X  f)(x) = x_d f(x)$$
	which is a self-adjoint operator when defined on a suitable dense subspace of $L^2(\RM^d)$ and thus also on $L^2(\RM^d)\otimes L^2(\Omega)$.
	
	\begin{enumerate}
		\item[(i)] $X$ is an $\edgealgebra$-multiplier when we consider $\edgealgebra$ via $\tilde{\pi}$ to be a $C^*$-subalgebra of $\Bb(L^2(\RM^d)\otimes L^2(\Omega))$.
		\item[(ii)] For every bounded Borel function $f:\RM \to \CM$, $f(X)$ is a multiplier of $\edgealgebra$. Moreover, for every pointwise convergent sequence $(f_n)_{n\in \NM}$ of Borel functions with $\sup_{n\in \NM}\norm{f_n}_\infty < \infty$ one has that $f_n(X)$ is a strictly convergent sequence of multipliers.
		\item[(iii)] The spectral projections $P_\pm =\chi(\pm X > 0)$ are $\interfacealgebra$-multipliers. When we consider $\bulkalgebra$ to be a subalgebra of $\interfacealgebra$ then every element $\hat{a}\in \interfacealgebra$ has a unique decomposition
       \begin{equation}
			\label{eq:decomp_discont}
			\tilde\pi(\hat{a}) = P_+ \tilde\pi(a_+) P_+ + P_- \tilde\pi(a_-)P_- + \tilde\pi(e)
		\end{equation}
        with $a_\pm\in \bulkalgebra$ and $e\in \edgealgebra$.

	\end{enumerate}
\end{lemma} 

\begin{proof} (i): $X$ is the well known self-adjoint $C_0(\RM)$-multiplier 
extended with the identity to $C_0(\Omega^0) = C_0(\RM)\otimes C(\Omega)$. 
All multipliers of a $C^*$-algebra $\Bb$ canonically define multipliers of all its twisted crossed products $\Bb \rtimes G$.

We prove that $P_\pm$ are $\edgealgebra$-multipliers, first for the case $d=1$. We have
$$\edgealgebra=C_{0}(\Omega\times \RM )\rtimes_{\alpha^*} \RM \simeq C( \Omega)\rtimes_\alpha \RM \rtimes_{\hat\alpha} \RM  \simeq C(\Omega) \otimes \KM(L^2(\RM))$$ 
where $\hat{\alpha}$ is the dual action. Via the explicit form of this Takai duality isomorphism \cite{Raeburn88} it is easy to check that any element $e\in \edgealgebra$ via its family of representations $(\tilde{\pi}_{\omega,0}(e))_{\omega\in \Omega}$ corresponds one-to-one to a norm-continuous family $(k_\omega)_{\omega\in \Omega}$ where $k\in C(\Omega,\KM(L^2(\RM)))$. Any $\omega$-independent bounded operator on $L^2(\RM)$ is therefore a multiplier of $\edgealgebra$, which includes any bounded Borel function $f(X)$ of $X$. Moreover, for a sequence of uniformly bounded pointwise convergent Borel functions, the measurable functional calculus gives rise to a $*$-strong convergent sequence of bounded operators $(f_n(X))_{n\in \NM}$. It is well-known that a product of a strongly-convergent sequence of operators and a compact operator converges in operator-norm, hence the strict convergence of $(f_n(X))_{n\in \NM}$ as an $\edgealgebra$-multiplier follows easily from a density argument, proving (ii) in the case $d=1$. 

To conclude the same for $d>1$ we note that one can choose a gauge for the vector potential of the magnetic field so that one can factor the twisted crossed product
$$C_0(\Omega^0)\rtimes_{\alpha^*,\gamma}\RM^d\simeq \left(C_0(\Omega^0)\rtimes_{\alpha^*_d} \RM\right)\rtimes_{\tilde{\alpha},\tilde{\gamma}}\RM^{d-1}$$
with $\alpha^*_d$ the translation in the direction $e_d$, $\tilde{\gamma}$ the restriction of $\gamma$ and $\tilde{\alpha}$ a modification of the action parallel to the edge which implements the magnetic field. The case $d=1$ shows that $f(X)$ is a multiplier of the innermost crossed product, hence $f(X)$ is an $\edgealgebra$-multiplier also for $d>1$ and any bounded Borel function. 

For (iii) the natural $\RM^d$-equivariant inclusion of $C(\Omega)$ into $C(\Omega^*)$ as a unital subalgebra extends to crossed products and then allows us to identify $\bulkalgebra$ with a subalgebra of $\interfacealgebra$. Choose a continuous switch function $\Theta\in C_0(\RM)$ with $\Theta(-\infty)=0$ and $\Theta(\infty)=1$. 
Setting $\Pp_+ =\Theta(X)$ one easily can read off from \eqref{eq:linear_span_crossedproduct} that left multiplication with $\Pp_+$ defines a bounded multiplier of $\interfacealgebra$. Thus we have a well-defined map $$(a_+,a_-)\in \bulkalgebra\oplus \bulkalgebra \mapsto \Pp_+ \tilde{\pi}(a_+) \Pp_+  + (1-\Pp_+) \tilde{\pi}(a_-) (1-\Pp_+).$$ This provides a linear inverse to the surjection $q: \interfacealgebra \to \bulkalgebra\oplus \bulkalgebra$ since  $q(\Pp_+)=(1,0)$ for the extension of $q$ to the multiplier algebra. By exactness we can then conclude that any element $\hat{a}\in \interfacealgebra$ can be written uniquely in the form 
\begin{equation}
\label{eq:interface_decomp}
\tilde{\pi}(\hat{a})=\Pp_+\tilde{\pi}(a_+)\Pp_+ + (1-\Pp_+)\tilde{\pi}(a_-)(1-\Pp_+) + \tilde{\pi}(e)
\end{equation}
for some $e\in \edgealgebra$ and $a_\pm \in \bulkalgebra$.

We may assume w.l.o.g.\ that the support of $\Theta$ is contained in $\RM^+$. There then exists some $g\in C_0(\RM)$ such that $$\Pp_+ - P_+ = P_+ g(X).$$
As $g(X) \tilde{\pi}(a_\pm)\in \tilde{\pi}(\edgealgebra)$ and $P_+$ is an $\edgealgebra$-multiplier we then conclude that $(\Pp_+-P_+)\tilde{\pi}(a_\pm) \in \tilde{\pi}(\edgealgebra)$.
Successively one can therefore replace all $\Pp_\pm$ in \eqref{eq:interface_decomp} by the projections $P_\pm$ and only produce corrections in $\edgealgebra$.
\end{proof}

\begin{remark}
This shows $\interfacealgebra$ does in fact contain many elements which appear to be  ``discontinuous'' in the direction normal to the boundary. In particular, it contains sharp truncations of resolvents of bulk differential operators to a half-space.

In general, the set of covariant potentials which define a multiplier of some algebra $C_0(\Omega)\rtimes \RM^d$ is usually larger than $\mult(C_0(\Omega))\simeq C_b(\Omega)$, e.g. containing also potentials with jump-discontinuities. One can intuitively make sense of that by noting all generators \eqref{eq:linear_span_crossedproduct} involve a smoothing pseudodifferential operator which regularizes the discontinuities. 
\end{remark}

We can now define the versions of half-space and edge algebras $\pmhsalgebra$, $\pmedgealgebra$ which we use here, as compressions of the interface and edge algebra to half-space, namely
\begin{align*}
\pmhsalgebra &= P_\pm \interfacealgebra P_\pm\\
\pmedgealgebra &= P_\pm \edgealgebra P_\pm.
\end{align*}
Here we consider $P_\pm$ as bounded multipliers of the abstract algebra, thus those are subalgebras of $\interfacealgebra $ and $ \edgealgebra$ respectively.
They fit into the exact sequences 
$$\pmedgealgebra \hookrightarrow \pmhsalgebra \stackrel{q_\pm}{\twoheadrightarrow}\bulkalgebra.$$
Note that, in the introduction where we considered translation invariant operators, we denoted $\Ee_+$ and $\Aa_+$ simply by $\Ee$ and $\Aa$.

\subsection{Affiliation and boundary conditions for translation-invariant differential operators}
\label{sec:aff_interface}
In this section we consider the problem of affiliation to $\hsalgebra$ for matrix-valued self-adjoint differential operators. Since this is difficult in general we restrict ourselves to the translation-invariant case so that Fourier transform reduces it to a family of one-dimensional problems. 
Consider a general translation-invariant (matrix-valued) self-adjoint differential operator $H$ which is affiliated to $M_N(C_0(\RM^d))$, i.e.\ formally 
$$H= \sum_{I\in\NM^d} c_I \partial^I$$ with scalar matrices $c_I\in M_N(\CM)$. 
Restricting $H$ as an operator on $L^2(\RM^d)\otimes \CM^N$ to the domain $C^\infty_c(\RM^{d-1}\times \RM_+) \otimes \CM^N$ and taking the graph closure we obtain a symmetric operator $\mathring{H}$ which after a partial Fourier transform (in the components parallel to the interface line) has the form
$$\mathring{H}=\int_{\RM^{d-1}}^\oplus \mathrm{d}k \, \mathring{H}(k).$$
For almost all $k$, $\mathring{H}(k)$ is a matrix-valued symmetric operator which depends polynomially on $k$ and has the core $C^\infty_c(\RM_+)\otimes\CM^N$. 

The first question is whether $\mathring{H}$ is affiliated to $M_N(\hsalgebra)$, so that we can apply the $C^*$-algebraic version of self-adjoint extension theory to see which self-adjoint extensions of $\mathring{H}$ are affiliated to  $M_N(\hsalgebra)$.
\begin{proposition} 
	\label{prop:extensions_fiberwise}
    Consider a symmetric operator $\mathring{H}$ obtained from a translation invariant differential operator $H$ as above. Suppose that there exists one self-adjoint extension $\mathring{H}_{u_0}$ which is an $M_N(\hsalgebra)$-multiplier and such that $u_0\in \mult(\hsalgebra)\otimes M_N(\CM)$. Then $\mathring{H}$ is a multiplier of $\hsalgebra\otimes M_N(\CM)$. Moreover, a self-adjoint extension $\mathring{H}_u$ of $\mathring{H}$ is a multiplier of $\hsalgebra\otimes M_N(\CM)$ if and only if
    $u\in \mult(\hsalgebra)\otimes M_N(\CM)$.
\end{proposition}
\begin{proof}
To simplify the notation we assume $N=1$.

The conditions Proposition~\ref{proposition:symmetric_multipliers} (i) and (ii) hold by assumption for $H_{u}$, hence it suffices to exhibit a norm-dense subset of $\Dd_{\hsalgebra}(\mathring{H})$.

Let $\Ss$ be the set of integral operators on $L^2(\RM^{d-1}\times \RM_+)$ whose integral kernel $k:(\RM^{d-1}\times \RM_+)\times (\RM^{d-1}\times \RM_+)\to \CM$ has the form
$$k(x,y)= \sum_{i=1}^n f_i(x_d) g_i(x-y)$$
for some finite $n$ and functions $g_i\in C_c^\infty(\RM^d)$ and bounded uniformly $f_i\in C^\infty([0,\infty])$ (i.e. smooth functions on $\RM_+$ which admit limits at $0$ and $\infty$). Define similarly $\Ss_0\subset \Ss$ with the additional condition that the functions $g_i$ vanish in a neighborhood of $0$. It is a consequence of Lemma~\ref{lemma:nonsmoothelements} that $\Ss$ is norm-dense in $\hsalgebra$ since there is a dense subset of $\interfacealgebra$ whose projections to the positive halfspace take this form. By Lemma~\ref{lemma:nonsmoothelements}(ii) the subset $\Ss_0$ is norm-dense in $\Ss$ since approximating $g_i$ as a pointwise limit gives an operator-norm convergent sequence. The elements of $\Ss_0$ map the domain $C^\infty_c((\RM\setminus \{0\})\times \RM^{d-1})$ of $\mathring H$ to itself and $\mathring{H}$ being a linear combination of partial derivatives acts on $\Ss_0$ by differentiating the integral kernels, hence $\mathring{H}\Ss_0\subset \Ss_0$ by construction. In particular, $\Ss_0\subset \Dd_{\hsalgebra}(\mathring{H})$. Thus condition (iii) of Proposition~\ref{proposition:symmetric_multipliers} is also satisfied.

The last statement now follows from Theorem~\ref{theorem:saextensions}.
\end{proof}
Additional bounded or even relatively bounded potentials can be incorporated using Proposition~\ref{prop-res-perturb}, hence it is in most cases enough to work out the extensions for the pure differential operators and then add the potentials afterwards.
\begin{remark}
The proposition as stated also holds for magnetically covariant differential operators; one must merely modify the construction of $\Ss$ with the appropriate magnetic covariance relation. However, since one generally cannot use the Fourier transform to reduce to a family of one-dimensional problems, computation of deficiency subspaces and von Neumann unitaries will be unfeasible in most cases. An exception is the two-dimensional case where one can use the Landau gauge to reduce to a differential operator that is translation-invariant in the direction parallel to the line $x_d=0$. We will leave those problems to future work.
\end{remark}
The second question is which of the affiliated self-adjoint extensions of $\mathring{H}$ are also resolvent-affiliated. It has a similar answer: if one can exhibit one extension which is resolvent-affiliated then this boils down to the study of the von Neumann unitaries determining the extension. 
In the translation invariant case any multiplier $u$ of $\hsedgealgebra \cong C_0(\RM^{d-1})\otimes \KM(L^2(\RM_+))$ can be viewed as a strongly continuous function $\RM^{d-1}\ni k\mapsto u(k)$. If each $u(k)$ is a partial isometry with finite and constant rank then this function is even continuous in the norm topology.
\begin{proposition}
\label{prop:res_aff_vn}
Consider a symmetric operator 
$\mathring{H}=\int_{\RM^{d-1}}^\oplus \mathrm{d}k \, \mathring{H}(k)$ obtained from a translation-invariant differential operator $H$ as above. Suppose that there exists one self-adjoint extension $\mathring{H}_{u_0}$ which is a resolvent-affiliated $M_N(\hsalgebra)$-multiplier with $u_0\in M_N(\mult(\hsalgebra))$. Assume furthermore that the deficiency subspaces of $\mathring{H}(k)$ have finite and equal dimensions. Then
a self-adjoint extension $\mathring{H}_u$ of $\mathring{H}$ is a resolvent-affiliated multiplier of $M_N(\hsalgebra)$ if and only if $u-u_0\in\edgealgebra$. This is equivalent to $k\mapsto h(k)-u_0(k)$ being continuous and
\begin{equation}
\label{eq:vn_unitary_asymptotic}
\lim_{k \to \infty} \norm{u(k) -u_0(k)} = 0,
\end{equation}
or,  $k\mapsto (\mathring{H}(k)_{u(k)}-\imath)^{-1}$ being continuous and
\begin{equation}
\label{eq:resolvent_asymptotic}
\lim_{k \to \infty} \norm{(\mathring{H}(k)_{u(k)}-\imath)^{-1}}=0.
\end{equation}
\end{proposition}
\noindent{\bf Proof.}
Assume again $N=1$ for simplicity. 
As part of the assumption,  $\mathring{H}_{u_{0}}$ is affiliated to $\hsalgebra$ hence by the last result all self-adjoint extensions of $\mathring{H}$ are affiliated to $\hsalgebra$. 
As $\hsalgebra\cong C_0(\RM^{d-1})\otimes C([-\infty,+\infty])\rtimes_{tr}
\RM$ we see that resolvent-affiliation of a self-adjoint extension $(\mathring{H}(k)_{u(k)}$ is equivalent to the statement that 
$k\mapsto (\mathring{H}(k)_{u(k)}-\imath)^{-1}$ is a $C_0$-function in $k$. This is the last statement.

Now suppose that  $\mathring{H}_{u_{0}}$ and $\mathring{H}_u$ are resolvent-affiliated self-adjoint extensions. Then  $k\mapsto (\mathring{H}(k)_{u(k)}-\imath)^{-1}-(\mathring{H}(k)_{u_0(k)}-\imath)^{-1}$ is $C_0$ in $k$. The latter is equivalent to 
$k\mapsto u(k)-u_0(k)$ being $C_0$, which is \eqref{eq:vn_unitary_asymptotic}. 
As the deficiency spaces are finite-dimensional by assumption, $u(k)-u_0(k)$ has finite rank and hence $k\mapsto u(k)-u_0(k)\in C_0(\RM^{d-1})\otimes C_0(\RM)\rtimes_{tr}\RM\cong \edgealgebra$. 

For the converse, $u-u_0\in \edgealgebra$ implies that the difference of the resolvents belong to $\edgealgebra$ hence $(\mathring{H}_u-\imath)^{-1}\in \hsalgebra$. 
\hfill $\Box$
\begin{remark}
If \eqref{eq:resolvent_asymptotic} holds then it is easy to see that for any finite interval $\Delta$ there can at most be a compact set of $k \in \RM^{d-1}$ such that  $(\mathring{H}_k)_{u(k)}$ has any spectrum in $\Delta$. Moreover, whenever $\Delta$ does not overlap the spectrum of the bulk operator then $\mathring{H}_u$ is invertible modulo $\hsedgealgebra$ in that interval. Since $\hsedgealgebra\simeq C_0(\RM^{d-1})\otimes \KM$ this means $k\in \RM^{d-1} \mapsto (\mathring{H}_k)_{u(k)}-\lambda$ is a family of self-adjoint Fredholm operators for any $\lambda$ not in the bulk spectrum. For fixed $k$ there are therefore at most isolated eigenvalues inside any bulk gap, which form continuous bands of edge modes as one varies $k$. Resolvent-affiliation on the half-space requires in particular that all bands of edge modes eventually leave any finite energy interval as $\abs{k}\to\infty$.
\end{remark}

For a complete characterization of (resolvent-affiliated) boundary conditions we therefore need to find at least one self-adjoint boundary condition for which we can assert that we have a (resolvent-affiliated) $\hsalgebra$-multiplier. Sometimes one can directly find such an operator, for example, for the Laplacian with Dirichlet boundary conditions where it is easy to compute the resolvent explicitly using the reflection principle (see Section~\ref{ssec:laplace} below), thereby providing a reference extension that is known to be a resolvent-affiliated multiplier. 

For the general case it will be useful to have a more constructive way to proceed. Instead of the half-space we now work on $\RM^d$ but add a fiducial boundary at $x_d=0$, i.e. we restrict a translation-invariant differential operator $H$ to a symmetric operator $\mathring{H}$ by restricting its domain to $C^\infty_0(\RM^{d-1}\times (\RM\setminus \{0\}))$. In that case we can still use a partial Fourier transform to decompose into one-dimensional differential operators on
$$H= \sum_{I\in\NM^d} c_I \partial^I, \qquad \mathring{H}=\int_{\RM^{d-1}}^\oplus \mathrm{d}k \, \mathring{H}(k).$$

Since $H$ is affiliated to $\bulkalgebra$ the inclusion $\bulkalgebra\hookrightarrow\interfacealgebra$ implies that it is also a $\interfacealgebra$-multiplier. Therefore we have a canonical self-adjoint extension and can prove exactly as Proposition~\ref{prop:extensions_fiberwise}:

\begin{proposition} 
	\label{prop:extensions_fiberwise2}
	Assume that $H$, $\mathring{H}$ are translation-invariant differential operators as above. We can write $H=\mathring{H}_{u_T}$ for a unique unitary $u_T$ corresponding to the so-called transparent boundary condition.
	
	If $u_T\in \mult(\interfacealgebra)$ then $\mathring{H}$ is a symmetric unbounded multiplier of $\interfacealgebra$.
\end{proposition}
This criterion is, for example, satisfied in all of our examples in Section~\ref{sec:examples}.

Half-space operators then arise naturally from considering interface conditions that decouple the two half-spaces:
\begin{proposition}
Let $\mathring{H}$ be the restriction of a translation-invariant bulk operator as above and assume it is a symmetric $\interfacealgebra$-multiplier.

If the positive halfspace is an invariant subspace for  $\mathring{H}_{{u}}$ in the sense that  $P_+\mathring{H}_{u}\subset \mathring{H}_{{u}}P_+$ for the halfspace projection $P_+\in \mult(\interfacealgebra)$, then $P_+ \mathring{H}_{{u}}P_+$ is a self-adjoint $\hsalgebra$-multiplier. If $\mathring{H}_{u}$ is $\interfacealgebra$-resolvent-affiliated then $P_+\mathring{H}_{u}P_+$ is $\hsalgebra$-resolvent-affiliated.
\end{proposition}
\noindent{\bf Proof.} It is a standard result that $\mathring{H}_{{u}}$ is a direct sum of two self-adjoint operators defined on dense subsets of $\Ran(P_+)$ and $\Ran(P_-)$ respectively. Therefore $F(\mathring{H}_{{u}})=F(P_+\mathring{H}_{{u}}P_+)\oplus F(P_-\mathring{H}_{{u}}P_-)$ which is in $\mult(\hsalgebra)\oplus \mult(\neghsalgebra)$ since by definition $\pmhsalgebra = P_\pm \interfacealgebra P_\pm$. To complete the proof that $P_+\mathring{H}_{{u}}P_+$ is an $\hsalgebra$-multiplier one merely needs to note that the dense subset $\Ss_0$ from the proof of Proposition~\ref{prop:extensions_fiberwise} is contained in $\Dd_{\hsalgebra}(P_+\mathring{H}P_+)$.
\hfill $\Box$ 

Note that $\interfacealgebra$-resolvent-affiliation can be characterized exactly by the analogous condition as Proposition~\ref{prop:res_aff_vn}. If the bulk Hamiltonian $H$ is $\bulkalgebra$-resolvent-affiliated then the $\interfacealgebra$-resolvent-affiliated and ultimately the $\hsalgebra$-resolvent-affiliated extensions can therefore be read off by comparing the asymptotic behavior of the von Neumann unitary with that of $u_T$.
\begin{remark} 
	Decoupling matching conditions obviously need not exist. For example, the restriction of the one-dimensional momentum operator $-i\partial$ on $L^2(\RM)$ to $C^\infty_c(\RM\setminus \{0\})$ cannot have self-adjoint extensions which decouple the two half lines, as the momentum operator restricted to the positive halfspace has no self-adjoint extensions.
	In general, the deficiency subspaces each have the finite-dimensional subspace $\Nn_{k,\pm}=P_+ \Ker(\mathring{H}^*_k \pm \imath)$. For fixed $k$, decoupling matching conditions exist if and only if $\mathrm{dim}(\Nn_{k,+})=\mathrm{dim}(\Nn_{k,-})$, which is clearly the case if and only if one and thus both half-space restrictions individually have self-adjoint extensions. To see this, note that the Cayley transform $\Cc(\mathring{H})$ is block-diagonal w.r.t. the decomposition $P_+\oplus P_-$ and hence the resolvent of $\mathring{H}_u$ commutes with $P_\pm$ if and only if the fiberwise finite-rank operator $u$ is also block-diagonal.
\end{remark}

\subsection{Bulk-boundary correspondence}
\label{sec:applications}

We have the bulk-edge exact sequence $\hsedgealgebra \hookrightarrow \hsalgebra \stackrel{q_+}{\twoheadrightarrow}\bulkalgebra$ for the positive halfspace. Throughout the explicit examples we will set $\Omega=\{*\}$, i.e. we there will be no disorder and also vanishing magnetic field, but the abstract statements cover all of those cases as well. In particular, all K-theoretic statements we derive are robust to introduction of homogeneous disorder which breaks the translation-invariance. 

A half-space Hamiltonian is a self-adjoint $\hsalgebra$-multiplier $\hat{H}$. We call $\hat H$ \emph{asymptotically invertible} if it is invertible modulo $\hsedgealgebra$. In particular, the bulk operator $H:=q_+(\hat H)$ then has a spectral gap around $0$. Such $\hat H$ defines the edge class $[\hat H]_1\in K_1(\edgealgebra)$.

The goal of bulk-edge correspondence is to derive as much information as possible about the edge invariants from knowledge about the bulk operators.  For this we can rely on the following strategies developed in Section~\ref{sec:ktheory}:
\begin{enumerate}
    \item[(i)] Simple bulk-edge correspondence: When the conditions for Theorem~\ref{th:bbc_strongly_affiliated} are satisfied then the edge invariant $[\hat{H}]_1$ is determined by applying the boundary map to the class $[H]_0=[P_{\leq 0}(H)]_0-[s(P_{\leq 0}(H))]_0$ of the bulk Fermi projection.
	\item[(ii)] Relative bulk-edge correspondence: If $\hat{H},\hat{H}'$ are  invertible modulo $\hsedgealgebra$ and $\hsalgebra$-comparable then Theorem~\ref{th:rel_bbc} applies and
	$$[\hat{H}]_1-[\hat{H}']_1 = \partial [H, H']_0 = \partial([P_{\leq 0}(H),P_{\leq 0}(H')]) \,.$$
	This applies when $\hat{H}, \hat{H}'$ have the same boundary conditions and are a relatively $\hat{\Aa}$-compact perturbation of each other.
	\item[(iii)] Extension theory: If $\hat{H},\hat{H}'$ are locally $\hsedgealgebra$-comparable self-adjoint extensions of a common symmetric $\hsalgebra$-multiplier corresponding to von Neumann unitaries $u,v$, respectively, then 
	$$[\hat{H}]_1-[\hat{H}']_1 = [1+uv^*-vv^*]_1\, .$$
	This allows us to compare different boundary conditions for the same symmetric operator.
\end{enumerate} 
The methods (ii) and (iii) only make statements over the differences of edge invariants, to obtain complete information one will need to use them to reduce to a system where the simple bulk-edge correspondence applies or is simple enough to compute the edge invariant directly.
\begin{remark}
	As mentioned before, the simple bulk-edge correspondence (i) is just a special case of (ii) which applies precisely when $\hat{H}'$ can be taken to be a scalar matrix, whence $[\hat{H}']_1=0$. 
\end{remark}
Applying the theory requires several analytical/algebraic tasks from verifying affiliation to computing von Neumann unitaries and finding asymptotically invertible extensions etc. Importantly the required algebraic conditions can and need to fail in some cases, thus one needs to check them very carefully. To avoid the technical overhead we will mostly restrict ourselves where $\hat{H}$ is $\hsalgebra$-resolvent-affiliated $(\hat{H}+\imath)^{-1}\in M_N(\hsalgebra)$, which is simultaneously also the case where one can formulate the most robust bulk-edge correspondence.

We want to consider a class of bulk Hamiltonians of the form
$$H_V = D + V$$
where $D$ is a $M_N(\CM)$-valued invertible elliptic differential operator and $V$ is a matrix with entries in $\mult(\bulkalgebra)$. The kinetic part (contained in $D$) is fixed, but one can still add a potential or tune some parameters (mass terms, coupling strengths etc.). For simplicity we assume $V$ is bounded; in principle our formalism can handle relatively $\bulkalgebra$-compact perturbations, such as differential operators of lower order than $D$, but that would lead to substantial complications when considering the eventual half-space problems. We impose the resolvent-affiliation $(D+\imath)^{-1}\in M_N(\bulkalgebra)$ as is appropriate for elliptic $D$ (see Theorem~\ref{th:resolvent_aff_elliptic}) which also implies the resolvent-affiliation $(H_V+\imath)^{-1}\in M_N(\bulkalgebra)$.

The corresponding half-space models then are to take the form
\begin{equation}
	\label{eq:halfspace_model}
	\hat{H}_{\hat{V},u} = \mathring{D}_u + \hat{V}
\end{equation}
where $\mathring{D}$ is the symmetric restriction of $D$ to the positive halfspace of which one chooses a self-adjoint extension corresponding to a von Neumann unitary $u$. Finally one can add a potential $\hat{V}\in C([0,\infty])\otimes M_N(\CM)\otimes \mult(\Aa_b)$ which converges to a scalar matrix at infinity.

Our main result about this class of Hamiltonians is as follows:

\begin{theorem}
	\label{th:bulk_edge_pluscorrection}
	Let $\hat{H}_{\hat{V},u}$ be of the form  $$\hat{H}_{\hat{V},u}= \mathring{D}_u + \hat{V}$$ where $\mathring{D}$ is a symmetric $\hsalgebra$-multiplier, $V=V^*\in \Mm(\hsalgebra)$. Assume there exists a von Neumann unitary $u_0 \in \Mm(\hsalgebra)$ which determines a $\hsalgebra$-resolvent-affiliated self-adjoint extension $(\hat{H}_{\hat{V},u_0}+\imath)^{-1}\in \hsalgebra \otimes \KM$ and that $H_V:= q(\hat{H}_{\hat{V},u})$ is invertible. For any other $\hsalgebra$-resolvent-affiliated $\hat{H}_{\hat{V},u}$ one has
	
	$$[\hat{H}_{\hat{V},u}]_1 = \partial_+([H_V, D]_0)+ [\mathring{D}_{u_0}]_1 + [1+uu_0^*-u_0u_0^*]_1.$$
\end{theorem}
\noindent{\bf Proof.} Due to the resolvent-affiliation any perturbation $\hat{V}$ is relatively $\hsalgebra$-compact, hence for equal boundary condition
$$[\hat{H}_{\hat{V},u}]_1-[\mathring{D}_{u}]_1  = \partial_+([H_V,D]_0)$$
by Theorem~\ref{th:rel_bbc} (here we use that $D$ is invertible by assumption). If we instead want to compare with $\mathring{D}_{u_0}$ then apply Proposition~\ref{prop:comparison_saext_ktheory} to obtain
\begin{equation}\label{eq-jk-bdy-cond1}
	[\mathring{D}_u]_1 - [\mathring{D}_{u_0}]_1 = [1+u u_0^*-u_0u_0^*]_1 \in K_1(\hsedgealgebra).
\end{equation}
\hfill $\Box$

In our examples in Section~\ref{sec:examples} we will consider different translation-invariant operators $D$ and parametrized families of self-adjoint boundary conditions (e.g., arising from local boundary equations). Identifying which of the corresponding extensions are resolvent-affiliated and the resulting corrections to the bulk–boundary correspondence is then intricate but tractable in principle, since then all computations reduce to computations with families of one-dimensional differential operators. Both the relative bulk invariant $\partial_+([H_V, D]_0)$ as well as the correction $[1+uu_0^*-u_0u_0^*]_1$ coming from the boundary condition can be computed without knowing the full spectral decomposition of the halfspace operators. To determine from this information the boundary invariant $[\hat{H}_{\hat{V},u}]_1$ one then merely needs to fix the normalization by computing the edge invariant for the model Hamiltonian $\mathring{D}_{u_0}$ for any choice of resolvent-affiliated boundary condition.

\begin{remark}
	
Denote by $\Uu$ the set of von Neumann unitaries for which $\mathring{D}_u$ is self-adjoint and $\Uu_R\subset \Uu$ those for which $\mathring{D}_u$ is resolvent-affiliated. The norm-topology on this set is equivalent to the norm-resolvent topology of the corresponding self-adjoint extensions, since
$$\norm{u-v} = \frac{1}{2}\norm{(\mathring{D}_u+\imath)^{-1} - (\mathring{D}_u+\imath)^{-1}}, \quad \forall u,v\in \Uu.$$
The set $\Uu_R$ can have infinitely many connected components. Fixing any $u_0\in \Uu_R$ to serve as a reference the correction term $[1+uu_0^*-u_0u_0^*]_1$ in \eqref{eq-jk-bdy-cond1}
for each $u\in \Uu_R$ depends only the connected component of $u$ in $\Uu_R$. Conversely, jumps of $[1+uu_0^*-u_0u_0^*]_1$ can only happen when resolvent-affiliation fails. Given a parametrized family of self-adjoint boundary conditions one can in principle work out the connected components and the corresponding corrections to bulk-boundary correspondence to obtain a complete classification of the resolvent-affiliated boundary conditions. Moreover, this classification remains unchanged under bounded perturbations by Proposition~\ref{prop:extensions_perturbation_ktheory}.
\end{remark}

We can similarly also consider interface models, which plays out almost exactly the same. Here we have interface Hamiltonians affiliated to $\interfacealgebra$ of the form
$$\hat{H}_{\hat{V}, u} = \mathring{D}_u + \hat{V}$$
where $\hat{V}\in \Mm(\interfacealgebra)$ and $\mathring{D}$ is the symmetric restriction of the formal differential operator $D$ to $\RM^d\setminus (\RM^{d-1}\times \{0\})$ as considered as in Section~\ref{sec:aff_interface}. Such a Hamiltonian models an interface between asymptotic bulk Hamiltonians $H_{V_\pm}$ where $V_\pm$ is the limit of $\hat{V}$ at $\pm \infty$. At the fiducial line $x_d=0$ one can enforce non-trivial matching conditions by the choice of self-adjoint extension, as is sometimes necessary to model sharp boundaries where two materials meet. There is always a preferred extension $u_T$, the transparent boundary condition, determined by $D=\mathring{D}_{u_T}$.
\begin{theorem}
	\label{th:bulk_interface_pluscorrection}
	Let $\hat{H}_{\hat{V},u}$ be of the form $$\hat{H}_{\hat{V},u}= \mathring{D}_u + \hat{V}$$ where $\mathring{D}$ is a symmetric $\interfacealgebra$-multiplier and $V=V^*\in \mult(\interfacealgebra)\otimes M_N(\CM)$. Assume that there exists a von Neumann unitary $u_T\in \mult(\interfacealgebra) \otimes M_N(\CM)$ (the transparent matching condition) such that $D:=D_{u_T}$ is $\bulkalgebra$-resolvent-affiliated. 
    
    If $H_{V_\pm}:= q_\pm(\hat{H}_{\hat{V},u})=D+ V_\pm$ are invertible $\Aa_b$-multipliers then comparing any resolvent-affiliated matching condition $u$ with $u_T$ one has
	
	$$[\hat{H}_{\hat{V},u}]_1 = \partial_+([H_{V_-},H_{V_+}]_0)+ [uu_T^*-u_Tu_T^*]_1.$$
\end{theorem}
\noindent{\bf Proof.}
We have
$$[\hat{H}_{\hat{V},u}]_1 = [\hat{H}_{\hat{V},u_T}]_1 + [u u_T^*-u_Tu_T^*]_1$$
by Proposition~\ref{prop:extensions_perturbation_ktheory} which reduces the computation to the transparent case.

Note that $\hat{H}_{\hat{V},u_T}=D+\hat{V}$ is $\interfacealgebra$-resolvent-affiliated for any bounded multiplier $\hat{V}$ due to $\bulkalgebra$-resolvent-affiliation of $D$. We can therefore apply the relative bulk-edge correspondence Theorem~\ref{th:rel_bbc} under the exact sequence
$$0 \to \edgealgebra \to \interfacealgebra \to \bulkalgebra\oplus \bulkalgebra \to 0$$ to yield
$$[\hat{H}_{\hat{V},u_T}]_1-[D]_1 = \partial[P_{\leq 0}(H_{V_+})\oplus P_{\leq 0}(H_{V_-}), P_{\leq 0}(D)\oplus P_{\leq 0}(D)]_0$$
for the boundary map $\partial: K_0(\bulkalgebra\oplus \bulkalgebra)\to K_1(\Ee)$.
Note that $[D]_1=0$ since $D$ is invertible. The boundary map satisfies
$$\partial =(\partial_+ \oplus 0) - (0\oplus \partial_+)$$
with the homomorphism $\partial_+: K_0(\bulkalgebra)\to K_1(\hsedgealgebra)\simeq K_1(\edgealgebra)$ since one can consider the two copies of $\bulkalgebra$ separately and the minus sign appears, since changing the orientation of the translation action in a Wiener-Hopf extension induces a minus sign for the boundary maps. This gives
\begin{align*}
	&\partial[P_{\leq 0}(H_{V_+})\oplus P_{\leq 0}(H_{V_-}), P_{\leq 0}(D)\oplus P_{\leq 0}(D)]_0 \\ &=\partial_+[P_{\leq 0}(H_{V_+}), P_{\leq 0}(D)]_0 - \partial_+[P_{\leq 0}(H_{V_-}), P_{\leq 0}(D)]_0 \\
	&=\partial_+([P_{\leq 0}(H_{V_+}), P_{\leq 0}(D)]_0 - [P_{\leq 0}(H_{V_-}), P_{\leq 0}(D)]_0 \\
	&=\partial_+([P_{\leq 0}(H_{V_+}), P_{\leq 0}(D)]_0 + [ P_{\leq 0}(D),P_{\leq 0}(H_{V_-})]_0 \\
	&=\partial_+[P_{\leq 0}(H_{V_+}), P_{\leq 0}(H_{V_-})]_0\, .
\end{align*}
\hfill $\Box$

The correction vanishes by definition for transparent boundary conditions, in the case where a domain wall is produced by a continuous variation of a potential  the interface invariant is therefore entirely determined by the bulk Hamiltonians.
\begin{remark}
	Halfspace models are a special case of interfaces where the boundary condition decouples both sides. Since the boundary condition can change the edge invariant in the halfspace setup it is therefore easy to see that there are also non-trivial matching conditions which lead to a non-vanishing corrections $[u u_T^*-u_Tu_T^*]_1$.
	
\end{remark}

\newcommand{\ra}{a}
\section{Examples}
\label{sec:examples}
In this section we consider translation-invariant models in two dimensions with constant coefficients with the aim to obtain explicit expressions for the bulk and edge invariants. 
Let us emphasize, however, that our K-theoretic results apply equally well to aperiodic models and tell us in particular that the observed bulk- and edge invariants are stable under arbitrary disorder that does not close the bulk spectral gap.

We start by recalling a fundamental result about certain pairings of $K$-theory classes with cyclic cohomology classes. 
The exact sequence underlying the bulk-edge correspondence is equivalent to the Wiener-Hopf extension \eqref{eq-SEShalf}, hence its boundary maps are isomorphisms. Moreover, when one derives numerical invariants using the pairing with cyclic cohomology one has the following duality:
\begin{theorem}[{\cite[Theorem 4.5.3]{SSt}}]
The boundary map $\partial_+: K_0(\bulkalgebra)\to K_1(\hsedgealgebra)$ is an isomorphism with 
$$\langle \Ch_d, x\rangle = \langle \Ch_{d-1}, \partial_+(x)\rangle, \qquad \forall x \in K_0(\bulkalgebra)$$
for the Chern cocycles defined from the $\RM^d$- and $\RM^{d-1}$-actions respectively (cf. Example~\ref{ex:chern_cocycles}).
\end{theorem}
We have described the algebras 
relevant for translation invariant differential operators in the introduction and Section~\ref{sec:aff_interface}. Let us briefly recapitulate the situation for the half-sided problem.
The bulk algebra is $\bulkalgebra=\CM \rtimes \RM^2$ (with trivial action), or, for $N\times N$-matrix valued operators $M_N(\bulkalgebra)$. Under the two-dimensional Fourier transformation $\CM \rtimes \RM^2$ is isomorphic to $C_0(\RM^2)$ and so we have $K_0(\bulkalgebra)\cong \ZM$ with the isomorphism induced by the top Chern cocycle 
$$K_0(\bulkalgebra) \ni x \mapsto \langle \Ch_2, x\rangle \in \ZM.$$
For a strongly affiliated bulk Hamiltonian $H$ with spectral gap around $0$ this integer is (up to physical constants) equal to the bulk Hall conductivity
$$\sigma_b =  \langle \Ch_2, [H]_0 \rangle $$
the right-hand side coinciding with \eqref{eq:chern}. Here we recall that a matrix valued operator which is affiliated to $\bulkalgebra$ is strongly affiliated if its bounded transform belongs to $M_N((\CM\rtimes\RM^2)^\sim)$ which is isomorphic to $M_N(C(\SM^2))$, as the $2$-sphere is the one-point compactification of $\RM^2$. For some examples below this is not the case and we need to consider relative pairings between 
$M_N(\bulkalgebra)$-comparable Hamiltonians, see Section~\ref{ssec:massive_dirac}.

The edge algebra $\hsedgealgebra$ is isomorphic to $C_0(\RM)\otimes\Kk(L^2(\RM_+))$ so that $K_1(\hsedgealgebra)\simeq \ZM$. The isomorphism is obtained by pairing the winding number cocycle $\Ch_{1}= \mathrm{Wind}$ where 
$$\langle \mathrm{Wind},[U]_1\rangle = \mathrm{Wind}(U^*,U):= \frac{1}{2\pi \imath} \int_{\RM}\Tr((U^*(k)-\one)\partial_k U(k)) \mathrm{d} k$$
defined for any unitary function $U\in \one + C_0(\RM)\otimes \KM$ representing the class $[U]_1\in K_1(C_0(\RM))$. 
If $\hat{H}$ is a $\hsalgebra$-multiplier which is invertible modulo $\hsedgealgebra$ one can write the pairing in the form (cf. \cite{KRS})
$$\sigma_e = \langle \Ch_{1}, [\hat{H}]_1\rangle = \frac{1}{\abs{\Delta}} \hat{\Tt}(P_{\Delta }(\hat{H}_u) [\hat{H}_u, X_1]), \qquad \in \ZM$$
where the right-hand side physically corresponds to the conductance of the edge states in a small enough spectral interval $\Delta$ around $0$. 

We consider translation-invariant boundary conditions, in which case $\hat H$ admits a partial Fourier transform
$$\hat H = \int_{\RM}^\oplus \hat{H}_k\; d k.$$
Invertibility mod $\hsedgealgebra$ means that all $\hat{H}_k$ are Fredholm operators and invertible for large $\abs{k}$. Therefore the family admits a well-defined spectral flow $\SF(k\in \hat{H}_k)\in \ZM$ counting the number of eigenvalues passing through $0$ from below minus those which pass $0$ from above. The pairing of the winding number cocycle with $[\hat{H}]_1$ coincides with the {\it spectral flow} of a family of Fredholm operators \cite[Proposition 2.6]{Wahl}
$$\langle \mathrm{Wind}, 
[\hat{H}]_1\rangle = 
\SF(k\mapsto \hat{H}_k).$$
If the dependence of the eigenvalues of $\hat{H}_k$ around $0$ on $k$ is regular enough, e.g. differentiable, then the spectral flow counts precisely the signed number of eigenvalues passing through the spectral value $0$.

Boundary or interface effects will be, as in Section~\ref{sec:aff_interface}, obtained by restricting a bulk Hamiltonian $H$ to a half-space or to the region away from a fiducial interface line and extending the restriction to a self-adjoint operator through boundary or matching conditions. The Hatsugai relation follows from the K-theoretic equality
$[\hat{H}]_1 = \partial_+([H]_0)$, provided the latter is well defined.
We will go beyond this and consider relative bulk-edge correspondence as well as corrections induced by the boundary condition. In this context another numerical index arises, namely 
$$\sigma_{cor}:=\langle \mathrm{Wind}, [1+u v^*- e_+]_1\rangle,$$ 
the pairing of the winding number cocycle with the relative $K_1$-class given by the two von Neumann unitaries $u,v\in \mult(\hsedgealgebra)$ defining the boundary conditions (here $e_+=uu^*=vv^*$ is the projection to one of the deficiency subspaces which is finite-dimensional).

\subsection{Self-adjoint extensions using boundary triples}
According to von Neumann's theory, the self-adjoint extensions of a closed symmetric operator $ H$ with equal deficiency indices are parametrized by unitaries $u:\Nn_{-\imath}\to \Nn_{\imath}$ between the deficiency subspaces $\Nn_{\mp \imath}$ where $\Nn_z := \ker({H}^*-z)$ is the eigenspace to eigenvalue $z$ of the adjoint operator ${H}^*$ \cite{ReedSimon}. The self-adjoint extension $H_u$ defined by $u$ is the restriction of ${H}^*$ to the space (its domain)
$$  \mathrm{dom} H_{u} =\{\psi+u(f)+f :\psi\in \mathrm{dom} {H},f\in \Nn_{-\imath} \}.$$
If $ H$ is a one-dimensional differential operator which is the (graph) closure of an operator defined on $C_c^\infty(\Omega)$ for some open $\Omega\subset \RM$ then the self-adjoint extensions are equivalently defined by local boundary conditions, that is, by specifying linear relations between the values of the functions and their derivatives (which belong to the domain of ${H}^*$) at the boundary points of $\Omega$. We are interested in the two cases, 
$\Omega=\RM^{>0}$ and $\Omega = \RM\backslash\{0\}$, corresponding to the half-space and the interface extension, respectively. 
In simple cases, applying these linear relations to $\psi+u(f)+f$ can directly be solved for the von Neumann unitary $u$, but already the regularised Dirac operator is sufficiently complicated so that we make use of the powerful theory of boundary triples. We give a very brief overview of this referring the reader to  \cite{BruningGeylerPankrashkin} for the details.
\begin{definition}\label{def-bdr-triple}
Let $H$ be a symmetric operator on $\Hh$ with isomorphic deficiency subspaces. A boundary triple for $H$ is a triple $(V,\Gamma_1,\Gamma_2)$ where $V$ is an auxiliary Hilbert space and
$\Gamma_{1},\Gamma_2:\Hh\to V$ two linear maps such that, for all $\psi,\phi\in\mathrm{dom} H^*$ 
$$\langle \psi , H^* \phi\rangle_\Hh - \langle H^*\psi , \phi\rangle_\Hh =
\langle \Gamma_1\psi , \Gamma_2 \phi\rangle_V - \langle \Gamma_2\psi , \Gamma_1\phi\rangle_V$$
and moreover,
$\mathrm{dom} H^* \ni \psi \mapsto (\Gamma_1\psi,\Gamma_2\psi)\in V\oplus V$
is surjective.
\end{definition}
Given two  linear operators $A,B:V\to V$ such that $iA+B$ is invertible and $AB^*$ selfadjoint, the restriction $H_{A,B}$ of $H^*$ to all $\psi\in \mathrm{dom} H^*$ which satisfy
\begin{equation}\label{eq-bdry-cond0}
A\Gamma_1\psi=B\Gamma_2\psi
\end{equation}
is a self-adjoint extension. Boundary triples always exist, though are not unique; if we have one then all self-adjoint extensions of $H$ arise in the above way for different choices of $A$ and $B$. 

If $z$ lies in the resolvent set of $H$ then the restrictions $\Gamma_1(z)$,  $\Gamma_2(z)$  of $\Gamma_1$ and $\Gamma_2$ to $\Nn_z$ 
are bijections and we can define the so-called Krein function $Q(z):V\to V$ 
$$Q(z) := \Gamma_2(z)\Gamma_1(z)^{-1}.$$
Furthermore, if $\Im(z)\neq 0$ then 
$$W(z):= A-BQ(z)$$ 
is invertible (still under the requirement that $iA+B$ is invertible and $AB^*$ selfadjoint).
A direct calculation now yields that the von Neumann unitary $u_{A,B}:\Nn_{-\imath}\to \Nn_{\imath}$ describing the extension $H_{A,B}$ is given by
$$ u_{A,B} = -\Gamma_1(\imath)^{-1} 
W(\imath)^{-1} W(-\imath)\Gamma_1(-\imath) .$$
Indeed, if $f\in\Nn_{-\imath}$ then $F = u_{A,B}f + f \in \Nn_{-\imath} + \Nn_{\imath}$ satisfies $A\Gamma_1F = B\Gamma_2 F$ which can be solved for $u_{A, B}$. For $A=1$ and $B=0$ we obtain
$u_{1,0} =  -\Gamma_1(\imath)^{-1} \Gamma_1(-\imath)$
which will serve as our reference system. 
Thus $u_{A,B} {u_{1,0}}^*$ is conjugate to the unitary
$$ U_{A,B} := W(\imath)^{-1} W(-\imath)$$
on the auxiliary Hilbert space.
From Proposition~\ref{prop:res_aff_vn} one can conclude that if $H_{1,0}$ is resolvent affiliated to the half-space algebra (or the interface algebra) and $u_{1,0}$ in its multiplier algebra then $H_{A,B}$ is also resolvent affiliated to this algebra if and only if $u_{A,B}-u_{1,0}$ belongs to the edge algebra. In the context of translationally invariant differential operators on a two-dimensional half-space when we reduce the problem to a family of half line operators parametrized by the wave vector $k$ along the edge as described in the last section, we study their self-adjoint extensions through a family of boundary triples $(V(k),\Gamma_1(k),\Gamma_2(k))$ and matrices $A(k)$, $B(k)$. Now the dependence on $k$ needs extra care.

\begin{lemma}\label{lem-k-dep}
Let $V=V(k)$ be independent of $k$ and $k\in \RM \mapsto \Gamma_i(z, k)$ be norm-continuous for $i=1,2$ and \begin{equation}\label{eq-k-dep}
\sup_{k\in \RM}\norm{\Gamma_1(z,k)}\norm{\Gamma_1(\overline{z},k)^{-1}}<\infty
\end{equation}
with $z=\pm\imath$. Then $k\mapsto u_{A,B}(k)$ and $k\mapsto U_{A.B}(k)$ are norm continuous and 
$u_{A,B}(k)-u_{1,0}(k)\stackrel{k\to\pm\infty}\longrightarrow 0$ if and only if $U_{A,B}(k)\stackrel{k\to\pm\infty}\longrightarrow 1_V$.
\end{lemma}
\begin{proof}
The norm continuity of $k\mapsto u_{A,B}(k)$ follows directy from the assumption. 
We have 
$$u_{A,B}(k)-u_{1,0}(k) = -\|\Gamma_1(-\imath,k)\|\Gamma_1(\imath,k)^{-1} 
(U_{A,B}(k)-1_V)\frac{\Gamma_1(-\imath,k)}{\|\Gamma_1(-\imath,k)\|}$$
showing that $u_{A,B}(k)-u_{1,0}(k)$ is multiplication of $U_{A,B}(k)-1$ by two invertible operators which are uniformly bounded in $k$.
Therefore $U_{A,B}\to 1$ implies $u_{A,B}(k)-u_{1,0}(k)\to 0$. The converse follows in a similar way, as
$$ U_{A,B}(k)-1_V= -\Gamma_1(\imath,k) 
(u_{A,B}(k)-u_{1,0}(k))\Gamma_1(-\imath,k)^{-1}$$
\end{proof}
If $V$ is one-dimensional, the condition (\ref{eq-k-dep}) is always satisfied. More generally, 
as $V$ is finite-dimensional (\ref{eq-k-dep}) is satisfied if the matrix representation of $\Gamma_1(z,k)$ w.r.t.\ orthonormal bases of $\Nn_z$ and $V$ is asymptotically of the form $f(k) C(z)$ where $f$ is a continuous function and $C(z)$ a matrix which does not depend on $k$. 

The spectral values of $H_{A,B}$ which do not belong the spectrum of $H_{1,0}$ are the real values $\lambda$ in the resolvent set of $H_{1,0}$ for which $W(\lambda)$ is not invertible \cite{BruningGeylerPankrashkin}.
This will provide us with an implicit equation for the dispersion relation of the edge (or interface) modes.

\subsection{Laplacian}
\label{ssec:laplace}
The two-dimensional Laplacian is the unbounded operator
$$H = -\partial^2_x -\partial^2_y$$
on $L^2(\RM^2)$. Its spectrum is $[0,\infty)$.
It is strongly affiliated to $\bulkalgebra$, because the Fourier transform of $F(H)$ is $(k_x^2+k_y^2)(1+(k_x^2+k_y^2)^2)^{-\frac12}$ which tends to $1$ as $(k_x^2+k_y^2)\to+\infty$. 
Since $H$ has no bounded gaps there is no non-trivial bulk invariant. 
Nevertheless this operator constitutes an interesting example, because its edge invariants are not necessarily trivial.

\begin{remark}
One can add a covariant potential $V\in C_b(\Omega)$ described by a compact dynamical system $(\Omega,\alpha,\RM^2)$, 
so as to obtain a bulk Hamiltonian $H=-\partial^2 + V$ which is strongly affiliated to the enlarged algebra $\bulkalgebra = C(\Omega)\rtimes_\alpha \RM^2$ and may have bounded gaps in its spectrum. The spectral projections of $-\partial^2+V$ onto states with energies below a gap then define non-trivial elements of $K_0(\bulkalgebra)$. 
\end{remark}

\subsubsection{Affiliation to the half-space algebra and boundary conditions}
\newcommand{\hH}{\hat H}
Let $\mathring H$ be the restriction of $H$ to the core
$C_c^\infty(\RM\times\RM^{>0})$. We will study the local translation invariant boundary conditions
\begin{equation}\label{eq-bc-L}
K\Psi(x,0) + L \partial_x\Psi(x,0) + M \partial_y\Psi(x,0) = 0
\end{equation}
where $K,L,M$ are complex numbers which do not depend on $x$. Following Proposition~\ref{prop:extensions_fiberwise} and Lemma~\ref{lem-k-dep} we can investigate (resolvent-)affiliation once we have found at least one (resolvent-)affiliated self-adjoint extension. For that we note the following:
\begin{proposition}
\label{prop:reflection_principle}
The half-space Laplacians $\Delta_D$ and $\Delta_N$ with Dirichlet respectively Neumann boundary conditions are resolvent-affiliated to $\hsalgebra$ and bounded from below, hence also strongly affiliated.
\end{proposition}
\begin{proof}
Since they are norm-densely defined (cf. the proof of Proposition~\ref{prop:extensions_fiberwise}) it is enough to check that some resolvents of $\Delta_D$ and $\Delta_N$ are in $\hsalgebra$. Let $\Rr: L^2(\RM\times \RM_+)\to L^2(\RM\times \RM_+)$ be the reflection operator $(\Rr\Psi)(x_1,x_2)=\Psi(x_1,-x_2)$ and denote by $\Delta_b$ the bulk Laplacian on $\RM^2$. As noted above, one has $(\Delta_b+\imath)^{-1}\in \Aa_b$. By the reflection principle we can write
$$(\Delta_D +\imath )^{-1} = P_+ (\Delta_b+\imath )^{-1} (1-\Rr)P_+$$
and
$$(\Delta_N +\imath )^{-1} = P_+ (\Delta_b+\imath )^{-1} (1+\Rr)P_+$$
since the right-hand sides give the (unique) solutions of the boundary value problems. It is easy to see that the operator $P_+ (\Delta_b+\imath)^{-1} \Rr P_+$ is in $\hsedgealgebra$ since its integral kernel decays exponentially with the distance from the boundary, hence Lemma~\ref{lemma:nonsmoothelements}(iii) finishes the proof. 
\end{proof}
This covers the cases $(K,L,M)\in \{(1,0,0), (0,0,1)\}$ which can therefore serve as references. 

When Fourier transformed along the boundary we obtain the family
$$\mathring H(k)  = k^2 - \partial_y^2$$
on a dense domain of $L^2(\RM^{>0})$ containing differentiable functions which vanish at $y=0$. Furthermore, the boundary conditions become
$$K\Psi(k,0) - i k L \Psi(k,0) + M \Psi'(k,0) = 0$$
where prime denotes now the derivative in $y$ and we denoted the Fourier transform of $x\mapsto \Psi(x,y)$ by $k\mapsto \Psi(k,y)$. 
We study the self-adjoint extensions of $\mathring H(k)$ with the help of the triple $(V,\Gamma_1,\Gamma_2)$, $V=\CM$,
$$\Gamma_1 \Psi = \Psi(k,0),\quad \Gamma_2 \Psi = \Psi'(k,0).$$
The deficiency subspace $\Nn_z$ is spanned by a solution of 
$\Psi'' = (k^2-z)\Psi$
which we take to be
$$\Psi_z(k,y) = \exp(-\mu(k,z) y)$$
where 
$ \mu(k)=\sqrt[+]{k^2-z} $ (the square root with positive real part). We use for $V=\CM$ the canonical basis vector $1$. Then 
$\Gamma_1 \Psi_z = 1$ and $\Gamma_2 \Psi_z = -\mu(k,z)$.
Hence $Q(k,z)  =  -\mu(k,z)$ and
$$W(k,z) = A-BQ(k,z)= K- \imath kL-M\mu(k,z).$$ 
The condition that $A\bar B$ must be real implies that $M$, $K$ and $iL$ must have the same phase.
Without loss of generality we may suppose that they are real. If $M=0$ then we have Dirichlet boundary conditions and $L$ must be $0$ and $K \neq 0$, because otherwise there is $k$ such that $K-\imath kL$ is not invertible. 
\begin{lemma} The self-adjoint extension 
$\hat H_{K,L,M}$ 
is resolvent affiliated to $\hsalgebra$ if and only if $iL\pm M\neq 0$ or $M=0$ or $K\neq 0$.
\end{lemma}
\begin{proof} 
It is easily seen that the conditions of Lemma~\ref{lem-k-dep} are satisfied. This implies that $u_{1,0}(k)$ is norm continuous in $k$ and hence an element of $\hsalgebra$. As $\hat H_{1,0}$ is the Dirichlet Laplacian on the half-space, of which we know that it is resolvent affiliated to $\hsalgebra$, we can apply Proposition~\ref{prop:res_aff_vn} to see 
that $\hat H_{K,L,M}$ is resolvent affiliated if and only if 
$U_{K,L,M}\to 1$ as $|k|\to  \infty$. Explicitly we have 
$$U_{K,L,M} \sim_{|k|\to\infty}  \frac{ K - \imath k L  - M \mu(k,-\imath)}{ K - \imath k L  - M \mu(k,\imath)}.$$
Replacing $\mu$ with its asymptotic formula
$$\mu(k,z) \sim_{|k|\to\infty} |k| - \frac{z}{2|k|}$$ we get
$$U_{K,L,M} \sim_{|k|\to\infty}  \frac{ K - k(iL  + \mathrm{sgn}(k)M) - \frac{i}{2|k|} M }{ K - k(iL  + \mathrm{sgn}(k)M) + \frac{i}{2|k|} M}$$
If $iL\pm M\neq 0$ or $M=0$ or $K\neq 0$ then, clearly, $U_{K,L,M}\to 1$. If $K=0$ and 
$M\neq 0$ and $iL + M = 0$ then $U_{K,L,M}\to -1$ for $k\to +\infty$; and if $K=0$ and 
$M\neq 0$ and $iL - M = 0$ then $U_{K,L,M}\to -1$ for $k\to -\infty$.
\end{proof}

\subsubsection{Edge modes}

The extension $\hat H_{A,B}(k)$ has edge modes at energy $\lambda<k^2$ if 
$$A(k)+\mu(k,\lambda)B(k)=0.$$
In particular, there are no edge modes if $A(k)=0$, as $\mu(k,\lambda)$ has no zeros if $\lambda<k^2$, neither if $B(k)=0$, as then $A(k)$ must be non-zero. Otherwise we obtain the equation $\mu(k,\lambda) = -\frac{A(k)}{B(k)}$ which is
$$ \lambda = k^2-\frac{A^2(k)}{B^2(k)}$$
conditioned to $0 > \frac{A(k)}{B(k)}>-\infty$. 

With the above choice (\ref{eq-bc-L}) of local boundary conditions we get, if $M\neq 0$, 
$$\lambda(k) = \frac{-K^2 + k^2 (L^2+M^2) - 2iKkL}{M^2}, \quad +\infty>\frac{K+ikL}{-M}>0$$
Besides the Dirichlet case this is
$$\lambda(k) = k^2 - (K+ikL)^2, \quad {K+ikL}>0$$
and hence there is a flat band if $|L| = 1 \neq 0$ and $K=0$. This band is given by $\lambda(k) = 0$ and exists only for one sign of $k$, namely $ikL>0$. 

\subsubsection{Summary and bulk edge correspondence}
Table~\ref{tab-Laplace} contains the results besides the Dirichlet case for which the extension is resolvent affiliated. The extension is bounded from below if and only if there is no spectral flow. It is then strongly affiliated.

\begin{table}\label{tab-Laplace}
\begin{tabular}{c|c|c|c|c|c|c|c}
$K$ & $|L|$ & $\mathrm{sgn}(iL)$ & $\SF$ & wind & semibounded & res.-aff. \\
\hline
\hline
$\in\RM$ & 0 & n.a. &  0 & n.a. & yes & no \\
\hline
$\in\RM$ & $<1$ & $\pm 1$ & 0 & 0 & yes & yes \\
\hline
$\in\RM$ & $>1$ & $+ 1$ & $-1$ & $-1$ & no & yes\\
\hline
$\in\RM$ & $>1$ & $- 1$ & $+1$ & $+1$ & no & yes\\

\hline
$>0$ & 1 & $+1$  & $-1$ & $-1$ & no & yes\\
\hline
$>0$ & 1 & $-1$ & $+1$ & $+1$ & no & yes\\
\hline
$<0$ & 1 & $\pm 1$ & 0 & 0 & yes & yes\\
\hline
$0$ & 1 & $\pm 1$ & 0 & n.a. & yes & no\\
\end{tabular}
\medskip
	\caption{Spectral flow and winding number for the half-space Laplace operator $\hat H_{K,L,1}$ with $|L| \neq 1$ or $K\neq 0$. $\SF$ is the spectral flow of the edge modes when $k$ varies from $-\infty$ to $\infty$, wind is the winding number of the von Neumann unitary under the same variation of $k$.}
\end{table}
The results confirm our theory, namely that, 
except for the case in which $|M|=|L|$ and $K=0$, in which we do not have resolvent affiliation, 
the winding of $U_{A,B}$ is minus the difference of the spectral flows of the boundary modes of $\hat H_{A,B}$ and $\hat H_{1,0}$. This means that the corrected bulk boundary correspondence is also verified in this example.
\bigskip

The above analysis shows that, even though the bulk is topologically trivial, 
one can still have non-trivial edge invariants and these depend on the boundary conditions. In particular, the simple bulk edge correspondence cannot hold independently of the boundary conditions.
We observe that, if there is no spectral flow of the edge states then $\hat H$ is bounded from below and therefore, except in the case of a flat dispersion relation, strongly affiliated. The simple bulk-edge correspondence holds in this case giving the trivial Hatsugai relation $0=0$. In the other cases, the flat dispersion still excluded, the corrected bulk-edge correspondence holds: the winding number of the von Neumann unitary coincides with the spectral flow.

\subsection{The Dirac operator}
\label{ssec:massive_dirac} The Dirac operator with mass $m\in\RM$ is the matrix valued differential operator
\begin{equation}
	\label{eq:massivedirac_bulk}
	H_m = -\imath(\sigma_x \partial_x + \sigma_y \partial_y) + m\sigma_z 
\end{equation}
acting on $L^2(\RM^2)\otimes\CM^2$. Here we used the standard $2\times 2$ Pauli matrices. The spectrum is the complement of
$(-\abs{m},\abs{m})$, hence neither bounded above nor below. We assume that $m\neq 0$. As $H_m$ is  an elliptic differential operator it is resolvent affiliated.
The Fourier transform of $F(H_m)$ is
$$\frac{\sigma_x k_x + \sigma_y k_y + m \sigma_z}{\sqrt{1+k_x^2+k_y^2+ m^2}} = \begin{pmatrix}
	0 & \arg(k_x+\imath k_y) \\
	\arg(k_x-\imath k_y) & 0
\end{pmatrix} + O((k_x^2+k_y^2)^{-\frac12})$$
which tends for large wave vectors to a direction-dependent matrix. This shows that $H_m$ is $\Bb$-affiliated to $M_2(\bulkalgebra)$ where $\Bb$ is the sub-algebra $M_2(C(\overline{\RM^2}))$ of the multiplier algebra of $M_2(\CM\rtimes \RM^2)\cong M_2(C_0(\RM^2))$ obtained when radially compactifying $\RM^2$ to the closed disk $\overline{\RM^2}$ by adding the circle $\SM^1$ of directions at infinity of $\RM^2$.  But $H_m$ is not strongly affiliated to $\hsalgebra$ because the above matrix depends on $k$. Therefore the spectral projection $P_{\leq 0}(H_m)$ is not an element of $M_2(\bulkalgebra^\sim)\cong M_2(C(\SM^2))$. In particular, it does not define an element of $K_0(\bulkalgebra)$. 
Technically, $P_{\leq 0}(H_m)$ defines a non-trivial element of $K_0(\Bb)$ at least, however, since the disk is contractible one has $K_0(\Bb)\simeq \ZM$ yet the integer is not the Chern number but just the rank of the projection as a vector bundle over the disk. As a consequence, one cannot formulate a meaningful bulk-edge correspondence for a single Dirac Hamiltonian but must go over to the relative theory. 
\begin{lemma}
Any two Dirac Hamiltonians with masses $m_1,m_2\neq 0$ are invertible and $M_2(\bulkalgebra)$-comparable, therefore have a well-defined relative bulk invariant $[H_{m_1},H_{m_2}]_0=[p_{m_1}, p_{m_2}]_0 \in K_0(\bulkalgebra)$ with $p_{m}=P_{\leq 0}(H_{m})$.

Under the isomorphism $K_0(\bulkalgebra)\simeq \ZM$ induced by the Chern number this class is given by
$$\langle \Ch_{2}, [p_{m_1},p_{m_2}]_0\rangle =
		\frac{1}{2}(\sgn(m_1)-\sgn(m_2)).$$
\end{lemma}
\begin{proof}
Since $(H_m+\imath)^{-1}\in M_2(\Aa_b)$ and $H_{m_1}-H_{m_2}=(m_1-m_2)\otimes \sigma_3$ is bounded one has $M_2(\bulkalgebra)$-comparability.

To compute the pairing of the Chern cocycle with this class, we follow the recipe of Appendix~\ref{sec:numerical_pairing} and extend the Chern cocyle to a pair algebra.

Specifically, as in Example~\ref{ex:sphere_cocycle} we identify the upper half-sphere $H\SM^2=\SM^2\cap (\RM^{2}\times \RM_+)$ with the radial compactification $\overline{\RM^2}$ of $\RM^2$. The boundary $\partial H\SM^2\simeq \mathbb{S}^1$ is exactly the circle at infinity of $\RM^2$. In this sense one has $p_{m}\in C(H\mathbb{S}^2)\otimes M_2(\CM)$ and $\left(p_{m_1}-p_{m_2}\right)\rvert_{\partial H\mathbb{S}^2}=0$. 

As seen in Example ~\ref{sec:numerical_pairing} we can then compute the pairing with the Chern cocycle as an integral over the sphere
\begin{equation}
	\begin{split}
		\label{eq-ch-rel}
		\langle &\Ch_{2}, [p_{m_1},p_{m_2}]_0\rangle = \frac{1}{2\pi\imath} \int_{\SM^2} \hat{\omega}_{m_1,m_2} \\
		&= \frac{1}{2\pi \imath}\int_{\RM^2} \Tr(p_{m_1}dp_{m_1}dp_{m_1}-p_{m_2}dp_{m_2}dp_{m_2})  \\
		&= \frac{1}{4\pi} \int_{\RM^2} \left(\frac{m_1}{\sqrt{m_1^2 + k^2}}-\frac{m_2}{\sqrt{m_2^2 + k^2}}\right) dk\\
        &=\frac{1}{2}(\sgn(m_1)-\sgn(m_2)) 
	\end{split}
\end{equation}
where $\hat{\omega}_{m_1,m_2}$ is the unique differential form on $\SM^2$ which pulls back to the form $\Tr(p_{m_i}dp_{m_i}dp_{m_i})$ on the upper ($i=1$) and the lower ($i=2$) half-sphere respectively. The integral over $\RM^2$ can be solved analytically. The fact that the contribution of either projection individually is not an integer reiterates that $p_{m_i}$ fails to define a class in $K_0(\RM^2)$.
\end{proof}
\begin{remark}
The above K-theoretic arguments justify exactly the prescription of \cite{Bal21} which directly defines the relative bulk invariant by gluing the Fermi projections on a sphere along the equator.
\end{remark}

\subsubsection{Extension to the interface}
We now study the self-adjoint extensions arising from local matching conditions at an interface defined by $y=0$. For that we restrict the Dirac operator to the core of matrix valued functions which together with their derivatives vanish on the line defined by $y=0$. 
After a Bloch transform along the boundary it decomposes into the family of symmetric operators
$$\mathring H(k) = \begin{pmatrix} m  & k + \partial_y \\ k - \partial_y & -m \end{pmatrix}= \sigma_x k + Y\partial_y + m \sigma_z,$$
$Y = \begin{pmatrix} 0 & 1 \\ -1 & 0 \end{pmatrix}$.
For $z$ in the resolvent set, the kernel $\ker (H^*(k) -z)$ is spanned by the two vectors 
$$\Psi_{\alpha,z}(k,y) = \phi_{\alpha,z}(k) \exp(-\alpha\mu(k)y)\chi_{\alpha y>0}(y),\quad \phi_{\alpha,z}(k) = 
 \begin{pmatrix} m+z \\ k+\alpha\mu(k,z)  \end{pmatrix} $$ 
where $\alpha \in\{\pm\}$ and 
$$\mu(k) = \sqrt{k^2+m^2-z^2}$$
(square root with positive real part). If $\alpha=+$ then $\Psi_{\alpha,z}$  vanishes on the left and decays exponentially to the right 
whereas for $\alpha=-$ it is the other way around. 

To study the self-adjoint extensions of $\mathring H(k)$ we consider the boundary triple $(\Gamma_1,\Gamma_2,\CM^2)$ given on $\Psi\in \mathrm{dom} H(k)^*$ by 
\newcommand{\op}{0^+}
\newcommand{\om}{0^-}
$$\Gamma_1 \Psi = -\Psi(\op) + \Psi(\om) ,\quad \Gamma_2 \Psi = \frac12 Y(\Psi(\op) + \Psi(\om)).$$
Local matching conditions which are invariant under translation along the boundary are then determined by two $2\times 2$-matrices $A$, $B$ and given in the form 
$A\Gamma_1\Psi = B\Gamma_2 \Psi$. These matrices must satisfy that $A+iB$ is invertible and $AB^*$ self-adjoint. 
For example, the conditions $\Psi(\op) = \Psi(\om)$ correspond to 
$A=1$ and $B=0$. The corresponding self-adjoint extension $\hat H_{1,0}$
coincides with the original Dirac operator on the plane and so these boundary conditions are called the transparent boundary conditions. 

When using the orthonomal basis $(\hat\Psi_{+,z},\hat\Psi_{-,z})$ of normalized eigenfunctions 
 for $\Nn(z)$ and the canonical basis for $V$ we obtain the following matrix representation $\Gamma_1(z)$ and $\Gamma_2(z)$:

$$\Gamma_1(z,k) = \begin{pmatrix} -(m+z) & (m+z)\\ -(k + \mu) & (k-\mu)\end{pmatrix}
\begin{pmatrix}
    \|\Psi_{+,z}(k,\cdot)\|^{-1} & 0 \\ 0 & \|\Psi_{-,z}(k,\cdot)\|^{-1}
\end{pmatrix}$$
$$ \Gamma_2(z,k) = \frac12 \begin{pmatrix} (k+\mu) & (k-\mu)\\ -(m+z) & -(m+z)\end{pmatrix}
\begin{pmatrix}
  \|\Psi_{+,z}(k,\cdot)\|^{-1} & 0 \\ 0 & \|\Psi_{-,z}(k,\cdot)\|^{-1}
\end{pmatrix}$$
Hence 
$Q(z) = \Gamma_2(z)\Gamma_1(z)^{-1}$
is given by
$$Q(z) = 
-\frac1{2 \mu} \begin{pmatrix} m-z & k \\ k & -m-z\end{pmatrix}.$$
\begin{lemma} \label{lem-Dirac}
The self-adjoint extension $\hat H_{A,B}$ defined by $A$ and $B$ is resolvent affiliated to $M_2(\interfacealgebra)$ if and only if 
$\det(A\pm \frac12 B \sigma_x)\neq 0$. 
\end{lemma}
\begin{proof} 
We have $$\Gamma_1(z,k)\sim_{ +\infty} \frac1{\sqrt{2k}}\begin{pmatrix}
0 & m+z \\ 1 & 0    
\end{pmatrix},\quad
\Gamma_1(z,k)\sim_{ -\infty} \frac1{\sqrt{2k}}\begin{pmatrix}
-(m+z) & 0 \\ 0 & -1   
\end{pmatrix}$$ from which we conclude that the conditions of Lemma~\ref{lem-k-dep} are satisfied and so in particular $u_{1,0}\in M_2(\hsalgebra)$. As $\hat H_{1,0}$ is the Dirac operator on the bulk, of which we know that it is resolvent affiliated to $M_2(\bulkalgebra)$ and hence also to $M_2(\interfacealgebra)$, we can apply Proposition~\ref{prop:res_aff_vn} and again Lemma~\ref{lem-k-dep} to see 
that $\hat H_{A,B}$ is resolvent affiliated to $M_2(\interfacealgebra)$ if and only if 
$$W(k,\imath)^{-1} W(k,-\imath)\stackrel{|k|\to +\infty}\longrightarrow 1$$
where $W(k,z) = A -BQ(z)$ is given by 
$$W(k,z) =  A+\frac{\mathrm{sgn}(k)}2 B \sigma_x - \frac1{2\mu} B(\mathrm{sgn}(k)(\mu-|k|)\sigma_x+z 1_2-m\sigma_z).$$
If $A\pm \frac12 B \sigma_x$ is invertible then $W(k,z) \sim_{|k|\to \infty} A + \frac{\mathrm{sgn}(k)}2 B \sigma_x$ which is independent of $z$ and hence $W(k,\imath)^{-1} W(k,-\imath) \stackrel{|k|\to \infty}\longrightarrow 1$.
If $A + \frac12 B \sigma_x$ has a zero eigenvalue with eigenvector $v$ then, for $k>0$, $W(k,z) v = \frac1{2\mu}B \tilde v$ where 
$$\tilde v = ((\mu-|k|)\sigma_x+m\sigma_z+z 1_2) v=:\tilde v_0 + O(k^{-1}),
\quad \tilde v_0 = (m\sigma_z+z 1_2) v.
$$
As $ \frac1{2\mu}W(k,z)^{-1} B \tilde v = v$ we see that the norm of $\frac1{2\mu}W(k,\imath)^{-1} B$ remains bounded from below by $\frac{\|v\|}{\|\tilde v_0\|} + O(k^{-1})$ ($\tilde v_0$ does not vanish). 
Therefore,
$$W(k,\imath)^{-1} W(k,-\imath) = 1 - W(k,\imath)^{-1}(W(k,\imath)-W(k,-\imath)) =
1 - \frac{\imath}{\mu} W(k,\imath)^{-1} B $$
cannot tend to $1$ as $k\to +\infty$. 
If $A - \frac12 B \sigma_x$ has a zero eigenvalue a similar argument shows that $W(k,\imath)^{-1} W(k,-\imath)$ cannot tend to $1$ as $k\to -\infty$.
\end{proof}
For example, the boundary conditions
$$\Psi_1(\op) = \Psi_1(\om)=0 $$
lead to a self-adjoint extension which is not resolvent affiliated, as these correspond to $A=\begin{pmatrix}
1 & 0 \\ 0 & 0    
\end{pmatrix}$ and 
$B=\begin{pmatrix}
0 & 0 \\ 0 & 1    
\end{pmatrix}$.

\subsubsection{Decoupling matching conditions}
Decoupling matching conditions are given by 
$$\Psi_1(0^\pm) = \ra^\pm \Psi_2(0^\pm),\quad \ra^\pm\in\RM\cup{\infty}$$
(one point compactification).  
They correspond to boundary conditions for the half-spaces defined by $y>0$ and $y<0$ separately.
We can rewrite this as $A\Gamma_1(\Psi) = B\Gamma_2(\Psi)$ with
$$A = \frac12\begin{pmatrix} -1 & \ra^+ \\ 1 & -\ra^-\end{pmatrix}
, \quad B= \begin{pmatrix}  \ra^+&1 \\ \ra^- & 1 \end{pmatrix}
.$$
One checks that $AB^*$ is self-adjoint and $A+iB$ invertible. Furthermore, 
$$A+\frac12 B\sigma_x = \begin{pmatrix} 0 & \ra^+ \\ 1 & 0
\end{pmatrix},\quad A-\frac12 B\sigma_x = \begin{pmatrix} -1 & 0 \\ 0 & -\ra^-
\end{pmatrix}$$
from which we conclude by the Lemma~\ref{lem-Dirac} that decoupling conditions lead to 
resolvent affiliated self-adjoint extensions if and only if $\ra^\pm\neq \{0,\infty\}$. We now obtain
\begin{eqnarray*}
A-BQ(z) 
&=& \frac1{-2\mu} \begin{pmatrix}
\mu-k-\ra^+(m-z) & -\ra^+(\mu+k)+m+z \\
-\mu-k-\ra^-(m-z) & \ra^-(\mu-k)+m+z
\end{pmatrix} \\
& = &\frac1{-2\mu} 
\begin{pmatrix}
\frac{\mu-k}{m-z}-\ra^+ & 0 \\
0 & \frac{\mu+k}{m-z}+\ra^-
\end{pmatrix}
\begin{pmatrix}
(m-z) & (\mu+k) \\
-(m-z) & (\mu-k) 
\end{pmatrix}
\end{eqnarray*}
(we used that $\frac{\mu-k}{m-z} = \frac{m+z}{\mu+k}$ and $\frac{\mu+k}{m-z}=\frac{m+z}{\mu-k}$). Its determinant is
$$\det (A-BQ(z)) = \frac1{2\mu(m-z)} ((\mu-k)-\ra^+(m-z))((\mu+k)+\ra^-(m-z)) $$
which is a product of three factors of which the second depends only on $\ra^+$ whereas the third depends on $\ra^-$. As a consequence, the winding number of the von Neumann unitary is a sum of two winding numbers corresponding to the extensions of the two different half-spaces. This is to be expected, as the matching conditions decouple the two half-spaces. Moreover, the fact that the two factors depend on $k$ with opposite sign implies that these winding numbers contribute in an opposite way. 

Let us compute the winding of the von Neumann unitary for the upper half-space $y>0$. Since the transparent matching conditions are not decoupled we choose a new reference system, namely $\ra^+ = 1$ (this is $\Psi_1(0^+) = \Psi_2(0^+)$ and corresponds to infinite mass boundary conditions). Then the relative von Neumann unitary becomes
\begin{eqnarray*}
U^{\ra^+}(k)&=&\frac{(\mu-k)-\ra^+(m-\imath)}{(\mu-k)-\ra^+(m+\imath)}\,\frac{(\mu-k)-(m+\imath)}{(\mu-k)-(m-\imath)}\\
&=& 
\frac{\mu-k-(\ra^++1)m + {\ra^+}\frac{m^2+1}{\mu-k} +(\ra^+-1)\imath}
{\mu-k-(\ra^++1)m + {\ra^+}\frac{m^2+1}{\mu-k} -(\ra^+-1)\imath}
\end{eqnarray*}
from which we infer that 
the numerator goes from $\infty$ to $\mathrm{sgn}(\ra^+)\infty$ when $k$ varies from $-\infty$ to $+\infty$ and since its imaginary part is constant, we see that 
the winding number is $0$ if $\ra^+>0$ and $-1$ if $\ra^+<0$. 

\subsubsection{Edge modes for decoupling conditions}
The edge modes have dispersion relation determined by $\det(A-BQ(\lambda)) = 0$ where $\lambda$ is now real but outside the spectrum of the reference extension. The equation is equivalent to 
$$\mu(k,\lambda)-k = \ra^+(m-\lambda)\quad \mbox{or} \quad \mu(k,\lambda)+k = \ra^-(m-\lambda).$$
The first relation determines the dispersion relation for the edge modes of the upper half plane and the second for the edge modes of the lower one.
We focus on the first relation which we can write ($\ra=\ra^+$)
$$m^2+k^2-\lambda^2 = \ra^2(m-\lambda)^2 +k^2 +2\ra (m-\lambda)k,\quad \ra(m-\lambda)+k>0 .$$
Parametrizing $\epsilon = \mathrm{sgn}(\ra)$ and $\ra = \epsilon e^t$ the solutions is $$\lambda = m\tanh t + k \frac{\epsilon}{\cosh t} ,\quad k {\sinh t} < {m\epsilon}$$
In particular, if $\ra = 0$, then $\lambda(k) = -m$ and $-k<0$, whereas if $|\ra| = \infty$, then $\lambda(k) = m$ and $k<0$. 
These are flat dispersion relations which explains why the corresponding self-adjoint extension is not resolvent affiliated. 
If $|\ra|=1$, that is, $t=0$, then $m$ and $\ra$ must have the same sign and hence $\lambda(k) = m k$, a whole line which does not touch the bulk spectrum. If $t\neq 0$ then $k_* = \frac{m\epsilon}{\sinh t}$ is well defined and 
$$\lambda(k_*) = m (\tanh t + \frac1{\sinh t\cosh t}) = m \sqrt{1+\frac{k_*^2}{m^2}}. $$
This means that the mode  touches the bulk band at $k_*$. 
The inequality  $k\sinh t > m\epsilon$ implies that the spectral curve is not a complete straight line (except if $t=0$) but only a half-line attached to $\lambda(k_*)$. Depending on the mass $m$ and the boundary condition $\ra= \frac{\Psi_1(0)}{\Psi_2(0)}$ this half line may or may not cross the fiducial line at energy $0$ with positive or negative slope. In this way we can determine the spectral flow.  
The results are summarized in Table~\ref{tab-Dirac}. 
\begin{table}
\begin{tabular}{|c|c|c|c|c|c|}
\hline
$\mathrm{sgn} (m)$ & $|\ra|$ & $\mathrm{sgn}(\ra)$ &  wind  & $\SF$ & touched bulk band\\
\hline
\hline
$+$1 & 1 & $+$1 &  0 & $+$1 &  none, $k\in\RM$ \\
\hline
$+$1 & $>1$ & $+$1 &  0 & $+$1 & upper, $k\in (-\infty,k_*)$ \\ 
\hline
$+$1 & $<1$ & $+$1 &  0 & $+1$ & lower, $k\in (k_*,+\infty)$ \\ 
\hline
\hline
$+$1 & $>1$ & $-$1 &  $-$1 & 0 & upper, $k\in (k_*,+\infty)$ \\ 
\hline
$+$1 & $<1$ & $-$1 &  $-$1 & 0 & lower, $k\in (-\infty,k_*)$ \\ 
\hline
\hline
$-$1 & 1 & $-$1 & $-$1 & $-$1 &   none, $k\in\RM$ \\
\hline
$-$1 & $>1$ & $-$1 & $-$1 & $-$1 &  upper, $k\in (-\infty,k_*)$ \\ 
\hline
$-$1 & $<1$ & $-$1 & $-$1 & $-$1 &  lower, $k\in (k_*,+\infty)$ \\ 
\hline
\hline
$-$1 & $>1$ & $+$1 & 0 & 0 &  upper, $k\in (k_*,+\infty)$ \\ 
\hline
$-$1 & $<1$ & $+$1 & 0 & 0 &  lower, $k\in (-\infty,k_*)$ \\ 
\hline
\end{tabular}
\medskip

\caption{Spectral flow and winding number for the half-space Dirac operator $\hat H_\ra$ with $\ra\neq 0,\infty$. 
Connected components in the parameters are separated with double lines.}
\label{tab-Dirac}
\end{table}

\begin{remark}
This family of self-adjoint extensions and their spectral flows were previously also computed in \cite{GL} in a slightly different parametrization.
\end{remark}

\subsubsection{Relative bulk edge correspondence}
Since $H_m$ is not strongly affiliated to the bulk algebra it does not define a $K_0$-class of that algebra and we do not have a simple bulk edge correspondence. But we can employ the relative bulk edge correspondence, or the correspondence for an interface model. 
For $m,m'\neq 0$ and $\ra\notin\{0,\infty\}$ we can apply Theorem~\ref{th:bulk_edge_pluscorrection} to conclude
$$[\hat H_{m,\ra}]_1 = [\hat H_{m',a'}]_1 +  \partial [H_m,H_{m'}]_0 +  [u_\ra u_{\ra'}^*]_1$$
(here $\hat H_{m,\ra}$ is the half-space extension of $H_m$ defined by $\ra$).
The winding number cocycle applied to $\partial [H_m,H_{m'}]_0$ is equal to the Chern number of $[H_m,H_{m'}]_0$. Applied to $[\hat H_{m,\ra},\hat H_{m',\ra'}]_1$ it gives the difference of the spectral flows  
of the edge states of the two operators. As $\SF(k\to \hat H_{m=-1,\ra'=1}(k))=0$
the above equation yields 
\begin{equation}
\label{eq:SF_Dirac}
\SF(k\!\to\!\hat H_{m,\ra}(k))=  
\frac{\sgn(m)+1}2 + \mathrm{wind}(U_\ra^*,U_\ra) = \frac{\sgn(m)+\sgn(a)}2
\end{equation}
which is in agreement with the results of Table~\ref{tab-Dirac}.

\subsubsection{Bulk-interface correspondence}

When we apply Theorem~\ref{th:bulk_interface_pluscorrection} we find: 
\begin{proposition}
Consider the self-adjoint differential operator $$\hat{H}_{m,A,B}(k) = \sigma_x k + Y\partial_y + + m(X_2) \sigma_z$$
subject to a matching condition defined by $A,B$ for which $\hat{H}_{m,A,B}$ is resolvent affiliated as described by Lemma~\ref{lem-Dirac} and with $m$ a piecewise continuous function on $\RM$ which has finitely many singularities and admits finite limits $m_\pm=\lim_{x_2\to\infty}m(x_2)$.

One has \begin{equation}
\label{eq:interface_dirac}
\SF(k\mapsto \hat{H}_{m,A,B}(k)) = \frac{\sgn(m_+)-\sgn(m_-)}2 + \mathrm{Wind}(U_{A,B}^*,U_{A,B})
\end{equation}
\end{proposition}
The correction $\mathrm{Wind}(U_{A,B}^*,U_{A,B})$ vanishes for the transparent matching conditions. As mentioned before, this is in line with the bulk-difference-interface correspondence from \cite{Bal19}. For decoupling matching conditions one easily finds that the correction is non-trivial whenever $\sgn(a^+)\neq \sgn(a^-)$. Within the $\interfacealgebra$-resolvent affiliated boundary conditions the connected components of the decoupled matching conditions are not isolated, i.e. there do exist matching conditions which couple the two half-spaces but have non-trivial contributions to the spectral flow. 
\begin{remark}
The spectral flow is not generally equal to the number of interface modes since there can be cancellations; in the interpretation as interface conductance this is clear since some of the interface states may propagate in opposite directions. 
\end{remark}

\begin{remark}
The infinite mass boundary conditions $a^+ = \pm 1$ can be obtained from the interface model in the limiting case where the Dirac mass $m(x_d)$ tends to $\mp \infty$ for all $x_d < 0$ \cite{BCTS} while maintaining transparent matching conditions. In that limit \eqref{eq:SF_Dirac} and \eqref{eq:interface_dirac} are consistent.
\end{remark}

\subsection{Regularised Dirac operator}
\newcommand{\al}{{\alpha}}
\newcommand{\be}{{\textcolor{blue}{\beta}}}

For the massive Dirac operator the Chern number the spectral projection onto the negative energy states is not a well-defined topological invariant. To avoid this problem one can regularize the Dirac operatorby adding a second-order term $\epsilon (\partial_x^2+\partial_y^2) \sigma_z$ with some small $\epsilon\neq 0$ which we suppose to be smaller in size than $\frac{1}{2}|m|$. This operator is referred to as the regularized (massive) Dirac operator and it is a continuum model for a Chern insulator. On a half-space the spectral flow can be seen to depend on the choice of boundary conditions \cite{TDV20, JudTauber2025}, which we will study using our formalism in this section.

The regularised Dirac operator is the matrix valued operator
\begin{equation}
	\label{eq:reg_massivedirac_bulk}H_{m,\epsilon} = -\imath (\sigma_x \partial_x  + \sigma_y \partial_y) + m \sigma_z - \epsilon (\partial_x^2+\partial_y^2) \sigma_z .
\end{equation}
If $\epsilon$ is small enough then 
the spectrum does not change under the regularisation, it is the complement of  $(-\abs{m},\abs{m})$. The effect of adding this term is that  $H_{m,\epsilon}$ is now strongly affiliated to $\bulkalgebra$. This can be checked by explicitly computing the asymptotic behavior of the Fourier transform of $F(H_{m,\epsilon})$, or simply by noting that the Hamiltonian is a relatively compact perturbation of the strongly affiliated Hamiltonian $\epsilon (\partial_x^2+\partial_y^2) \sigma_z$ in the sense of Proposition~\ref{prop-res-perturb}.  
As a consequence, the spectral projection $P_{\leq 0}(H_{m,\epsilon})$ defines an element of $K_0(\bulkalgebra)$ and its Chern number is well-defined; it is 
\begin{equation}\label{eq-bulk-ch}
\sigma_b=\langle \Ch_{2}, [H_{m,\epsilon}]_0\rangle = \frac{1}{2} (\sgn(m) - \sgn(\epsilon))
\end{equation}
as one can compute by direct evaluation of the integral formula.

\subsubsection{Affiliation to the half-space algebra and boundary conditions} Let $\mathring{H}$ be the closure of the symmetric operator obtained from restricting $H$ to $y>0$, that is, to the core $C_c^{\infty}(\RM\times\RM_+)$.  We are interested to know which boundary conditions furnish us a self-adjoint extension which is resolvent affiliated to the half-space algebra.
We consider here only local translation invariant boundary conditions:
\begin{equation}\label{eq-bc}
K\Psi(x,0) + L \partial_{x}\Psi(x,0) + M \partial_{y}\Psi(x,y) = 0 \end{equation}
where $K,L,M$ are constant complex $2\times 2$ matrices (with some restrictions owing to self-adjointness, see below). 
When performing a Bloch decomposition along the boundary we obtain a family of symmetric Hamiltonians
$$\mathring H(k) = \begin{pmatrix} m + \epsilon (k^2-\partial_y^2) & k + \partial_y \\ k - \partial_y & -m-\epsilon (k^2-\partial_y^2) \end{pmatrix}
$$
and the boundary conditions can be written
\begin{equation}\label{eq-bc-regD}
K \Psi(k,0) - i k L \Psi(k,0) + M \Psi'(k,0) = 0.
\end{equation}
\begin{proposition}
\label{prop:dd_and_dn_affiliation}
The self-adjoint extensions $\hat{H}_{D,D}:=\hat{H}_{\one_2,0,0}$ (Dirichlet-Dirichlet) and $\hat{H}_{D,N}:=\hat{H}_{\diag(1,0),0_2,\diag(0,1)}$ (Dirichlet-Neumann) are resolvent affiliated to $M_2(\hsalgebra)$ and strongly affiliated to $\hsalgebra$.
\end{proposition}
\begin{proof}
Denote $\Delta_D$ and $\Delta_N$ the Dirichlet respectively Neumann Laplacian on $\RM\times \RM_+$. Since those are resolvent-affiliated to $\hsalgebra$ it is easy to show directly that $\hat{H}^{(0)}_{D,D} := -\epsilon \sigma_3 \otimes \Delta_D$ and $\hat{H}_{D,N}^{(0)} := -\diag(\epsilon \Delta_D, -\epsilon \Delta_N)$ are $\hsalgebra$-resolvent-affiliated and moreover strongly affiliated since $$F(\hat{H}_{D,*}^{(0)})=-\diag(F(\epsilon \Delta_D), -F(\epsilon \Delta_*)).$$ One can check that $\hat{V}_{D,*}:=\hat{H}^{(0)}_{D,*}-\hat{H}_{D,*}$ is a first-order differential operator which is symmetric on the given domain for both choices $*\in \{D,N\}$. With 
$$(\hat{H}^{(0)}_{D,*}+\imath)^{-1}= -\diag\left( (\epsilon \Delta_D-\imath)^{-1}, -(\epsilon \Delta_*+\imath)^{-1}\right)$$
and the explicit expressions of Proposition~\ref{prop:reflection_principle} for the resolvents it is simple to verify that $\hat{V}_{D,*}(\hat{H}^{(0)}_{D,*}+\imath)^{-1}\in M_2(\CM)\otimes \hsalgebra$. Therefore, the conditions of Proposition~\ref{prop-res-perturb} are satisfied.
\end{proof}
\begin{remark}
For general boundary conditions this perturbation argument fails since the first-order terms are generally neither symmetric on their own nor relatively bounded.
\end{remark}
It is convenient to set up the boundary triple in such a way that the Dirichlet Hamiltonian $\hat H_{1,0,0}$ corresponds the reference system $(A,B)=(1,0)$.

To derive an expression for the von Neumann unitary $U_{K,L,M}(k)$
we use the boundary triple $(\CM^2,\Gamma_1,\Gamma_2)$,  
$$\Gamma_1\Psi = \Psi(0), \quad \Gamma_2 = -\frac12 Y\Gamma_1 -\epsilon \sigma_z \Gamma_1',\quad \Gamma_1'\Psi = \Psi'(0)$$ where $Y=\begin{pmatrix} 0 & 1 \\ -1 & 0 \end{pmatrix}$. 
Now $A\Gamma_1\Psi = B \Gamma_2\Psi$ corrresponds to (\ref{eq-bc-regD}) if $A +B\frac12Y=K-\imath kL$ and 
$\epsilon B\sigma_z = M$. Recall that  $U_{K,L,M}(k)=W(k,\imath)^{-1} W(k,\imath)$ where 
$W(k,z) = A - B \Gamma_2(k,z) \Gamma_1^{-1} (k,z)$ which we can write as  
$$W(k,\imath) 
= K-\imath kL +
M\Gamma_1'(k,\imath) \Gamma_1^{-1}(k,\imath).$$ 
To obtain an expression for $W(k,z)$ we solve $\mathring{H}^*\Psi = z\Psi$ for $z$ in the resolvent set of $\mathring H$. We obtain the two independent solutions
\begin{equation}\label{eq-base}
\Psi_{\alpha,z}(x)  = \phi_{\alpha,z}  \exp^{-\mu_{\alpha}x},
\end{equation}
$\alpha = \pm$, where 
\begin{equation}\label{eq-sol1}
\phi_{+,z} = \begin{pmatrix} \nu_++z \\ \mu_++k  \end{pmatrix},
\quad 
\phi_{-,z} = \begin{pmatrix} \nu_--z \\ \mu_--k  \end{pmatrix},
\end{equation}
with $\mu_{\alpha}(k,z) = \sqrt{k^2+\zeta_\alpha(z)}, 
\nu_\alpha(z) = m -\epsilon \zeta_{\alpha}(z)$ and 
$$\zeta_{\pm}(z)  = \frac1{2\epsilon^2}\left( 1+2m\epsilon \,\pm\, \sqrt{1+4m\epsilon +4\epsilon^2 z^2}\right) .$$
The square root taken in the definition of $\mu_\alpha$ is that with positive real part. 
Here we assume that $\epsilon$ is small enough so that $\zeta_{\pm}(z)$ is positive real number if $z=\pm \imath$. 
Taking the above two solutions as a basis for $\Nn(z)$ and the canonical basis for $\CM^2$ we obtain the  matrix expressions for $\Gamma_1$ and $\Gamma_1'$, $\Gamma_1 = (\phi_{+,z}\:  \phi_{-,z})$, $\Gamma'_1 = -(\mu_+ \phi_{+,z}\:  \mu_-\phi_{-,z})$. This leads to 
$$\Gamma_1'\Gamma_1^{-1} = \Gamma_1\begin{pmatrix}
-\mu_+ & 0 \\ 0 & -\mu_-\end{pmatrix}\Gamma_1^{-1} = -\frac{\mu_++\mu_-}2 1_2 - \frac{\mu_+-\mu_-}2 \Gamma_1 \sigma_z \Gamma_1^{-1},$$
\begin{eqnarray*}
\Gamma_1 \sigma_z \Gamma_1^{-1} & = & \frac1{\tilde\nu_+ \tilde \nu_- - \tilde\mu_+ \tilde \mu_-} 
\begin{pmatrix} \tilde \nu_+ \tilde\nu_- + \tilde \mu_- \tilde\mu_+& -2\tilde \nu_+ \tilde \mu_-  \\ 2\tilde\mu_+\tilde\nu_-  &  - \tilde \nu_+ \tilde\nu_- - \tilde \mu_- \tilde\mu_+  \end{pmatrix} 
\end{eqnarray*}
where $\tilde \mu_\alpha = \mu_\alpha + \alpha k$ and $\tilde \nu_\alpha = \nu_\alpha + \alpha z$.

\begin{lemma}
If $\det(iL\pm M)\neq 0$ then
the extension $\hat H_{K,L,M}$ is resolvent affiliated to $M_2(\hsalgebra)$. 
\end{lemma}
\begin{proof} While the above matrix representation of $\Gamma_1(z,k)$ is convenient to express $W(z,k)$,  it is not w.r.t.\ to an orthonormal basis of $\Nn_z$. To verify condition (\ref{eq-k-dep}) we need a matrix representation in an orthonormal basis of $\Nn_z$. We can write such a base as $(\Psi_{+,z}(k,\cdot),\Psi_{-,z}(k,\cdot))\mathcal P(z,k)$ with an appropriate $2\times 2$ matrix $\mathcal P(z,k)$. Then the matrix representation of $\Gamma_1$
is given by $(\phi_{+,z},\phi_{-,z})\mathcal P(z,k)$. We only need to know its asymptotic form for large $k$ and there it is helpful that the normalized eigenfunctions $(\hat\Psi_{+,z}(k,\cdot),\hat\Psi_{-,z}(k,\cdot))$ form asymptotically an orthonormal basis. Then we obtain expressions that are very similar to those in the proof of Lemma~\ref{lem-Dirac} and the same arguments show that $(\phi_{+,z},\phi_{-,z})\mathcal P(z,k)$ has asymptotic form $f(k) C(z)$. Thus 
the criterion for resolvent affiliation of $\hat H_{K,L,M}$ is  that $U_{K,L,M}(k) \longrightarrow 1$ as $k$ tends to infinity.

While the diagonal elements of $\Gamma_1 \sigma_z \Gamma_1^{-1} $ are bounded in $k$, the off-diagonal ones are at most of order $k$.  As $\mu_+-\mu_-$ is of order $O(k^{-1})$ we see that
$$ \Gamma'_1  \Gamma_1^{-1} = - |k| 1_2 + O(1).$$
Hence $W(k,z) = - \imath k L -|k|M + O(1)$. This shows 
that under the assumption of the lemma $U_{K,L,M}(k)$
tends to $1$ as $|k|\to \infty$.
\end{proof}
The condition that $\det(iL\pm M)\neq 0$ is sufficient but not necessary for resolvent affiliation. For instance, if $L=M=0$ then $K$ must be invertible in which case we have Dirichlet boundary conditions. As already seen $\hat H_{1,0,0}$ is resolvent affiliated. Another interesting class of boundary conditions is given by
\begin{equation}\label{eq:reg_dirac_bc}
\Psi_1(x,0) = 0,\quad 
ia\partial_{x}\Psi_2(x,0) = \partial_{y} \Psi_2(x,0)
\end{equation}
depending on a real parameter $a$. 
They correspond to
$$K-\imath kL = \begin{pmatrix} 1 & 0 \\ 0 & -ak \end{pmatrix}, \quad 
M = \begin{pmatrix} 0 & 0 \\ 0 & 1\end{pmatrix} $$
and we denote $\hat H_{K,L,M}$ now shorter by $\hat H_a$.
It follows that
$$\det W(k,z) = -ak - \frac{ \mu_- \tilde \nu_- \tilde \nu_+ -\mu_+ \tilde \mu_+ \tilde \mu_- }{\tilde \nu_+  \tilde \nu_- - \tilde \mu_+  \tilde \mu_-}$$
\begin{lemma}
The extension $\hat H_a$ is resolvent affiliated $M_2(\hsalgebra)$  if and only if $a\neq \pm 1$.
\end{lemma}
\begin{proof}
As the $11$-component of 
$W(k,\imath)^{-1} W(k,-\imath)$ is $1$ the whole matrix tends to $1$ if and only if its determinant tends to $1$ (the matrix is unitary). 
We can write 
$$\det W(k,z) = -ak - |k| + G(k,z) $$
where $G(k,z)$ is of order $O(k^{-1})$.
Therefore the condition $a\neq \pm 1$ is sufficient for resolvent affiliation. 

Furthermore, $G(k,\imath)^{-1} G(k,-\imath)$ does not tend to $1$ if $k$ goes to $+\infty$ or $-\infty$. Hence if $|a|=1$ then
 $W(k,\imath)^{-1} W(k,-\imath)$ does not tend to $1$ as $-ak$ goes to 
 $+\infty$.  
\end{proof}

\subsubsection{Edge modes}
The dispersion relation of the edge modes for energies outside the spectrum of $\hat H_{1,0,0}$ are determined by the relation $W(k,\lambda) = 0$. This equation is analytically difficult to solve and so we content ourself here to some special solutions which are mostly numerically determined.

We start with flat band modes. There are two of them. Indeed
$$\Psi_{-m} = \begin{pmatrix}
0 \\ |k|+k  \end{pmatrix} e^{-|k|y},\quad \Psi_{m} = \begin{pmatrix}
|k|-k \\ 0 \end{pmatrix} e^{-|k|y}$$
are solutions to the eigenvalue equation $\mathring H(k) \Psi = \lambda \Psi$ with energy $\lambda$ which does not depend on $k$. In the first case, $\lambda = -m$ and $k$ must be positive, while in the second $\lambda = m$ and $k$ must be negative. Note that the first mode satisfies Dirichlet boundary conditions in the first component and $(k-\partial_y)\Psi_2(0)= 0$. So it is an edge mode of $\hat H_{a=1}$. This explains why $\hat H_{a=1}$ cannot be resolvent affiliated. The failure of $\hat H_{a=-1}$ to be resolvent affiliated is more subtle, as it has an edge mode with a dispersion relation which is only asymptotically flat. Note that 
the second solution above does not correspond to a mode of $\hat H_a$ for some $a$ as it 
satisfies Dirichlet boundary conditions in the second component and $(k+\partial_y)\Psi_1(0)= 0$.

In general, the dispersion relation for edge modes is given implicitly by the equation $W(k,\lambda(k))=0$ with real $\lambda(k)$ outside of the spectrum of $\hat H_{1,0,0}(k)$. For the operator $\hat H_{1,0,0}$ it needs to be solved directly, that is, through the eigenvalue equation $H^*\Psi = \lambda\Psi$ with eigenfunctions $\Psi$ decaying at $y=+\infty$ and satisfying the boundary conditions $\Psi_1(k,0)=\Psi_2(k,0) = 0$. 
The solutions (\ref{eq-base}) remain valid for $\lambda$ outside of the spectrum of 
$\mathring H(k)$ and we can express
$\Psi$ in the basis , 
$\Psi = c_+\Psi_{+,\lambda} +  c_-\Psi_{-,\lambda}$ to see that $\Psi$ satisfies Dirichelt boundary conditions if and only if 
$$c_+\phi_{+,\lambda} +  c_-\phi_{-,\lambda}=0.$$
This means that the exterior product $\phi_{+,\lambda} \wedge \phi_{-,\lambda}$ must vanish, leading to the relation
$(\nu_++\lambda)(\nu_--\lambda) = (\mu_--k)(\mu_++k)$
which is equivalent to
\begin{equation}\label{eq-disp-D}
m+\lambda =- \epsilon(\mu_++k)(\mu_-+k).
\end{equation}
We have to keep in mind that $\mu_\al$ and $\nu_\al$ depend on 
$\lambda$ through $\zeta_\al$. Note that $\mu_+(k,\lambda)\mu_-(k,\lambda)\geq \frac{|m|}{|\epsilon|}+2|k|^2$ if $\lambda=0$ with equality if also $k=0$.

Therefore, if $m$ and $\epsilon$ have the same sign then $\lambda$ cannot be $0$ and we do not have any spectral flow through $0$. 

Suppose that  $m$ and $\epsilon$ have opposite sign. Then $(\lambda,k)=(0,0)$ is a solution of (\ref{eq-disp-D}) and moreover, it is the only solution when $\lambda=0$. By deriving this relation w.r.t.\ $k$ we see that 
the curve $\lambda(k)$ associated to this solution satisfies $-\epsilon \lambda'(0)>0$. We conclude that the spectral flow is equal to minus the sign of $\epsilon$. 

We thus see that the simple bulk edge correspondence (Hatsugai's relation) $\sigma_b=\sigma_e$ holds in both cases, when the signs of $m$ and $\epsilon$ are equal and when they are opposite. Note that the relation itself is different in the two cases.

For the family of boundary conditions \eqref{eq:reg_dirac_bc} the dispersion relation is given by $\det W(k,\lambda) = 0$ which amounts to
\begin{equation}\label{eq-dispertion}
(\nu_++\lambda)(\mu_-+k)(\mu_--ak)= 
(\nu_-+\lambda)(\mu_++k)(\mu_+-ak) 
\end{equation}
If $a=1$ then this equation simplifies to 
$(\nu_++\lambda)\zeta_-=(\nu_-+\lambda)\zeta_+$
which has solution $\lambda = -m$. However, $-m$ lies in the spectrum of $\hat H_{1,0,0}$ so that the formula  $\det W(k,\lambda) = 0$ is no longer justified. Yet our explicit calculation above shows, that there is an edge mode with $\lambda = -m$, but only for $k>0$. 

A full-fledged analysis of the possible cases is beyond the scope of this paper. Figure~\ref{fig-1} shows the numerical solution to this relation for various values of $a$. It allows to read off the spectral flows which are summarized in Table~~\ref{fig:specflow_regdirac}.
 \begin{figure}\label{fig-1}
\centering
\includegraphics[width=1\textwidth]{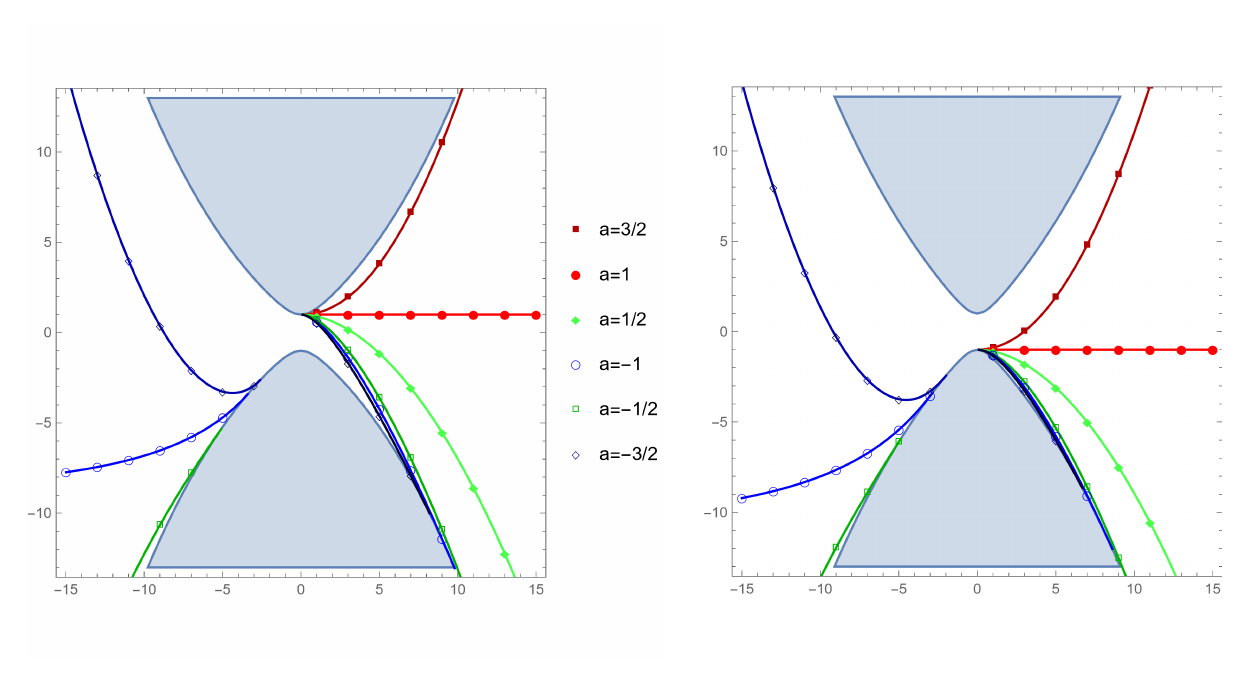}
\caption{Numerical solutions of \eqref{eq-dispertion}
for representative values of $a$, $\epsilon=0.1$ and $m=-1$ (left) and $m=1$ (right).}
\end{figure}
\begin{table}
	\label{fig:specflow_regdirac}
	\begin{tabular}{ r | c  c }
		& $m=-1$ & $m=+1$ \\
		\hline	
		$(K,L,M)=(1,0,0)$ & $-1$ & 0 \\
		$a > 1$ & $-2$ & $-1$ \\
		$-1< a < 1$ & $-1$ & $0$ \\
		$a < -1$ & $0$ & $1$ \\
		\hline
		$\sigma_b$ & $-1$ & 0
	\end{tabular}
	\medskip
	\caption{Spectral flows for the regularized Dirac Hamiltonian on half-space with Dirichlet and with boundary conditions \eqref{eq:reg_dirac_bc}. The last line is the bulk Chern number (\ref{eq-bulk-ch}). Here $\epsilon = 0.1$.}
\end{table}

Those results are in perfect agreement with Theorem~\ref{th:bulk_edge_pluscorrection}, namely the difference of first with the second column in the table is the difference of the bulk Chern numbers. We can see that in the space of resolvent-affiliated boundary conditions there are at least three connected components with non-trivial relative winding numbers. Note in particular that the Hatsugai relation holds for the Dirichlet-Dirichlet boundary condition $(K,L,M)=(1,0,0)$ as well as the connected component $-1 <a < 1$ which contains the Dirichlet-Neumann boundary condition (as the case $a=0$), since we can assert strong affiliation by Proposition~\ref{prop:dd_and_dn_affiliation} and thus the standard bulk-edge correspondence by Theorem~\ref{th:bbc_strongly_affiliated}.

Jud and Tauber \cite{JudTauber2025} have studied the boundary conditions (\ref{eq-bc}) in full detail, also with the aim to see when the Hatsugai relation fails. However, their definition of edge invariant is different from ours. Instead of considering the spectral flow of the eigenvalues across $0$ energy, they count the number of eigenvalues which emanate from the upper bulk spectrum minus the number of eigenvalues which get absorbed into the upper bulk spectrum when varying $k$. This number $n_b$ need not to coincide with our spectral flow. It is then compared to the Chern number of the upper bulk band and their naive expectation of bulk edge correspondence would be that these two numbers coincide. What they find is that there is a quantity $w_\infty$, called the anomaly at infinite energy, which has to be added to their edge invariant $n_b$ to correct their version of bulk edge correspondence. The anomaly at infinity is derived from a relative version of Levinson's theorem. Both, for the definition of their edge invariant and for the anomaly at infinity, the existence of a Brillouin zone (translational symmetry) is crucial. This is rather different from our methods which allow also for a definition of the corrected edge invariant for systems without translation invariance. 

\subsection{Shallow water model}
\label{sec-shallow}
We consider an effective model for shallow water waves introduced in \cite{DelplaceEtAl} to derive a topological explanation of equatorial waves. It has subsequently been under intense study since it provides ample examples of (apparent) violations of bulk-edge correspondence \cite{
TDV19,TDV20,GJT21,TT21,BalYu,GrafTarantola2025}.

In the Hamiltonian formulation the model is described by
the matrix valued Hamiltonian
\begin{equation}
	\label{eq:shw_bulk}
	H_{f,\nu}=\begin{pmatrix}
		0 & p_1 & p_2 \\ 
		p_1 & 0 & \imath(f - \nu p^2)\\
		p_2 & -\imath(f - \nu p^2) & 0
	\end{pmatrix}
\end{equation}
with $f,\nu \in \RM$ the Coriolis parameter and odd viscosity respectively. One has a spectral gap $\sigma(H_{m,\epsilon})=(-\infty, -\abs{f}]\cup [\abs{f},\infty)$. The term proportional to $\nu$ has a similar purpose as the second order term in the regularized Dirac operator: For $\nu\neq 0$ the Hamiltonian is strongly affiliated as can be checked by explicit computation in the Fourier representation. While the symbol is not elliptic one still has $M_3(\Aa_b^\sim)$-resolvent-affiliation for $\nu\neq 0$, namely
\begin{equation}
\label{eq:sh_res_aff}(H_{f,\nu}+\imath)^{-1} \in 
\begin{pmatrix} 
	-\imath & 0 & 0 \\
	0 & 0 & 0\\
	0 & 0 & 0 \end{pmatrix} + M_3(C_0(\RM^2))
\end{equation}
which has a non-vanishing scalar part. This is no longer the case if $\nu=0$, which makes perturbation theory more difficult. Nevertheless, one can check by explicitly computing the resolvents that
$$(H_{f,\nu}+\imath)^{-1}-(H_{f',\nu}+\imath)^{-1} \in M_3(C_0(\RM^2)),$$
for all $f,f',\nu \in \RM$. In particular Hamiltonians with different Coriolis parameter are resolvent-comparable even if $\nu=0$.

For $\nu\neq 0$ the bounded transform of $H_{f,\nu}$ belongs to $M_3(\Aa_b^\sim)$, that is, strong affiliation holds and one can compute 
$$\langle \Ch_{2}, [P_{<0}(H_{f,\nu})]_0\rangle = -\sgn(f)-\sgn(\nu).$$
For the relative Chern number one finds
$$\langle \Ch_{2}, [P_{<0}(H_{f_1,\nu}),P_{<0}(H_{f_2,\nu})]_0\rangle = \sgn(f_2)-\sgn(f_1)$$
which also holds for $\nu=0$.

For $\nu\neq 0$ one can find boundary conditions which lead to a self-adjoint operator which is $M_3(\hsalgebra^\sim)$-resolvent-affiliated: If we restrict $H_{f,\nu}$ to the positive half-space with Dirichlet-Dirichlet boundary conditions
$$\Psi_1\rvert_{x_2 = 0}=0, \quad \Psi_2\rvert_{x_2 = 0}=0$$
we get a self-adjoint half-space operator $\hat{H}_{f,\nu,DD}$. Viewing this operator as a perturbation of its second-order part one can obtain a norm-convergent series expansion for its resolvent from which one can read off that $\hat{H}_{f,\nu,DD}$ is $M_3(\hsalgebra^\sim)$-resolvent-affiliated with the same non-vanishing scalar part as in \eqref{eq:sh_res_aff}. However, $\hat{H}_{f,\nu,DD}$ cannot be strongly affiliated as it turns out that Hatsugai's relation is not satisfied with these boundary conditions. 

By a simpler perturbation argument one may then conclude that $\hat{H}_{f_1,\nu,DD}$ and $\hat{H}_{f_2,\nu,DD}$ for $f_i \neq 0$ are $M_3(\hsalgebra)$-comparable, hence one can apply the relative bulk-edge correspondence of Theorem~\ref{th:rel_bbc} to conclude
$$\langle \Ch_{1}, [\hat{H}_{f_1,\nu,DD}-\epsilon,\hat{H}_{f_2,\nu,DD}-\epsilon]_1\rangle  = \sgn(f_2)-\sgn(f_1)$$
Here we shift the half-space Hamiltonian by some small $\epsilon>0$ since there is bulk spectrum at energy $0$; the edge states we are interested in lie in the bulk gap above that.

For $\nu\neq 0$ the same right-hand side gives the spectral flow 
$$\langle \Ch_{1}, [\hat{H}_I-\epsilon]_1\rangle  = \sgn(f_2)-\sgn(f_1)$$
in an interface model
$$\hat{H}_I=\begin{pmatrix}
	0 & p_1 & p_2 \\ 
	p_1 & 0 & \imath(f(X_2) - \nu p^2)\\
	p_2 & -\imath(f(X_2) - \nu p^2) & 0
\end{pmatrix}
$$
where the Coriolis parameter is now a function which interpolates between $f_1$ and $f_2$. This can be derived by arguing exactly as in Theorem~\ref{th:bulk_interface_pluscorrection} since the Coriolis term is orthogonal to the scalar part of the bulk resolvent, hence $\hat{H}_I$ is $M_3(\interfacealgebra)$-comparable to $H_{f,\nu}$. As seen in \cite{BalYu} the number of interface modes for $\nu=0$ can be different if $f$ has jump discontinuities, thereby violating the bulk-interface correspondence. It is not clear at this point whether those violations also have an interpretation in terms of K-theory.

\appendix

\section{Numerical pairings of \texorpdfstring{$K$}{K}-theory classes}
\label{sec:numerical_pairing}
It is well-known that $K$-theory classes can be paired with cyclic cocycles to obtain numerical invariants. In the applications  we consider, the pairings may be interpreted as linear response coefficients such as conductivities, the bulk boundary correspondence thus relating bulk with boundary conductivities. We use here the formulation of cyclic cocycles as characters of cycles.  
\begin{definition}
	An $n$-cycle $(\Omega, d, \varphi)$ over an algebra $A$ is given by 
	\begin{itemize}
		\item[{\rm (i)}] A differential graded algebra (DGA) $(\Omega,d)$. This is a graded algebra $\Omega = \bigoplus_{j\geq 0} \Omega^j$ with $\Omega^{j_1} \Omega^{j_2} \subset \Omega^{j_1+j_2}$ and a graded differential $d: \Omega \to \Omega$, i.e.\ a linear map with  $d\Omega^j \subset \Omega^{j+1}$, $d^2=0$ and $d(ab)= (da)b + (-1)^{\mathrm{deg}(a)} a (db)$ for homogeneous elements.
		
		\item[{\rm (ii)}] A homomorphism $\rho: A \to \Omega^0$. 
		
		\item[{\rm (iii)}] A closed graded trace $\varphi: \Omega \to \CM$ of top degree $n$, which means $\varphi(ab)=(-1)^{\mathrm{deg}(a) \mathrm{deg}(b)}\varphi(ba)$, $\varphi(da)=0$ and finally $\varphi(a)=0$ for $\mathrm{deg}(a) > n$.
		
	\end{itemize}
\end{definition}
The character of an 
$n$-cycle over $A$ is 
$$
\tilde{\varphi}(a_0, \ldots ,a_n) \;:=\; \varphi\big(\rho(a_0)d\rho(a_1)\cdots d\rho(a_n)\big)
$$
and one can show that any cyclic $n$-cocycle can be written as the character of some $n$-cycle $(\Omega,d,\varphi)$ \cite{Connes-book,SSt}.

A dense subalgebra $A\subset \Aa$ of a $C^*$-algebra $\Aa$ is called smooth if $A$ is closed under holomorphic functional calculus of $\Aa$. Under this condition any class in $K_i(\Aa)$ can be represented by an element in  $ M_N(A^\sim)$.
The pairing of the character $\tilde\varphi$ of an $n$-cycle with an element of $K_i(\Aa)$ is defined with the help of such a representative:
\begin{itemize}
\item
If $n=2k$ is even and $e \in M_N(A^\sim)$ a projection defining the class $[e]_0-[s(e)]_0\in K_0(\Aa)$ then 
$$\langle \tilde \varphi,[e]_0-[s(e)]_0\rangle  =
\left(\frac{1}{2\pi\imath}\right)^k \frac{1}{k!}\, \tilde{\varphi}(e-s(e), e-s(e), \ldots , e-s(e)).$$
\item If $n=2k+1$ is odd and $u \in M_N(A^\sim)$ a unitary defining the class $[u]_1\in K_1(\Aa)$ then 
$$\langle \tilde\varphi, [u]_1\rangle = 2^{-(2k+1)} \left(\frac{1}{\pi \imath}\right)^{k+1}\frac{1}{(2k+1)!!} \tilde\varphi(u^*-1, u-1, u^*-1, \ldots , u-1).$$
\end{itemize}
The normalization constants are chosen such that the pairings are real-valued and in certain cases integers.

Let us now consider the specific class of cocycles which play a major role in solid state physics, the so-called Chern-cocycles:

\newcommand{\mean}{\mu}
\begin{example}
	\label{ex:chern_cocycles}
	\rm{ Consider a twisted crossed product 
		algebra $\Aa = \Cc\rtimes_{\alpha,\gamma}\RM^d$ which we discussed in Section~\ref{ssec:algebras}. Mostly $\Cc = C_0(\Omega)$, and $\Omega$ comes with a translation-invariant measure $\mean$. What is needed to define these Chern-cocycles are, a densely defined $\alpha$-invariant trace $\check \Tt$ on $\Cc$, and $d$ densely defined commuting derivations. The trace $\check\Tt$ extends to a lower-semicontinuous trace $\Tt$ on $\Aa$ given on the generators (\ref{eq:linear_span_crossedproduct}) by 
		$$\Tt( c\, (T_1) ... f_d(T_d)) = \check\Tt( c)\, \prod_{i=1}^d \int_{\RM}  f_i(t)dt.$$ 
		In the case $\Cc = C_0(\Omega)$ the trace $\check\Tt$ is given by integration over $\Omega$ w.r.t.\. the measure $\mean$. 
		The derivations are the derivatives of the action which is dual to $\alpha$. More precisely,  
		$$\nabla_i (c f_1(T_1) ... f_d(T_d)) = c f_1(T_1) \cdots f'_i(T_i) \cdots f_d(T_d).$$
		In the important special case where $\Omega=\{*\}$ is just one point and $\alpha$ trivial, one has $\Aa\simeq C_0(\RM^d)$, $\nabla$ is the usual gradient operator
		and $\Tt$ is the integral over $\RM^d$.
		Let $\Lambda \RM^d $ be the exterior algebra over $\RM^d$, with generators  $e_1,...,e_d$ (of degree one). 
		The algebra $A\otimes \Lambda \RM^d$ is a differential graded algebra with $\mathrm{deg}(a\otimes \omega)=\mathrm{deg}(\omega)$ and the graded derivation $$d(a \otimes \omega) = \sum_{i=1}^n \nabla_{i} a \otimes e_{i}\wedge\omega.$$
		The top Chern-cocycle is the character of the $d$-cycle $(A\otimes \Lambda\RM^d,d,\varphi)$ where the graded trace is given by
		\begin{equation}
			\label{eq:ch_trace}
			\varphi_A(a\otimes \omega) = \begin{cases}
				\Tt(a) & \text{if }\omega=e_1\wedge...\wedge e_d\\
				0 & \text{otherwise}.
			\end{cases}
		\end{equation}
		
		As a result, the top Chern-cocycle $\Ch_d=\tilde\varphi_A$ is given by
		$$\Ch_d(a_0,\cdots,a_d) = c_d \, \sum_{\sigma\in S_d} \mathrm{sgn}(\sigma) \Tt\big(a_0 \nabla_{\sigma(1)} (a_{\sigma(1)}) \cdots \nabla_{\sigma(d)} (a_{\sigma(d)})\big)$$
		
		Chern cocycles of lower degree $n$ can be obtained if one restricts the differential graded algebra to a subspace $\RM^n$ of $\RM^d$ 
		and correspondingly the graded trace. In this case the volume element $e_1\wedge...\wedge e_d$ in formula (\ref{eq:ch_trace}) is to be replaced by a choice of volume element in the subspace $\RM^n$.
		
		For the Chern-cocycles to define a pairing with $K_i(\Aa)$ its domain must include a smooth subalgebra $A$ of $\Aa$. In the given setting there are many possible choices, a convenient one is to consider as in \cite{SSt} the elements of $\Aa$ for which all of the seminorms
		$$\norm{a}_{\scrA_m}:= \sum_{\substack{j \subset \NM^d\\ \abs{j}\leq m}} \norm{\nabla^j a}+ \Tt(|\nabla^j a|)$$
		are finite, where  $\nabla^j=\nabla_{e_1}^{j_1}...\nabla_{e_d}^{j_d}$. This gives a dense Fr\'echet subalgebra $A$ which is smooth in $\Aa$. With the respective normalization constants the pairing with the $K$-groups of $\Aa$ is thus given by as follows:
		If $d$ is even and  $e \in M_N(A^\sim)$ a projection defining a $K_0$-class then
		\begin{equation}\label{eq:even-pairing}
			\langle \Ch_d,[e]_0-[s(e)]_0\rangle  =c_d\,\sum_{\sigma\in S_d} \mathrm{sgn}(\sigma)
			\Tt \big((e-s(e))\nabla_{\sigma(1)} (e) \cdots \nabla_{\sigma(d)} (e)\big)
		\end{equation}
		while, if $d$ is odd and $u \in M_N(A^\sim)$ a unitary defining a $K_1$ class
		\begin{equation}\label{eq:odd-pairing}
			\langle \Ch_d,[u]_1\rangle  =c_d\,\sum_{\sigma\in S_d} \mathrm{sgn}(\sigma)
			\Tt \big((u^*-1)\nabla_{\sigma(1)} (u) \nabla_{\sigma(2)} (u^*)\cdots \nabla_{\sigma(d)} (u)\big)
		\end{equation}
		
	}$\diamond$
\end{example}

Above we have used the standard picture to represent $K$-theory elements. As seen in Section~\ref{sec:ktheory} have seen that otherwise there is a way to represent $K$-theory elements by means of pairs of Hamiltonians, defining elements of 
the pair algebra 
$\PM(\Bb,\Aa)$ where $\Bb$ is a unital $C^*$-algebra which contains $\Aa$ as an ideal. To adapt the pairing with an existing cocycle over $\Aa$ to this case we suppose that 
$\Bb$ contains a smooth unital subalgebra $B$, and that $B$ contains a smooth subalgebra $A$ of $\Aa$ as an ideal. Let then $(\Omega_B,d)$ be a differential graded algebra over $B$ and consider the ideal
$$\Omega_A := \Omega_B (\rho(A)+d\rho(A))\Omega_B$$
which is a differential graded algebra over $A$. 
Define the graded algebra $\Omega=\bigoplus_{j\geq 0} \Omega^j$ with 
	$$\Omega^{j} := \{(\omega_1,\omega_2)\in \Omega^j_B\oplus \Omega^{j}_B:\, \omega_1-\omega_2\in \Omega^{j}_A\}$$
which we equip with the differential $d(\omega_1,\omega_2) = (d\omega_1,d\omega_2)$. This then is a differential graded algebra over 
$$\PM(B,A)=\{(b_1,b_2)\in B\oplus B:\, b_1-b_2\in A\}.$$ Finally let $\varphi_A:\Omega_A^n\to\CM$ be a linear map which is graded cyclic in the stronger sense 
	\begin{equation}
		\label{eq:gradedtrace_extension}
		\varphi_A(\omega_B \omega_A)= (-1)^{\mathrm{deg}(\omega_A)\mathrm{deg}(\omega_B)}\varphi_A(\omega_A \omega_B)
	\end{equation}
where $\omega_A\in \Omega_A^k$ and $\omega_B \in \Omega_B^{n-k}$.
\begin{proposition}
\label{prop:extending_n_cycles}
In the above setting $$\varphi_{\PM(B,A)}((\omega_1,\omega_2)):= \varphi_A(\omega_1-\omega_2)$$
defines a graded trace on the differential graded algebra $(\Omega,d)$ over 
$\PM(B,A)$.
\end{proposition}
\begin{proof} Let $\omega=(\omega_1,\omega_2)\in \Omega^k$, $\tilde{\omega}=(\tilde{\omega}_1,\tilde{\omega}_2)\in \Omega^{n-k}$ then
\begin{align*}
	\varphi_{\PM(B,A)}(\omega,\tilde{\omega}) &= \varphi_A(\omega_1 \tilde{\omega}_1- \omega_2\tilde{\omega}_2)\\&= \varphi_A(\omega_1 (\tilde{\omega}_1-\tilde{\omega}_2)+(\omega_1-\omega_2) \tilde{\omega}_2)\\&= (-1)^{\mathrm{deg}(\omega)\mathrm{deg}(\tilde{\omega})}\varphi_A((\tilde{\omega}_1-\tilde{\omega}_2)\omega_1 +\tilde{\omega}_2(\omega_1-\omega_2) )\\
	&=(-1)^{\mathrm{deg}(\omega)\mathrm{deg}(\tilde{\omega})}\varphi_{\PM(B,A)}(\tilde{\omega},\omega).
\end{align*}
the third equation following from 
(\ref{eq:gradedtrace_extension}).
\end{proof}
As $A$ embeds into $\PM(B,A)$,
$A\stackrel{i}\hookrightarrow \PM(B,A)$, $i(a) = (a,0)$,
the above $n$-cycle can be considered as an extension of the $n$-cycle $(\Omega_A,d,{\varphi_A})$ over $A$. This extension is exactly what one needs to compute the pairings with K-theory under the isomorphism \eqref{eq:K-alt}:

\begin{corollary}
Assume that $A$ is smooth in the $C^*$-algebra $\Aa$ and $\PM(B,A)$ is smooth in $\PM(\Bb,\Aa)$. The pairings with the characters $\tilde \varphi_A, \tilde \varphi_{\PM(B,A)}$ defines homorphisms $\langle\tilde \varphi_A,\cdot\rangle: K_i(\Aa) \to \CM$ and $\langle\tilde \varphi_{\PM(B,A)},\cdot\rangle: K_i(\PM(\Bb,\Aa)) \to \CM$.

For the injection $i_*: K_i(\Aa)\to K_i(\PM(\Bb,\Aa))$ as in \eqref{eq:K-pair-splitting} one has
$$ \langle\tilde \varphi_A,\cdot\rangle = \langle\tilde \varphi_{\PM(B,A)},i_*(\cdot)\rangle$$
as maps on $K_i(\Aa)$. Moreover,
$\langle\tilde \varphi_{\PM(B,A)},{t_*(K_i(\Bb))}\rangle = 0$
and hence $\langle\tilde \varphi_{\PM(B,A)},\cdot\rangle$ descends to a homomorphism 
$$\langle\tilde \varphi_{\PM(B,A)},\cdot\rangle:K_i(\PM(B,A))/t_*(K_i(\Bb))\to \CM$$
which is equivalent to $\langle\tilde \varphi_A,\cdot\rangle$ under the isomorphism \eqref{eq:K-alt}.
\end{corollary}
\begin{proof}
The first statement is clear since $\tilde{\varphi}_{\PM(B,A)}$ is an extension of the original cocycle $\tilde{\varphi}$ if one considers $\Omega_A$ as a subalgebra of $\Omega_{\PM(B,A)}$.

For the remaining statements, note that every class $x\in t_*(K_i(\Bb))\subset K_i(\PM(\Bb,\Aa))$ is represented by a pair $(b,b)$ with a representative in $b \in M_n(\Bb)$, hence it is easy to see that $\langle \varphi_{\PM(B,A)},x\rangle=0$.
\end{proof}

\begin{example}
\rm{
We can apply this prescription to extend the Chern cocycles in the setting of Example~\ref{ex:chern_cocycles} to a pair algebra $\PM(B,A)$. 
Here we take for $B$ all elements of the multiplier algebra $\mult(\Aa)$ which are infinitely often norm-differentiable w.r.t.\ to $\nabla_1,...,\nabla_d$.
As the derivations $\nabla_i$ are well defined on $B$ we obtain a differential graded algebra $(\Omega_B,d)$ with
$\Omega_B= B \otimes \Lambda\RM^d$ and with formally the same differential $d$. The functional $\varphi_A$ from (\ref{eq:ch_trace}) satisfies \eqref{eq:gradedtrace_extension}, since the trace property
$\Tt(b a) = \Tt(a b)$ still holds for all $b\in B$ if $a\in A$. We thus obtain a $d$-cycle $(\Omega,d,\varphi_{\PM(B,A)})$ over $\PM(B,A)$.
One can show that $\PM(B,A)$ is smooth in $\PM(\Bb,\Aa)$ where $\Bb =\overline{B}$ is the $C^*$-closure in $\mult(\Aa)$.

The pairing with the $K$-groups of $\Aa$ can thus be formulated as follows:
If $d$ is even and  $p,q \in M_N(B)$ projections such that $p-q\in M_N(A)$  then 
\begin{equation}
\label{eq:chern_difference}
\langle \Ch_d,[p,q]_0\rangle  =c_d\,\sum_{\sigma\in S_d}(-1)^\sigma
\Tt \big(p\nabla_{\sigma(1)} (p) \cdots \nabla_{\sigma(d)} (p)-q\nabla_{\sigma(1)} (q) \cdots \nabla_{\sigma(d)} (q)
\big)
\end{equation}
There is a similar formula if $d$ is odd and  $u,v \in M_N(B)$ unitaries such that $u-v\in M_N(A)$ which we do not write out, as it is more efficient to use $[u,v]_1=[uv^*,1]_1$ and to employ (\ref{eq:odd-pairing}) with $uv^*$ in place of $u$.
}$\diamond$
\end{example}

For the special case $\alpha=\mathrm{id}$ of a trivial action and trivial twist $\gamma$ there is also a more geometric way extend the cocycles:

\begin{example}
\label{ex:sphere_cocycle}
{\rm We consider the algebra $\Aa=C_0(\RM^d)$ with the dense subalgebra $A=\Ss(\RM^d)$, the Schwartz functions, so that 
$A\otimes\Lambda\RM^d$ is a dense subalgebra of the exterior algebra over $\RM^d$. The top Chern cocycle is an integral over a differential form
$$\Ch_d(a_0,\cdots,a_d) = \int_{\RM^d} a_0 da_1\cdots da_d $$
where $d$ is the exterior derivative. 

Consider $B=C^\infty_{vg}(\RM^d)$ the smooth functions whose derivatives lie in $C^\infty_0(\RM^d)$ (vg stands for vanishing gradient at infinity). Choosing a homeomorphism $\sigma: \RM^d \to H\SM^d$ with the upper half-sphere $H\SM^d=\SM^d\cap (\RM^d\times \RM_+)$ provides an isomorphism $C_0(H\SM^d)\simeq C_0(\RM^d)$. The $C^*$-closure $\Bb=\overline{B}$ corresponds to $C(H\SM^d)$ since $\Bb$ is exactly the algebra of continuous functions on $\RM^d$ which admit continuous radial limits. There is then an isomorphism $\PM(\Bb,\Aa)\simeq C(\SM^d)$ given by mapping the second component of the pair onto the lower half-sphere, which glues continuously since the image of $\Aa$ vanishes at the equator. If $\sigma$ is a diffeomorphism the DGA's $\Omega_A$ and $\Omega_B$ can be pushed forward to differential forms on $H\SM^d$ and the extended cocycle can be seen as
\begin{align*}(\Ch_v)^{\PM(B,A)}((\omega_1,\omega_2))&=\int_{H\SM^d} (\sigma_*(\omega_1)-\sigma_*(\omega_2))=\int_{\SM^d} \hat{\omega}
\end{align*}
where $\hat{\omega}$ is the form on $\SM^d$ obtained by gluing $\sigma_*(\omega_1)$ together with the reflection of $\sigma_*(\omega_2)$ onto the lower half-sphere. In this way the Chern cocycle associated to $\PM(B,A)$ can be thought of as the usual Chern cocycle on $\SM^d$.}$\diamond$
\end{example}

\end{document}